\documentclass[acmsmall]{acmart}
\acmJournal{PACMPL}
\acmVolume{1}
\acmNumber{CONF} % CONF = POPL or ICFP or OOPSLA
\acmArticle{1}
\acmYear{2018}
\acmMonth{1}
\acmDOI{} % \acmDOI{10.1145/nnnnnnn.nnnnnnn}
\startPage{1}

\setcopyright{none}
\usepackage{booktabs}   %% For formal tables:
\usepackage{subcaption} %% For complex figures with subfigures/subcaptions
\usepackage{stmaryrd} 
\usepackage{proof}
\usepackage{amssymb}
\usepackage{amsmath}
\usepackage{mathtools}
\usepackage{amsthm}
\usepackage{mathpartir}
\usepackage{color}
\usepackage{xstring}
\usepackage{ntabbing}
\usepackage{listings}
\usepackage{varwidth}
\usepackage{enumitem}

\newlength{\rWidth}

\newcommand{\funtype}[1]{%
    {\settowidth{\rWidth}{\ensuremath{#1}}%
        \;\ensuremath{{\xrightarrow{\hspace{\rWidth}}\hspace{-0.84\rWidth}}\!\!\!^%
         {#1}%{\BehindSubString{,}{#1} / \BeforeSubString{,}{#1}}%
         \hspace{0.2\rWidth}\;\;}}}

\newcommand{\m}[1]{\mathsf{#1}}
\newcommand{\mb}[1]{\mathbf{#1}}

\newcommand{\esync}{\; \m{esync}}

\newcommand{\s}{\m{S}}
\renewcommand{\l}{\m{L}}
\renewcommand{\c}{\m{T}}
\newcommand{\p}{\m{R}}

\newcommand{\lang}[1]{\mathbf{L}(#1)}

\newcommand{\W}{\Omega}
\newcommand{\Sg}{\Sigma}
\newcommand{\xvdash}[1]{%
  \vdash^{\mkern-8mu\scriptstyle\rule[-.9ex]{0pt}{0pt}#1}%
}
\newcommand{\xVdash}[1]{%
  \Vdash^{\mkern-8mu\scriptstyle\rule[-.9ex]{0pt}{0pt}#1}%
}

\newcommand{\potconf}[1]{\overset{#1}{\vDash}}

\newcommand{\D}{\Delta}
\renewcommand{\G}{\Gamma}

\newcommand{\plin}{\; \m{purelin}}

\newcommand{\proc}[2]{\m{proc}(#1, #2)}
\newcommand{\msg}[2]{\m{msg}(#1, #2)}

\newcommand{\step}{\mapsto}

\newcommand{\fresh}[1]{(#1 \text{ fresh})}

\newcommand{\val}{\; \m{val}}

\newcommand{\lam}[3]{\lambda #1 : #2 . M_x}
\newcommand{\inl}[1]{l \cdot #1}
\newcommand{\inr}[1]{r \cdot #1}
\newcommand{\case}[3]{\m{case} \; #1 \; (l \hookrightarrow #2, r \hookrightarrow #3)}
\newcommand{\pair}[2]{\left\langle #1, #2 \right\rangle}
\newcommand{\projl}[1]{#1 \cdot l}
\newcommand{\projr}[1]{#1 \cdot r}
\newcommand{\match}[4]{\m{match} \; #1 \; ([] \rightarrow #2, #3 \rightarrow #4)}
\newcommand{\eproc}[3]{\{#1 \leftarrow #2 \leftarrow #3\}}

\newcommand{\ecase}[3]{\m{case} \; #1 \; (#2 \Rightarrow #3)}

\newcommand{\erecvch}[2]{#2 \leftarrow \m{recv} \; #1}

\newcommand{\esendch}[2]{\m{send} \; #1 \; #2}

\newcommand{\ewait}[1]{\m{wait} \; #1}

\newcommand{\eclose}[1]{\m{close} \; #1}

\newcommand{\fwd}[2]{#1 \leftarrow #2}
\newcommand{\fwdp}[2]{#1 \overset{+}{\leftarrow} #2}
\newcommand{\fwdn}[2]{#1 \overset{-}{\leftarrow} #2}
\newcommand{\esendl}[2]{#1.#2}

\newcommand{\ecut}[4]{#1 \leftarrow #2 \leftarrow #3 \semi #4}

\newcommand{\etick}[1]{\m{tick} \; (#1)}
\newcommand{\ework}[1]{\m{work} \; \{#1\}}
\newcommand{\eget}[2]{\m{get} \; #1 \; \{#2\}}
\newcommand{\epay}[2]{\m{pay} \; #1 \; \{#2\}}
\newcommand{\procdef}[3]{#3 \leftarrow #1 \leftarrow #2}
\newcommand{\procdefna}[2]{#2 \leftarrow #1}
\newcommand{\casedef}[1]{\m{case} \; #1}
\newcommand{\labdef}[1]{#1 \Rightarrow}

\newcommand{\eif}[1]{\m{if} \; (#1)}
\newcommand{\ethen}{\; \m{then} \; }
\newcommand{\eelse}{\m{else} \; }

\newcommand{\lolli}{\multimap}
\newcommand{\tensor}{\otimes}
\newcommand{\with}{\mathbin{\binampersand}}

\newcommand{\one}{\mathbf{1}}

\newcommand{\semi}{\; ; \;}
\newcommand{\ichoiceop}{\oplus}
\newcommand{\echoiceop}{\with}
\newcommand{\ichoice}[1]{\ichoiceop \{ #1 \}}
\newcommand{\echoice}[1]{\echoiceop \{ #1 \}}

\newcommand{\mi}[1]{\mbox{\it #1}}

\newcommand{\arrow}{\to}
\newcommand{\product}{\times}

\newcommand{\tproc}[2]{\{#1 \leftarrow #2\}}

\newcommand{\pot}[2]{#1^{#2}}

\newcommand{\tensorpot}[1]{\overset{#1}{\tensor}}

\newcommand{\entailpot}[1]{\xvdash{#1}}
\newcommand{\exppot}[1]{\xVdash{#1}}
\newcommand{\paypot}{\triangleright}
\newcommand{\getpot}{\triangleleft}
\newcommand{\tgetpot}[2]{\getpot^{#2} #1}
\newcommand{\tpaypot}[2]{\paypot^{#2} #1}
\newcommand{\bigeval}[3]{#1 \Downarrow #2 \mid #3}
\newcommand{\share}{\curlyveedownarrow}

\newcommand{\tlist}[1]{\m{list}_{#1}}
\newcommand{\plist}[2]{L^{#2}(#1)}

\newcommand{\Next}{\raisebox{0.3ex}{$\scriptstyle\bigcirc$}}
\renewcommand{\next}[1]{\Next #1}
\newcommand{\tdelay}[2]{
    \IfEqCase{#2}{%
        {1}{\next{#1}}%
    }[{\Next^{#2} (#1)}]%
}%

\newcommand{\dc}{\mathcal{D}}

\newcommand{\tint}{\m{int}}

\newcommand{\tbool}{\m{bool}}
\newcommand{\lc}{\tlist{\m{coin}}}
\newcommand{\auction}{\m{auction}}

\newcommand{\plcoin}{\m{plcoin}}
\newcommand{\plcoins}{\m{plcoins}}
\newcommand{\plcoinsn}{pl coins}

\newcommand{\B}[1]{\textcolor{blue}{#1}}
\newcommand{\R}[1]{\textcolor{red}{#1}}

\newcommand{\down}{\downarrow^{\m{S}}_{\m{L}}}
\newcommand{\up}{\uparrow^{\m{S}}_{\m{L}}}
\newcommand{\eacquire}[2]{#1 \leftarrow \m{acquire} \; #2}
\newcommand{\eaccept}[2]{#1 \leftarrow \m{accept} \; #2}
\newcommand{\erelease}[2]{#1 \leftarrow \m{release} \; #2}
\newcommand{\edetach}[2]{#1 \leftarrow \m{detach} \; #2}

\newtheorem{theorem}{Theorem}
\newtheorem{definition}{Definition}
\newtheorem{lemma}{Lemma}

\definecolor{verylightgray}{rgb}{.97,.97,.97}

\lstdefinelanguage{Solidity}{
	keywords=[1]{anonymous, assembly, assert, balance, break, call, callcode, case, catch, class, constant, continue, constructor, contract, debugger, default, delegatecall, delete, do, else, emit, event, experimental, export, external, false, finally, for, function, gas, if, implements, import, in, indexed, instanceof, interface, internal, is, length, library, log0, log1, log2, log3, log4, memory, modifier, new, payable, pragma, private, protected, public, pure, push, require, return, returns, revert, selfdestruct, send, solidity, storage, struct, suicide, super, switch, then, this, throw, transfer, true, try, typeof, using, value, view, while, with, addmod, ecrecover, keccak256, mulmod, ripemd160, sha256, sha3}, % generic keywords including crypto operations
	keywordstyle=[1]\color{blue}\bfseries,
	keywords=[2]{address, bool, byte, bytes, bytes1, bytes2, bytes3, bytes4, bytes5, bytes6, bytes7, bytes8, bytes9, bytes10, bytes11, bytes12, bytes13, bytes14, bytes15, bytes16, bytes17, bytes18, bytes19, bytes20, bytes21, bytes22, bytes23, bytes24, bytes25, bytes26, bytes27, bytes28, bytes29, bytes30, bytes31, bytes32, enum, int, int8, int16, int24, int32, int40, int48, int56, int64, int72, int80, int88, int96, int104, int112, int120, int128, int136, int144, int152, int160, int168, int176, int184, int192, int200, int208, int216, int224, int232, int240, int248, int256, mapping, string, uint, uint8, uint16, uint24, uint32, uint40, uint48, uint56, uint64, uint72, uint80, uint88, uint96, uint104, uint112, uint120, uint128, uint136, uint144, uint152, uint160, uint168, uint176, uint184, uint192, uint200, uint208, uint216, uint224, uint232, uint240, uint248, uint256, var, void, ether, finney, szabo, wei, days, hours, minutes, seconds, weeks, years},	% types; money and time units
	keywordstyle=[2]\color{teal}\bfseries,
	keywords=[3]{block, blockhash, coinbase, difficulty, gaslimit, number, timestamp, msg, data, gas, sender, sig, value, now, tx, gasprice, origin},	% environment variables
	keywordstyle=[3]\color{violet}\bfseries,
	identifierstyle=\color{black},
	sensitive=false,
	comment=[l]{//},
	morecomment=[s]{/*}{*/},
	commentstyle=\color{gray}\ttfamily,
	stringstyle=\color{red}\ttfamily,
	morestring=[b]',
	morestring=[b]"
}

\newcommand{\secref}[1]{Section~\ref{#1}}

\begin{document}

%% Title information
\title[Resource-Aware Session Types]{Resource-Aware
Session Types for Digital Contracts}         %% [Short Title] is optional;
                                        %% when present, will be used in
                                        %% header instead of Full Title.
%\titlenote{with title note}             %% \titlenote is optional;
                                        %% can be repeated if necessary;
                                        %% contents suppressed with 'anonymous'
\subtitle{Technical Report}                     %% \subtitle is optional
%\subtitlenote{with subtitle note}       %% \subtitlenote is optional;
                                        %% can be repeated if necessary;
                                        %% contents suppressed with 'anonymous'

%% Author information
\author{Ankush Das}
\affiliation{
  \institution{Carnegie Mellon University}            %% \institution is required
}
\email{ankushd@cs.cmu.edu}          %% \email is recommended

%% Author with single affiliation.
\author{Stephanie Balzer}
\affiliation{
  \institution{Carnegie Mellon University}            %% \institution is required
}
\email{balzers@cs.cmu.edu}          %% \email is recommended

%% Author with single affiliation.
\author{Jan Hoffmann}
\affiliation{
  \institution{Carnegie Mellon University}            %% \institution is required
}
\email{jhoffmann@cmu.edu}          %% \email is recommended

%% Author with single affiliation.
\author{Frank Pfenning}
\affiliation{
  \institution{Carnegie Mellon University}            %% \institution is required
}
\email{fp@cs.cmu.edu}

\author{Ishani Santurkar}
\affiliation{
  \institution{Carnegie Mellon University}            %% \institution is required
}
\email{ivs@andrew.cmu.edu} 

%% Abstract
%% Note: \begin{abstract}...\end{abstract} environment must come
%% before \maketitle command
\begin{abstract}

  Programming digital contracts comes with unique challenges, which include \textit{(i)}
  expressing and enforcing protocols of interaction, \textit{(ii)} controlling resource usage,
  and \textit{(iii)} preventing the duplication or deletion of a contract's assets.
  This article presents the design and type-theoretic foundation of
  \emph{Nomos}, a programming language for digital contracts
  that addresses these challenges.
  To express and enforce protocols, Nomos is based on \emph{shared binary session types}.
  To control resource usage, Nomos employs \emph{automatic amortized resource analysis}.
  To prevent the duplication or deletion of assets, Nomos uses a \emph{linear type system}.
  A monad integrates the effectful session-typed language
  with a general-purpose functional language.
  Nomos' \emph{prototype implementation} features
  \emph{linear-time type checking} and efficient type reconstruction that
  includes automatic \emph{inference of resource bounds} via off-the-shelf linear optimization.
  The effectiveness of the language is evaluated with case studies
  about implementing common smart contracts such as auctions,
  elections, and currencies.
  Nomos is completely formalized, including the type system, a cost
  semantics, and a transactional semantics to instantiate Nomos contracts
  on a blockchain.
  The type soundness proof ensures that protocols are followed at
  run-time and that types establish sound upper bounds on the resource
  consumption, ruling out re-entrancy and out-of-gas vulnerabilities,
  respectively.
\end{abstract}

\lstset{basicstyle=\ttfamily, columns=fullflexible}

%% 2012 ACM Computing Classification System (CSS) concepts
%% Generate at 'http://dl.acm.org/ccs/ccs.cfm'.

%% End of generated code

%% Keywords
%% comma separated list
%\keywords{session types, resource analysis, smart contracts} 
%% \keywords are mandatory in final camera-ready submission

%% \maketitle
%% Note: \maketitle command must come after title commands, author
%% commands, abstract environment, Computing Classification System
%% environment and commands, and keywords command.
\maketitle

\section{Introduction}
\label{sec:intro}

Digital contracts are programs that implement the execution of a contract.
With the rise of blockchains and
cryptocurrencies such as Bitcoin~\cite{satoshi-bitcoin}, 
Ethereum~\cite{ethereum}, and Tezos~\cite{tezos-white}, digital
contracts have become popular in the form of smart contracts, which
provide potentially distrusting parties with programmable money and a
distributed consensus
mechanism.  Smart contracts are used to implement auctions~\cite{Auction-Solidity},
investment instruments~\cite{TheDAO2016}, insurance agreements \cite{initiative2008b3i},
supply chain management~\cite{law2017smart}, and mortgage loans~\cite{morabito2017smart}.
They hold the promise to lower cost, increase fairness, and expand access to the financial
infrastructure.  However, like all software, smart contracts can contain bugs and security
vulnerabilities~\cite{AtzeiBC17}, which have direct financial consequences. A well-known
example is the attack on The DAO~\cite{TheDAO2016}, resulting in a \$60 million dollar
theft by exploiting a contract vulnerability.

Today's contract languages are typically derived from existing general-purpose languages like
JavaScript (Ethereum's Solidity~\cite{Auction-Solidity}), Go (Hyperledger
project~\cite{cachin2016architecture}), or OCaml (Tezos' Liquidity~\cite{Liquidity}), which
fail to accommodate the domain-specific requirements of digital contracts.  These requirements
are: \textit{(i)} expressing and enforcing protocols of interaction, \textit{(ii)} controlling
resource (or gas) usage, and \textit{(iii)} preventing duplication or discard of a contract's
assets.

In this article, we present the type-theoretic foundation of Nomos, a programming
language for digital contracts that accommodates the aforementioned requirements by
construction.

To express and enforce the protocols underlying a contract, Nomos is based on \emph{session
types} \cite{HondaCONCUR1993, HondaESOP1998, HondaPOPL2008}, in particular on the works
rooted in the Curry-Howard correspondence between linear logic and the session-typed
process calculus~\cite{CairesCONCUR2010, ToninhoESOP2013, PfenningFOSSACS2015,
WadlerICFP2012}. Session types capture protocols of interactions in the type, and
type-checking statically guarantees adherence to those protocols at run-time.
Session types make the core functionality of a contract and its intended interactions
with various parties explicit, rather than buried in implementation code. Delimiting the
sequences of actions that must be executed atomically, session types moreover prevent
interception of a contract in an inconsistent state, as is possible through re-entrancy
in some contract languages.

To control resource usage, Nomos employs and further develops \emph{automatic amortized
resource analysis (AARA)}, a type-based technique for automatically inferring symbolic
resource bounds~\cite{Jost03, Jost10, HoffmannAH10, HoffmannW15, CarbonneauxHRS17}.
AARA is parametric in the cost model, making it directly applicable to track gas cost of
Nomos contracts.  A unique feature of Nomos' resource-aware type system is that it allows
contracts to store gas in internal data structures to amortize the cost of resource
intensive transactions. Failure to support estimation of gas usage bares the risk of high
losses in case transactions fail due to dynamic out-of-gas exceptions and
makes contracts vulnerable to denial-of-service attacks.

To prevent duplication or deletion of assets, Nomos uses \emph{linearity}
\cite{GirardTCS1987},
which naturally arises from the Curry-Howard correspondence established between linear logic
and the session-typed process calculus~\cite{CairesCONCUR2010, PfenningFOSSACS2015,
  WadlerICFP2012}.  Accidental or malicious duplication and deletion is a source of major
concern in today's contract languages~\cite{meredith2015linear}.  To support the writing of
general-purpose programs, Nomos moreover complements the session-typed language with a
\emph{functional language}, using a contextual monad to shield expressions from effectful
processes.

We formalize Nomos by giving its type system and operational semantics and by proving type safety.  Integrating the seemingly disparate
approaches (session types, resource analysis, linearity, and functional programming) and
combining them with the different roles that arise in a digital contract (contract, asset,
transaction) in a way that the result remains consistent, presents unique challenges.  For
one, both the functional as well as session-typed language use potential annotations to
predict the resource consumption, which requires care when functional values are exchanged as
messages between processes. For another, prior work on integrating shared and
linear session types~\cite{BalzerICFP2017} preclude contracts from persisting their linear
assets across transactions, a feature essential to digital contract development; a
restriction that we lift in this work. Fundamental is the use of different forms of
\emph{typing judgments} for expressions and processes along with \emph{judgmental modes}
to distinguish the different roles in a digital contract.  The modes are essential in
ensuring type safety, as they allow the expression of mode-indexed invariants on the
typing contexts and their enforcement by the typing rules.

A challenge in Nomos' design was the sound integration of session types,
resource analysis, linearity, and functional programming so that type
checking is linear in the size of the program and resource bounds can
efficiently inferred with an off-the-shelf LP
solver. Efficient type checking is particularly important if
type-checking is part of contract validation and can be used for denial-of-service attacks.

To evaluate Nomos, we implemented a publicly available
open-source prototype~\cite{NomosGithub} and conducted 8 case studies
implementing common smart contracts such as auctions, elections, and
currencies. Our experiments show that type-checking overhead
is less than $0.7$
ms for each contract and bound inference (needed once at
deployment) takes less than $10$ ms. Moreover, gas bounds are
tight for most contracts.

Our contributions are:

\begin{itemize}

  \item design of Nomos, a language that addresses the domain-specific requirements of digital contracts by
  construction;
  
  \item a fine-tuned system of typing judgments (\secref{sec:sharing}) that uses \emph{modes}
    to orchestrate the sound integration of session types (\secref{sec:session}), functions
    (\secref{sec:fun}), and resource analysis (\secref{sec:resource});
  
  \item extension of shared session types to support linear assets;
  
  \item resource cost amortization by allowing gas storage in internal data structures (\secref{sec:resource});
    
  \item type safety proof of Nomos using a novel asynchronous cost semantics (\secref{sec:sound});
  
  \item a prototype implementation and case study of prominent blockchain applications (\secref{sec:study});
  
  \item a transactional semantics to instantiate Nomos contracts and transactions on a blockchain (\secref{sec:model}).
  
  \end{itemize}

In addition, \secref{sec:overview} provides an overview of Nomos' main features based
on an example. \secref{sec:model} discusses known limitations. \secref{sec:related}
reviews related work, and \secref{sec:conclusion} concludes this article with future
directions. The appendix formalizes the complete language with
typing rules, cost semantics, and the type safety theorem and proof. It also shows
the implementation of all the smart contract applications used in the main paper.

%%% Local Variables:
%%% mode: latex
%%% mode: flyspell
%%% TeX-master: "icfp19"
%%% End:

\section{Nomos by Example}
\label{sec:overview}

%\jan{1.5 pages\\
%1) explicit protocols\\
%2) no re-entrance\\
%3) predicting cost\\
%4) linear resources (money)\\
%\\
%- auction example(type and code)\\
%- limitations (deadlocks, no contract-contract communication)
%}
% Nomos is a programming language for writing digital contracts based on
% resource-aware \cite{DasHP17} and shared~\cite{BalzerICFP2017} session types.
This section introduces the main features of Nomos using a simple auction
contract as an example.  The subsequent sections explain each feature in
technical detail.

% This section uses a simple auction contract to showcase the most significant features of the
% language.  The subsequent sections explain each feature in technical detail.

\paragraph{\textbf{\textit{Explicit Protocols of Interaction}}}

Digital contracts, like traditional contracts, follow a \emph{predefined protocol}. For
instance, an auction contract follows the protocol of a bidding phase where bidders submit
their bids to the auctioneer (possibly multiple times), followed by a collection phase
where the highest bidder receives
the lot while all other bidders receive their bids back. In existing smart contract languages,
like Solidity \cite{Auction-Solidity}, the bidding part of the auction is typically implemented
using the $\m{bid}$ function below. This function receives a bid ($\m{msg.value}$) from a
bidder ($\m{msg.sender}$) and adds it to their total previous bids, stored in the
$\m{bidValue}$ hash map.
\begin{lstlisting}[basicstyle=\small\ttfamily, language=Solidity,
    columns=fullflexible]
    function  bid() public payable {
      bidder = msg.sender; bid = msg.value;
      bidValue[bidder] = bidValue[bidder] + bid; }
\end{lstlisting}
The above code does not guarantee that a bid can only be placed in the bidding phase.  To
enforce this constraint, we can introduce a state variable, such as $\m{status}$, to track the
different phases of a contract.  Using this variable we can guard the above code block with the
precondition
\begin{lstlisting}[basicstyle=\small\ttfamily, language=Solidity,
    columns=fullflexible]
    require (status == running)
\end{lstlisting}
checking whether the auction is still running and thus accepts bids.  The precondition is
checked at run-time, aborting the execution if the condition is not met.  It is the
responsibility of the programmer to define state variables, update them,
and corresponding guard functions.

Rather than burying the contract's interaction protocol in implementation code using state
variables and run-time checks, in Nomos, protocols can be expressed explicitly using a
\emph{session type}.  Type-checking then makes sure that the program implements the protocol
defined by the session type.  The auction's protocol amounts to the below session
type:
\begin{tabbing}
$\auction = $
$\B{\up} \textcolor{red}{\getpot^{22}} \ichoice{$\=$\mb{running} : \echoice{$\=$
\mb{bid} : \m{id} \arrow \m{money} \lolli \B{\down} \auction,$
\hspace{2.2em}\= \% recv bid from client\\
\>\>$\mb{cancel} : \textcolor{red}{\paypot^{21}} \B{\down} \auction},$
\> \% client canceled \\
\>$\mb{ended} : \echoice{$\=$\mb{collect} : \m{id} \arrow$
\=$ \ichoice{$\=$\mb{won} :
\m{lot} \tensor \B{\down} \auction,$
\hspace{3em}\= \% client won\\
\>\>\>\>$\mb{lost} : \m{money} \tensor \R{\paypot^{7}} \B{\down} \auction},$
\> \% client lost \\
\>\>$\mb{cancel} : \textcolor{red}{\paypot^{21}} \B{\down} \auction}}$
\hspace{7.6em} \% client canceled
\end{tabbing}
We first focus on how the session type defines the main interactions of a contract with a
bidder, ignoring the operators $\up$, $\down$, $\getpot$, and $\paypot$ for now.  To
distinguish the main two states an auction can be in, the session type uses the internal choice
type constructor ($\oplus$), leading the contract to either send the label $\m{running}$ or
$\m{ended}$, depending on whether the auction still accepts bids or not.  Dual to an internal
choice is an external choice ($\with$), which leaves the choice to the client (i.e., bidder)
rather than the provider (i.e., contract).  For example, in case the auction is still running,
the client can choose between placing a bid (label $\m{bid}$) or backing out ($\m{cancel}$). If
the client chooses to place a bid, they have to indicate their identifier (type $\m{id}$),
followed by a payment (type $\m{money}$). Nomos session types allow transfer of both non-linear
values that can be duplicated or discarded (e.g. $\m{id}$), using the arrow ($\arrow$)
constructor, and linear
assets, using the linear implication ($\lolli$) constructor.  Using a linear type to
represent digital money
($\m{money}$) makes sure that such a value can neither be duplicated nor lost.  Should the
auction have ended, the client can choose to check their outcome (label $\m{collect}$) or back
out ($\m{cancel}$). In the case of $\m{collect}$, the auction will answer with either $\m{won}$
or $\m{lost}$.  In the former case, the auction will send the $\m{lot}$ (commodity being
auctioned, represented as a linear type), in the latter case, it will return the client's bid.
The linear product ($\tensor$) constructor is dual to $\lolli$ and denotes the transfer
of a linear
value from the contract to the client.  The $\auction$ type guarantees that a
client cannot collect during the $\m{running}$ phase, while they cannot bid during the
$\m{ended}$ phase.

Our discussion so far describes the interaction of one client with the auction, prescribing the
sequences of steps to be taken according to the protocol defined by the session type.  In
reality, however, an auction will have several clients.  Nomos uses a \emph{shared} session
type~\cite{BalzerICFP2017} to guarantee that bidders interact with the auction in mutual
exclusion from each other and that the sequences of actions are executed atomically.  To
demarcate the parts of the protocol that become a \emph{critical section}, the above session
type uses the $\up$ and $\down$ type modalities.  The $\up$ modality denotes the beginning of a
critical section, the $\down$ modality denotes its end.  Programmatically, $\up$ translates
into an \emph{acquire} of the auction session and $\down$ into the \emph{release} of the
session.  As indicated by the $\auction$ session type, acquire and release tend to be the begin
and end points of a session, framing the critical section that is described by a \emph{linear}
session type.

In Nomos, contracts are implemented by \emph{processes}, revealing
the concurrent, message-passing nature of session-typed languages. The
implementation below shows the process $\mi{run}$ representing the running
phase of the auction. It internally stores a linear hash map of bids $b :
\m{hashmap}_{\m{bid}}$ and a linear lot $l$ and offers on a shared channel
$sa : \auction$. The $\m{bid}$ session type (line~\ref{bid-tp}) can be queried
for the stored identifier and bid value, and is offered by a process (not
shown) that internally stores this identifier and money.
Line~\ref{bid-tp} shows the syntax for session type definitions.
\begin{ntabbing}
$\m{stype \; bid} = \echoice{\mb{addr} : \m{id} \product \m{bid},
\mb{val} : \m{money}}, \quad \m{stype \; bids} = \m{hashmap}_{\m{bid}}$
\label{bid-tp}\\
$(b : \m{bids}), (l : \m{lot}) \vdash
\mi{run} :: (sa : \auction)$ \label{run-decl}
\hspace{5em}\=\=\=\=\=\% \quad syntax for process declaration\\
\quad \= $\procdef{\mi{run}}{b \; l}{sa} =$ \label{run-def}
\>\% \quad syntax for process definition\\
\> \quad \= $\eaccept{la}{sa} \semi$ \label{run-accept}
\>\% \quad accept a client acquire request\\
\>\> $\esendl{la}{\m{running}} \semi$ \label{run-run}
\>\% \quad auction is running\\
\>\> $\casedef{la}$ \= $(\; \labdef{\mb{bid}}$
\= $\erecvch{la}{r} \semi$ \label{run-recv-id}
\>\% \quad receive identifier $r : \m{id}$\\
\>\>\>\> $\erecvch{la}{m} \semi$ \label{run-recv-money}
\>\% \quad receive bid $m : \m{money}$\\
\>\>\>\> $\edetach{sa}{la} \semi$ \label{run-detach}
\>\% \quad detach from client\\
\>\>\>\> $\procdef{\mi{addbid} \; r}{b \; m}{b'} \semi$ \label{run-addbid}
\>\% \quad store bid internally\\
\>\>\>\> $\procdef{\mi{check}}{b' \; l}{sa}$ \label{run-call}
\>\% \quad check if threshold reached\\
\>\>\> $\mid \labdef{\mb{cancel}}$
\= $\edetach{sa}{la} \semi$ \label{cancel-detach}
\>\% \quad detach from client\\
\>\>\>\> $\procdef{\mi{run}}{b \; l}{sa})$ \label{cancel-call}
\>\% \quad recurse
% \\
%\; $(b : \m{bids}),(w : \m{id}) \semi (M : \m{money}), (l : \m{lot}) \vdash$
%\label{end-decl}\\
%\hspace{12em}$\mi{end} :: (sa : \auction)$\\
%\quad \= $\procdef{\mi{end} \; b \; w}{M \; l}{sa} =$ \label{end-def}\\
%\> \quad \= $\eaccept{la}{sa} \semi$ \label{end-accept}\\
%\>\> $\esendl{la}{\m{ended}} \semi$ \label{end-end}\\
%\>\> $\casedef{la}$ \label{end-case}\\
%\>\> \quad \= $(\labdef{\mb{collect}}$
%\= $\erecvch{la}{r} \semi$ \label{end-recv-id}\\
%\>\>\>\> $\eif{w = r} \ethen ($ \label{end-comp}\\
%\>\>\>\>\quad \= $\esendl{la}{\m{won}} \semi$ \label{end-won}\\
%\>\>\>\>\> $\esendch{la}{l} \semi$ \label{end-lot}\\
%\>\>\>\>\> $\edetach{sa}{la} \semi$ \label{end-detach-if}\\
%\>\>\>\>\> $\procdef{\mi{end'} \; b \; w}{M}{sa})$ \label{end'-call}\\
%\>\>\>\> $\eelse ($ \label{end-else}\\
%\>\>\>\>\> $\esendl{la}{\m{lost}} \semi$ \label{end-lost}\\
%\>\>\>\>\> $(v, b') = \mi{findbid} \; b \; r \semi$ \label{end-find}\\
%\>\>\>\>\> $\esendl{M}{\m{subtract}} \semi$ \label{end-sub}\\
%\>\>\>\>\> $\esendch{M}{v} \semi$ \label{end-send-val}\\
%\>\>\>\>\> $\erecvch{M}{m} \semi$ \label{end-recv-money}\\
%\>\>\>\>\> $\esendch{la}{m} \semi$ \label{end-send-money}\\
%\>\>\>\>\> $\edetach{sa}{la} \semi$ \label{end-detach-else}\\
%\>\>\>\>\> $\procdef{\mi{end} \; b' \; w}{M \; l}{sa})$ \label{end-call}
\end{ntabbing}
The contract process first \emph{accepts} an acquire
request by a bidder (line~\ref{run-accept}) and then sends the message $\m{running}$
(line~\ref{run-run}) indicating the auction status.  It then waits for
the bidder's choice.  Should the bidder choose to make a bid, the process waits to
receive the
bidder's identifier (line~\ref{run-recv-id}) followed by money equivalent to the
bidder's bid (line \ref{run-recv-money}).  After this linear exchange, the process
leaves the critical section by issuing a \emph{detach} (line~\ref{run-detach}),
matching the bidder's release request.
Internally, the process stores the pair of the bidder's identifier and bid in the
data structure $\m{bids}$ (line~\ref{run-addbid}). The $\m{ended}$ protocol of the
contract is governed by a different process (not shown), responsible for
distributing the bids back to the clients.
%Process $\m{ended}$
%keeps in addition to this data structure also the identifier of the client who has won the
%auction ($w : \m{id}$), which it uses to send the winner the lot
%(lines~\ref{end-won} - \ref{end-lot}). Otherwise, if the bidder loses, the contract sends
%the message $\m{lost}$, followed by returning their bid (lines~\ref{end-lost} -
%\ref{end-send-money}). Also, $\m{bids}$ is updated to reflect this change
%(line~\ref{end-find}), so that a client is unable to withdraw their money twice. 
The contract transitions to the $\m{ended}$ state when
the number of bidders reaches a threshold (stored in auction). This is achieved
by the $\mi{check}$ process (line~\ref{run-call}) which checks if the threshold
has been reached and makes this transition, or calls $\mi{run}$ otherwise.

\paragraph{\textbf{\textit{Linear Assets}}}
Nomos integrates a linear type system that tracks the assets stored in a process. The type
system enforces that these assets are never duplicated, but only exchanged between
processes. Moreover, the type system forbids a process to terminate while it stores any linear
assets, preventing an asset from being discarded.  As an example, the auction contract treats
$\m{money}$ and $\m{lot}$ as linear assets, which is witnessed by the use of $\lolli$ and
$\tensor$ (type operators for linear exchange)
for their exchange in the $\auction$ session type.  In contrast, no provisions to
handle assets linearly exist in Solidity, allowing such assets to be created out of thin air,
or readily duplicated or discarded.  In the above $\m{bid}$ function, for instance, the
language does not prevent the programmer from writing $\m{bidValue[bidder] = bid}$ instead,
losing the bidder's previous bid.

\paragraph{\textbf{\textit{Re-Entrancy Vulnerabilities}}}
A contract function is re-entrant if, once called by an external user, it
can potentially be called again before the previous call is completed.
As an illustration,
consider the $\m{collect}$ function in Solidity of the auction contract
(on the left) where the funds are transferred to the bidder before the
hash map is updated to reflect this change.

\begin{minipage}{.5\textwidth}
  \begin{lstlisting}[basicstyle=\small\ttfamily, language=Solidity,
      columns=fullflexible]
function  collect() public payable {
  require (status == ended);
  bidder = msg.sender; bid = bidValue[bidder];
  bidder.send(bid); bidValue[bidder] = 0; }
  \end{lstlisting}
\end{minipage}
\hfill\vrule\vrule\hfill
\begin{minipage}{.4\textwidth}
  \begin{lstlisting}[basicstyle=\small\ttfamily, language=Solidity,
    columns=fullflexible]
function () payable {
  // 'auction' variable stores the
  // address to auction contract
  auction.collect(); }
  \end{lstlisting}
\end{minipage}

If a bidder creates a dummy contract with a function that calls $\m{collect}$
on the auction contract, it causes a re-entrant situation. The $\m{send}$
function on the left transfers execution control to the dummy contract, essentially
triggering an unnamed \emph{fallback} function (on the right) in the dummy
contract code base. The fallback function in turn calls $\m{collect}$ on the
auction leading to an infinite recursive call to $\m{collect}$,
depleting all funds from the auction. This vulnerability was exposed by the
infamous DAO attack \cite{TheDAO2016}, where \$60 million worth of funds were
stolen, and detecting them has since been critical~\cite{GrossmanCallback}.
The message-passing framework of session types eliminates this
vulnerability. While session types provide multiple clients access to a contract, the
acquire-release discipline ensures that clients interact with the contract in
mutual exclusion. To attempt re-entrancy, a bidder will need to acquire the auction
contract twice without releasing it, and the second acquire would fail to execute.

\paragraph{\textbf{\textit{Resource Cost}}}

Another important aspect of digital contracts is their \emph{resource usage}. The
state of
all the contracts is stored on the \emph{blockchain}, a distributed ledger which
records the history of all transactions. Executing a contract function, aka
\emph{transaction} and updating the blockchain state requires new blocks to be added
to the blockchain. In existing blockchains like Ethereum, this is done by
\emph{miners} who charge a fee based on the \emph{gas}
usage of the transaction, indicating the cost of its execution.
Precisely computing this cost is important because the sender of a transaction
must pay this fee to the miners. If the sender does not pay a sufficient amount,
the transaction will be rejected by the miners and the sender's fee is lost!

Nomos uses resource-aware session types~\cite{DasHP17} to statically analyze the resource cost
of a transaction. They operate by assigning an initial \emph{potential} to each process. This
potential is consumed by each operation that the process executes or can be transferred between
processes to share and amortize cost. The cost of each operation is defined by a cost
model. % Finally, the bounds produced by the resource-aware types are
% completely parametric in this cost model, making the technique extensible.
Resource-aware session types express the potential as part of the session type using the type
constructors $\getpot$ and $\paypot$. The $\getpot$ constructor prescribes that the client must send potential
to the contract, with the amount of potential indicated as a superscript.  Dually, the $\paypot$ constructor prescribes that the
contract must send potential to the client.  In case of the auction contract, we
require the client to pay potential for the operations that the contract
must execute, both while placing and collecting their bids.
If the cost model
assigns a cost of $1$ to each contract operation, then the maximum cost
of an auction session is $22$ (taking the max number of operations in all branches).
Thus, we require
the client to send $22$ units of potential at the start of a session.   In the $\m{cancel}$
branch of the $\m{auction}$ type, on the other hand, the
contract returns 21 units of potential to the client using the $\paypot^{21}$
type constructor. This is analogous to gas usage in smart contracts, where
the sender initiates a transaction with some initial gas, and the leftover
gas at the end of transaction is returned to the sender.  In contrast to existing smart
contract languages like Solidity, which provide no support for analyzing the cost of a
transaction, Nomos type soundness theorem guarantees that the total initial potential of a process plus the potential it receives
during a session reflect the upper bound on the gas usage, assuming that the cost model assigns a cost equivalent to their
gas cost to each operation.

\paragraph{\textbf{\textit{Bringing It All Together}}}
A main contribution of this paper is to combine all these features in a
single language while retaining type safety.  To this end, we introduce four
different \emph{modes} of a channel, identifying the role of the process
offered along that channel.
The mode $\p$ denotes \emph{purely linear processes}, typically amounting to
linear assets or private data structures, such as $b$ and $l$ in the auction.
The modes $\s$ and $\l$ denote \emph{sharable
processes} that are either in their shared phase or linear phase, respectively,
and are typically used for contracts, such as $sa$ and $la$, respectively,
in the auction.  The mode $\c$, finally, denotes a \emph{transaction process}
that can refer to shared and linear processes and is typically issued by a user, such
as bidder in the auction. The mode assignment carries over into the process
typing judgments (see Section~\ref{sec:sharing}) ascertaining certain
well-formedness conditions (Definition~\ref{def:proc_typ}) on their type. This 
introduction of modes is simply a technical device to preserve the tree
structure of linear processes at run-time, establishing type safety.

%%% Local Variables:
%%% mode: latex
%%% mode: flyspell
%%% TeX-master: "pldi19"
%%% End:

\section{Base System of Session Types}
\label{sec:session}

Nomos builds on linear session types for message-passing
concurrency~\cite{HondaCONCUR1993, HondaESOP1998, HondaPOPL2008,
CairesCONCUR2010, WadlerICFP2012}
and, in particular, on the line of works that have a logical
foundation due to the existence of a Curry-Howard correspondence between linear
logic and the session-typed
$\pi$-calculus~\cite{CairesCONCUR2010,WadlerICFP2012}.
Linear logic~\cite{GirardTCS1987} is a substructural logic that exhibits
exchange as the only structural property, with no contraction or
weakening. As a result, linear propositions can be viewed as resources
that must be used \emph{exactly once} in a proof.  Under the Curry-Howard correspondence, an
intuitionistic linear sequent $A_1, A_2, \ldots, A_n \vdash C$
can be interpreted as the offer of a session $C$ by a process $P$ using the sessions $A_1, A_2,
\ldots, A_n$
\begin{center}
\begin{minipage}{0cm}
\begin{tabbing}
$(x_1 : A_1), (x_2 : A_2), \ldots, (x_n : A_n) \vdash P :: (z : C)$
\end{tabbing}
\end{minipage}
\end{center}
We label each antecedent as well as the conclusion with the name of the channel along which the
session is provided. The $x_i$'s correspond to channels
\emph{used by} $P$, and $z$ is the channel \emph{provided by} $P$.
% The resulting judgment formally states that process $P$ provides a service
% of session type $C$ along channel $z$, while using services of session
% types $A_1, \ldots, A_n$ provided along channels $x_1, \ldots, x_n$
% respectively.
% All these channel names must be distinct, and we
% sometimes explicitly rename them to preserve this invariant.
As is standard, we use the linear context $\D$ to combine multiple assumptions.

\begin{table*}[t]
\centering
\begin{tabular}{@{}l l l l p{13.5em}@{}}
\textbf{Session Type} & \textbf{Contin-} & \textbf{Process Term} &
\textbf{Contin-} & \multicolumn{1}{c}{\textbf{Description}} \\
& \textbf{uation} & & \textbf{uation} \\
\midrule
$c : \ichoice{\ell : A_{\ell}}_{\ell \in L}$ & $c : A_k$ &
$\esendl{c}{k} \semi P$
& $P$ & provider sends label $k$ along $c$ \\
& & $\ecase{c}{\ell}{Q_{\ell}}_{\ell \in L}$ & $Q_k$ &
client receives label $k$ along $c$ \\
\addlinespace
$c : \echoice{\ell : A_{\ell}}$ & $c : A_k$ &
$\ecase{c}{\ell}{P_{\ell}}_{\ell \in L}$
& $P_k$ & provider receives label $k$ along $c$ \\
& & $\esendl{c}{k} \semi Q$ & $Q$ & client sends label $k$ along $c$ \\
\addlinespace
$c : A \tensor B$ & $c : B$ & $\esendch{c}{w} \semi P$
& $P$ & provider sends channel $w : A$ on $c$ \\
& & $\erecvch{c}{y} \semi Q_y$ & $[w/y]Q_y$ &
client receives channel $w : A$ on $c$ \\
\addlinespace
$c : A \lolli B$ & $c : B$ & $\erecvch{c}{y} \semi P_y$
& $[w/y]P_y$ & provider receives chan. $w : A$ on $c$ \\
& & $\esendch{c}{w} \semi Q$ & $Q$ &
client sends channel $w : A$ on $c$  \\
\addlinespace
$c : \one$ & $-$ & $\eclose{c}$
& $-$ & provider sends $\mi{end}$ along $c$ \\
& & $\ewait{c} \semi Q$ & $Q$ & client receives $\mi{end}$ along $c$ \\
\midrule
\end{tabular}
\caption{Overview of binary session types with their operational description}
\label{tab:basetypes}
\vspace{-2.5em}
\end{table*}

For the typing of processes in Nomos, we extend the above judgment with two
additional contexts ($\Psi$ and $\G$), a resource annotation $q$, and a mode $m$
of the offered channel:
\begin{center}
\begin{minipage}{0cm}
\begin{tabbing}
$\Psi \semi \G \semi \D \entailpot{q} P :: (x_m : A)$
\end{tabbing}
\end{minipage}
\end{center}
We will gradually introduce each concept in the remainder of this article.  For future
reference, we show the complete typing rules, with additional contexts, resource annotations,
and modes henceforth, but highlight the parts that will be discussed in later sections in blue.

 % For now, they can be
% viewed as constant, and are simply threaded through by the rules (marked in blue).

The Curry-Howard correspondence gives each connective of linear
logic an interpretation as a session type, as demonstrated by the grammar:
\begin{center}
\begin{minipage}{0cm}
\begin{tabbing}
$A, B \quad ::= \quad \ichoice{\ell : A}_{\ell \in K} \mid
\echoice{\ell : A}_{\ell \in K}
\mid A \lolli_m B \mid A \tensor_m B \mid \one$
\end{tabbing}
\end{minipage}
\end{center}
Each type prescribes the kind of
message that must be sent or received along a channel of that type
and at which type the session continues after the exchange. % We present an
% overview of the session types in SILL along with a brief
% description of their communication protocol. They follow the type
% grammar below.
%$V$ denotes a type variable.
Types are defined mutually
recursively in a global signature, where type definitions are
constrained to be \emph{contractive}~\cite{Gay05acta} (no definitions
of the form $V = A$ where $A$ is a type name). This allows us
to treat them equi-recursively~\cite{CraryPLDI1999}, meaning we can silently replace a
type variable by its definition for type-checking.

%Internal choice $A \oplus B$ and external choice $A \with B$ have been
%generalized to n-ary labeled sums $\ichoice{\ell : A_{\ell}}_{\ell \in K}$ and
%$\echoice{\ell : A_{\ell}}_{\ell \in K}$ (for some index set $L$),
%respectively. As a provider of an internal choice
%$\ichoice{\ell : A_{\ell}}_{\ell \in K}$, a process can send one of the labels
%$\ell$ to its client. Dually, a provider of an external choice
%$\echoice{\ell : A_{\ell}}_{\ell \in K}$ receives one of the labels $l_i$
%sent by its client.
%%
%We require
%external and internal choices to comprise at least one label,
%otherwise there would exist a linear channel without
%observable communication along it, which is computationally
%uninteresting and would complicate our proofs.
%The connectives $\tensor$ and $\lolli$ are used to send
%channels via other channels. A provider of $A \tensor B$
%sends a channel of type $A$ to its client and then
%behaves as a provider of $B$. Dually, a provider
%of $A \lolli B$ receives a channel of type $A$ from its
%client. The types of the provider and client change
%consistently, and the process terms of a provider
%and client come in matching pairs.

Following previous work on session types~\cite{PfenningFOSSACS2015,
ToninhoESOP2013}, the process expressions of Nomos are defined
as follows.
\begin{center}
\begin{minipage}{0cm}
\begin{tabbing}
$P ::= \esendl{x}{l} \semi P \mid \ecase{x}{\ell}{P}_{\ell \in K} \mid 
\fwd{x}{y}
\mid \eclose{x} \mid \ewait{x} \semi P \mid \esendch{x}{w} \semi P
\mid \erecvch{x}{y} \semi P$
\end{tabbing}
\end{minipage}
\end{center}

Table~\ref{tab:basetypes} provides
an overview of the types along with their operational meaning.
Because we adopt the intuitionistic version of linear logic, session types
are expressed from the point of view of the provider.
Table~\ref{tab:basetypes} provides the viewpoint of the provider in the
first line, and that of the client in the second line for each connective.
Columns 1 and 3 describe the session type and process term before
the interaction. Similarly, columns 2 and 4 describe the
type and term after the interaction. Finally, the last column describes
the provider and client action. Figure~\ref{fig:typing-binary} 
provides the corresponding typing rules. As illustrations of the statics and
semantics, we explain internal choice ($\oplus$) and linear implication
($\lolli$) connectives.

\paragraph{\textbf{\textit{Internal Choice}}}
The linear logic connective $A \oplus B$ has been
generalized to n-ary labeled sum $\ichoice{\ell : A_{\ell}}_{\ell \in K}$.
A process that provides $x : \ichoice{\ell : A_{\ell}}_{\ell \in K}$
can send any label $l \in K$ along $x$ and then continues by
providing $x : A_l$. The corresponding process term is written
as $(\esendl{x}{l} \semi P)$, where $P$ is the continuation. A client branches on the label
received along $x$ using the term $\ecase{x}{\ell}{Q_{\ell}}_{\ell \in K}$.
The typing rules for the provider and client are $\oplus R$
and $\oplus L$, respectively, in Figure~\ref{fig:typing-binary}.

\begin{figure}[t]
    \raggedright{
    \fbox{$\B{\Psi \semi \G \semi} \D \entailpot{\B{q}} P :: (x_{\B{m}} : A)$} \quad
    \text{Process $P$ uses linear channels in $\D$
    and offers type $A$ on channel $x$}}
    \begin{mathpar}
    \infer[\oplus R]
    {\B{\Psi \semi \G \semi} \D \entailpot{\B{q}}
    \esendl{x_{\B{m}}}{l} \semi P ::
    (x_{\B{m}} : \ichoice{\ell : A_{\ell}}_{\ell \in K})}
    {\B{\Psi \semi \G \semi} \D \entailpot{\B{q}}
    P :: (x_{\B{m}} : A_l) \qquad (l \in K)}
    \and
    \inferrule*[right = $\oplus L$]
    {\B{\Psi \semi \G \semi} \D, (x_{\B{m}} : A_{\ell}) \entailpot{\B{q}}
    Q_{\ell} :: (z_{\B{k}} : C)
    \qquad (\forall \ell \in K)}
    {\B{\Psi \semi \G \semi} \D, (x_{\B{m}} : \ichoice{\ell : A_{\ell}}_{\ell \in K})
    \entailpot{\B{q}} \ecase{x_{\B{m}}}{\ell}{Q_{\ell}}_{\ell \in K} :: (z_{\B{k}} : C)}
    \and
    %%
    %\infer[\with R]
    %{\Psi \semi \G \semi \D \entailpot{q}
    %\ecase{x_m}{\ell}{P_{\ell}}_{\ell \in L} ::
    %(x_m : \echoice{\ell : A_{\ell}}_{\ell \in L})}
    %{\Psi \semi \G \semi \D \entailpot{q}
    %P :: (x_m : A_{\ell}) \qquad (\forall \ell \in L)}
    %%
    %\and
    %%
    %\infer[\with L]
    %{\Psi \semi \G \semi \D, (x_m : \echoice{\ell : A_{\ell}}_{\ell \in L})
    %\entailpot{q} \esendl{x_m}{k} \semi P :: (z_k : C)}
    %{\Psi \semi \G \semi \D, (x_m : A_{\ell}) \entailpot{q}
    %Q_{\ell} :: (z_k : C)
    %\qquad (k \in L)}
    %%
    %\and
    %
    %\infer[\tensor R]
    %{\Psi \semi \G \semi \D, (w_\p : A) \entailpot{q}
    %\esendch{x_m}{w_\p} \semi P :: (x_m : A \tensor B)}
    %{\Psi \semi \G \semi \D \entailpot{q}
    %P :: (x_m : B)}
    %%
    %\and
    %%
    %\infer[\tensor L]
    %{\Psi \semi \G \semi \D, (x_m : A \tensor B)
    %\entailpot{q} \erecvch{x_m}{y_\p} \semi Q :: (z_k : C)}
    %{\Psi \semi \G \semi \D, (y_\p : A), (x_m : B) \entailpot{q}
    %Q :: (z_k : C)}
    %%
    %\and
    %
    \infer[\lolli_{\B{n}} R]
    {\B{\Psi \semi \G \semi} \D \entailpot{\B{q}}
    \erecvch{x_{\B{m}}}{y_{\B{n}}} \semi P :: (x_{\B{m}} : A \lolli_{\B{n}} B)}
    {\B{\Psi \semi \G \semi} \D, (y_{\B{n}} : A) \entailpot{\B{q}}
    P :: (x_{\B{m}} : B)}
    \and
    \infer[\lolli_{\B{n}} L]
    {\B{\Psi \semi \G \semi} \D, (w_{\B{n}} : A), (x_{\B{m}} : A \lolli_{\B{n}} B)
    \entailpot{\B{q}} \esendch{x_{\B{m}}}{w_{\B{n}}} \semi Q :: (z_{\B{k}} : C)}
    {\B{\Psi \semi \G \semi} \D, (x_{\B{m}} : B) \entailpot{\B{q}}
    Q :: (z_{\B{k}} : C)}
    \and
    %
    %\infer[\one R]
    %{\Psi {\semi} \G {\semi} \cdot \entailpot{q}
    %\eclose{x_m} :: (x_m : \one)}
    %{q = 0}
    %%
    %\and 
    %%
    %\infer[\one L]
    %{\Psi {\semi} \G {\semi} \D, (x_m : \one) \entailpot{q}
    %\ewait{x_m} ; Q :: (z_k : C)}
    %{\Psi \semi \G \semi \D \entailpot{q} Q :: (z_k : C)}
    %%
    %\and
    %
    \infer[\m{fwd}]
    {\B{\Psi \semi \G \semi} (y_{\B{m}} : A) \entailpot{\B{q}}
    \fwd{x_{\B{m}}}{y_{\B{m}}} :: (x_{\B{m}} : A)}
    {\B{q = 0}}
    \end{mathpar}
    \vspace{-2em}
    \caption{Selected typing rules for process communication}
    \label{fig:typing-binary}
    \vspace{-1em}
    \end{figure}

The operational semantics is formalized as a system of
\emph{multiset rewriting rules}~\cite{Cervesato2009SSOS}.
We introduce semantic objects $\proc{c_m}{w, P}$ and $\msg{c_m}{w, N}$
denoting process $P$ and message $N$, respectively, being provided along
channel $c$ at mode $m$.
The resource annotation $w$ indicates the work
performed so far, the discussion
of which we defer to Section~\ref{sec:resource}.
%A \emph{process
%configuration} is a multiset of such object, where any two offered
%channels are distinct.
Communication is \emph{asynchronous}, allowing  the
sender $(\esendl{c_m}{l} \semi P)$ to continue with $P$ without waiting for $l$
to be received. As a
technical device to ensure that consecutive messages arrive
in the order they were sent, the sender also creates a fresh
continuation channel $c^+_m$ so that the message $l$ is actually
represented as $(\esendl{c_m}{l} \semi \fwd{c_m}{c^+_m})$ (read: send $l$
along $c_m$ and continue as $c^+_m$):
\begin{center}
\begin{minipage}{0cm}
\begin{tabbing}
$(\oplus S):  \quad$
$\proc{c_{\B{m}}}{\B{w}, \esendl{c_{\B{m}}}{l} \semi P} \step
\proc{c^+_{\B{m}}}{\B{w}, [c^+_{\B{m}}/c_{\B{m}}]P},
\msg{c_{\B{m}}}{\B{0}, \esendl{c_{\B{m}}}{l} \semi \fwd{c_{\B{m}}}{c^+_{\B{m}}}}$
\end{tabbing}
\end{minipage}
\end{center}
Receiving the message $l$ corresponds to selecting branch $Q_l$
and substituting continuation $c^+$ for $c$:
\begin{tabbing}
$(\oplus C): \quad$
$\msg{c_{\B{m}}}{\B{w}, \esendl{c_{\B{m}}}{l} \semi \fwd{c_{\B{m}}}{c^+_{\B{m}}}},
\proc{d_{\B{k}}}{\B{w'}, \ecase{c_{\B{m}}}{\ell}{Q_{\ell}}_{\ell \in K}} \step$\\
\hspace{26em}$\proc{d_{\B{k}}}{\B{w+w'}, [c^+_{\B{m}}/c_{\B{m}}]Q_l}$
\end{tabbing}
The message $\msg{c_m}{w, \esendl{c_m}{l} \semi \fwd{c_m}{c^+_m}}$
is just a particular
form of process, where $\fwd{c_m}{c^+_m}$ is \emph{forwarding}, which is explained
below. Therefore, no separate typing rules for messages are needed;
they can be typed as processes~\cite{BalzerICFP2017}.

% \paragraph{\textbf{\textit{Delegation.}}}
\paragraph{\textbf{\textit{Channel Passing.}}}
Nomos allows % channels to be \emph{higher-order} allowing
the exchange of channels over channels, also referred to as higher-order channels.
% sometimes called \emph{delegation}.
  A process providing $A \lolli_n B$ can receive
a channel of type $A$ at mode $n$ and then continue with providing $B$.
The provider process term is $(\erecvch{x_m}{y_n} \semi P)$,
where $P$ is the continuation. The corresponding client sends this
channel using $(\esendch{x_m}{w_n} \semi Q)$. The corresponding typing
rules are presented in Figure~\ref{fig:typing-binary}. Operationally,
the client creates a message containing the channel:
\begin{center}
\begin{minipage}{0cm}
\begin{tabbing}
$(\lolli_{\B{n}} S): $
$\proc{d_{\B{k}}}{\B{w}, \esendch{c_{\B{m}}}{e_{\B{n}}} \semi P} \step
\msg{c^+_{\B{m}}}{\B{0}, \esendch{c_{\B{m}}}{e_{\B{n}}} \semi \fwd{c^+_{\B{m}}}{c_{\B{m}}}},
\proc{d_{\B{k}}}{\B{w}, [c^+_{\B{m}}/c_{\B{m}}]P}$
\end{tabbing}
\end{minipage}
\end{center}
The provider receives this channel, and substitutes it appropriately.
\begin{tabbing}
$(\lolli_{\B{n}} C): $
$\proc{c_{\B{m}}}{\B{w'}, \erecvch{c_{\B{m}}}{x_{\B{n}}} \semi Q},
\msg{c^+_{\B{m}}}{\B{w}, \esendch{c_{\B{m}}}{e_{\B{n}}} \semi \fwd{c^+_{\B{m}}}{c_{\B{m}}}}
\step$ \\
\hspace{24em} $\proc{c^+_{\B{m}}}{\B{w+w'}, [c^+_{\B{m}}/c_{\B{m}}][e_{\B{n}}/x_{\B{n}}]Q}$
\end{tabbing}
An important distinction from standard session types is that the
$\lolli$ and $\tensor$ types are decorated with the mode $m$ of the
channel exchanged. Since modes distinguish the status of the channels
in Nomos, this mode decoration is necessary to ensure type safety.

\paragraph{\textbf{\textit{Forwarding}}}
A forwarding process $\fwd{x_m}{y_m}$ (which provides channel $x$)
identifies channels $x$ and $y$ (both at mode $m$) so that any
further communication
along $x$ or $y$ occurs on the unified channel. The typing rule
$\m{fwd}$ is given in Figure~\ref{fig:typing-binary} and corresponds
to the logical rule of \emph{identity}.
\begin{center}
\begin{minipage}{0cm}
\begin{tabbing}
$(\m{id}^+ C): \quad$ \= $\msg{d_{\B{m}}}{\B{w'}, N},
\proc{c_{\B{m}}}{\B{w}, \fwd{c_{\B{m}}}{d_{\B{m}}}}$ \hspace{3em}\= $\step$ \quad \=
$\msg{c_{\B{m}}}{\B{w + w'}, [c_{\B{m}}/d_{\B{m}}]N}$\\
$(\m{id}^- C):$ \> $\proc{c_{\B{m}}}{\B{w}, \fwd{c_{\B{m}}}{d_{\B{m}}}},
\msg{e_{\B{k}}}{\B{w'}, N(c_{\B{m}})}$ \> $\step$ \>
$\msg{e_{\B{k}}}{\B{w + w'}, N(d_{\B{m}})}$
\end{tabbing}
\end{minipage}
\end{center}
Operationally, a process $\fwd{c}{d}$ \emph{forwards} any message
$N$ that arrives along $d$ to $c$ and vice versa.  Since linearity ensures that every process
has a unique client, forwarding results in terminating the forwarding process and corresponding
renaming of the channel in the client process.
% Since channels
% are linear with a unique provider and client, the forwarding process
% can terminate making sure to apply proper renaming. The rules for
% the operational semantics as follows.
%\[
%\inferrule*
%{\msg{d_k}{w', M}, \quad
%\proc{c_m}{\B{w}, \fwd{c_m}{d_k}} \step \\
%\msg{c_m}{w + w', [c_m/d_k]M}}
%{}
%\]
%\[
%\inferrule*
%{\proc{c_m}{\B{w}, \fwd{c_m}{d_k}}, \quad
%\msg{e_l}{w', M(c_m)} \step \\
%\msg{e_l}{w + w', M(d_k)}}
%{}
%\]
The full semantics are given in the appendix.

%%% Local Variables:
%%% mode: latex
%%% mode: flyspell
%%% TeX-master: "pldi19"
%%% End:

\section{Sharing Contracts}
\label{sec:sharing}

Multi-user support is fundamental to digital contract
development.  Linear session types, as defined in Section~\ref{sec:session}, unfortunately
preclude such sharing because they restrict processes to exactly one client;
only one bidder for the auction, for instance (who will always win!).  To support multi-user
contracts, we base Nomos on \emph{shared} session
types~\cite{BalzerICFP2017}.  Shared session types impose an
acquire-release discipline on shared processes to guarantee that multiple clients interact with
a contract in \emph{mutual exclusion} of each other.  When a client acquires a
shared contract, it obtains a private linear channel along which it can communicate
with the contract undisturbed
by any other clients.  Once the client releases the contract, it loses its private linear
channel and only retains a shared reference to the contract.

A key idea of shared session types is to lift the acquire-release discipline to the type level.
Generalizing the idea of type
\emph{stratification}~\cite{PfenningFOSSACS2015, BentonCSL1994,
ReedUnpublished2009}, session
types are stratified into a linear and shared layer with two \emph{adjoint modalities} going
back and forth between them:
\begin{tabbing}
\quad$A_\s \quad $\=$::= \quad$\=$ \up A_\l $ \hspace{5em} \= shared session type\\
\quad$A_\l \quad $\>$::= \quad$\=$\ldots \mid \;  \down A_\s \qquad$ \> linear session types
\end{tabbing}
The $\up$ type modality translates into an
\emph{acquire}, while the dual $\down$ type modality into a \emph{release}.
Whereas mutual exclusion is one key ingredient to guarantee session fidelity (a.k.a. type
preservation) for shared session types, the other key ingredient is the requirement that a
session type is \emph{equi-synchronizing}.  A session type is equi-synchronizing if it
imposes the invariant on a process to be released back to the same type at which the
process was previously acquired. This is also the key behind eliminating
\emph{re-entrancy vulnerabilities} since it prevents a user from interrupting an
ongoing session in the middle and initiating a new one.

Recall the process typing judgment in Nomos $\Psi \semi \G \semi \D
\entailpot{q} P :: (x_m : A)$ denoting a process $P$ offering service of
type $A$ along channel $x$ at mode $m$. The contexts $\G$ and $\D$ store the shared and
linear channels that $P$ can refer to, respectively ($\Psi$ and $q$ are explained later and thus
marked in blue in Figure~\ref{fig:typing-shared}).  The stratification of
channels into layers arises from a difference in structural properties that
exist for types at a mode. Shared propositions exhibit weakening, contraction
and exchange, thus can be discarded or duplicated, while linear propositions
only exhibit exchange.

% Nomos integrates shared and linear session types with a functional language,
% yielding the following typing judgment:
% \begin{center}
% \begin{minipage}{0cm}
% \begin{tabbing}
% $\Psi \semi \G \semi \D \entailpot{q} P :: (x : A)$
% \end{tabbing}
% \end{minipage}
% \end{center}
% This denotes a process $P$ providing a service of type $A$ along channel $x$
% and using functional variables in $\Psi$, shared channels in $\G$,
% and linear channels in $\D$. The stratification of channels into layers
% arises from a difference in structural properties that exist for types at a mode.
% Shared propositions exhibit weakening, contraction and exchange, thus can be discarded
% or duplicated, while linear propositions only exhibit exchange.

% Shared propositions (mode $\s$) exhibit weakening, contraction and
% exchange, thus can be discarded or duplicated, while linear propositions
% (mode $\l, \c, \p$) only exhibit exchange.

\paragraph{\textbf{\textit{Allowing Contracts to Rely on Linear Assets}}}
% \paragraph{\textbf{\textit{Allowing Shared Contracts to Rely on Linear Assets}}}

As exemplified by the auction contract, a digital contract typically amounts to
a process that is shared at the outset, but oscillates between shared and linear
to interact with clients, one at a time. Crucial for this pattern is the ability
of a contract to maintain its linear assets (e.g., $\m{money}$ or $\m{lot}$ for
the auction) regardless of its mode.  Unfortunately, current shared session
types~\cite{BalzerICFP2017} do not allow a shared process to rely on any linear
channels, requiring any linear assets to be consumed before becoming shared.
This precaution was logically motivated~\cite{PruiksmaTR2018} and also crucial
for type preservation.

A key novelty of our work is to lift this restriction while \emph{maintaining
type preservation}. The
main concern regarding preservation is to prevent a process from acquiring
its client, which would result in a cycle in the linear process tree.
To this end, we factorize the process typing judgment according to the
\emph{three roles} that arise in digital contract programs: \emph{contracts},
\emph{transactions}, and \emph{linear assets}. Since contracts are shared
and thus can oscillate between shared and linear, we get 4 sub-judgments
for typing processes, each characterized by the mode of the
channel being offered.

\begin{figure}[t]
  \[
  \begin{array}{rcl}
  A_\p & ::= & \ichoice{\ell : A_\p}_{\ell \in L} \mid
  \echoice{\ell : A_\p}_{\ell \in L} \mid \one
  \mid A_m \lolli_m A_\p \mid A_m \tensor_m A_\p
  \mid \tau \arrow A_\p \mid \tau \product A_\p \\
  A_\l & ::= & \ichoice{\ell : A_\l}_{\ell \in L} \mid
  \echoice{\ell : A_\l}_{\ell \in L} \mid \one
  \mid A_m \lolli_m A_\l \mid A_m \tensor_m A_\l
  \mid \tau \arrow A_\l \mid \tau \product A_\l \mid \; \down A_\s \\
  A_\s & ::= & \up A_\l \\
  A_\c & ::= & A_\p \\
  \end{array}
  \]
  \vspace{-2em}
  \caption{Grammar for shared session types}
  \label{fig:type-grammar}
  \vspace{-1.5em}
\end{figure}

\begin{definition}[Process Typing]\label{def:proc_typ}
  The judgment $\Psi \semi \G \semi \D \entailpot{q} P :: (x_m : A)$
  is categorized according to mode $m$. This factorization imposes
  certain invariants on the judgment outlined below. $\lang{A}$ denotes the
  language generated by the grammar of $A$.

  \begin{enumerate}
    \item If $m = \p$, then (i) $\G$ is empty, (ii) for all $d_k \in \D
    \implies k = \p$, and (iii) $A \in \lang{A_\p}$.

    \item If $m = \s$, then (i) for all $d_k \in \D \implies k = \p$,
    and (ii) $A \in \lang{A_\s}$.

    \item If $m = \l$, then $A \in \lang{A_\l}$.

    \item If $m = \c$, then $A \in \lang{A_\c}$.
  \end{enumerate}

\end{definition}

Figure~\ref{fig:type-grammar} shows the session type grammar in Nomos.
The first sub-judgment in Definition~\ref{def:proc_typ} is for typing
linear assets. These type a purely linear
process $P$ using a purely linear context $\D$ (types belonging to grammar
$A_\p$ in Figure~\ref{fig:type-grammar}) and offering a purely linear type $A$ along
channel $x_\p$. The mode $\p$ of the channel indicates that a purely linear
session is offered.
The second and third sub-judgments are for typing contracts.
The second sub-judgment shows the type of a contract process $P$ using 
a shared context $\G$ and a purely linear channel context $\D$
(judgment $\D \plin$) and offering shared type $A$ on the shared channel
$x_\s$. Once this shared channel is acquired by a user, the shared
process transitions to its linear phase, whose typing is governed by
the third sub-judgment. The offered channel transitions to linear mode $\l$,
while the linear context may now contain channels at arbitrary modes ($\l, \c$
or $\p$). \emph{This allows contracts to interact with other contracts without
compromising type safety}. Finally, the fourth typing judgment types a 
linear process, corresponding to a \emph{transaction}
holding access to shared channels $\G$ and linear channels $\D$, and
offering at mode $\c$.

This novel factorization and the fact that contracts, as the only
shared processes, can only access linear channels at mode $\p$,
upholds preservation while allowing shared contract processes to
rely on linear resources.

Shared session types introduce new typing rules into our system,
concerning the \emph{acquire-release} constructs (see Figure
\ref{fig:typing-shared}).  An acquire is applied to the shared channel $x_\s$ along which
the shared process offers and yields a linear channel $x_\l$ when
successful.
%The shared channel $x_\s$ is not made available to the continuation $Q$.
A contract process can \emph{accept} an acquire request along
its offering shared channel $x_\s$. After the accept is successful,
the shared contract process transitions to its linear phase, now offering
along the linear channel $x_\l$.

The synchronous dynamics of the \emph{acquire-accept} pair is
\begin{tabbing}
$(\up C): $
$\proc{a_\s}{\B{w'}, \eaccept{x_\l}{a_\s} \semi P_{x_\l}},
\proc{c_m}{\B{w}, \eacquire{x_\l}{a_\s} \semi Q_{x_\l}} \step$ \\
\hspace{25em}
$\proc{a_\l}{\B{w'}, P_{a_\l}},
\proc{c_m}{\B{w}, Q_{a_\l}}$
\end{tabbing}
This rule exploits the invariant that a contract process' providing channel
$a$ can come at two different modes, a linear one $a_\l$, and a shared
one $a_\s$. The linear channel $a_\l$ is substituted for the
channel variable $x_\l$ occurring in the process terms $P$ and $Q$.

The dual to acquire-accept is \emph{release-detach}. A client can
\emph{release} linear access to a contract process, while the contract
process \emph{detaches} from the client. The corresponding typing
rules are presented in Figure~\ref{fig:typing-shared}. The effect of releasing the
linear channel $x_\l$ is that the continuation $Q$ loses access to
$x_\l$, while a new reference to $x_\s$ is made available in the
shared context $\G$.  The contract, on the other hand, detaches from the client by
transitioning its offering channel from linear mode $x_\l$ back
to the shared mode $x_\s$.
Operationally, the release-detach rule is inverse to the acquire-accept rule.
\begin{tabbing}
$(\down C): $
$\proc{a_\l}{\B{w'}, \edetach{x_\s}{a_\l} \semi P_{x_\s}},
\proc{c_m}{\B{w}, \erelease{x_\s}{a_\l} \semi Q_{x_\s}} \step$ \\
\hspace{24em}
$\proc{a_\s}{\B{w'}, P_{a_\s}}, \quad \proc{c_m}{\B{w}, Q_{a_\s}}$
\end{tabbing}

\begin{figure}[t]
  \begin{mathpar}
  \raggedright{
  \fbox{$\B{\Psi \semi} \G \semi \D \entailpot{\B{q}} P :: (x_m : A)$} \quad
  \text{
  Process $P$ uses shared channels in $\G$
  and offers $A$ along $x$.}}
  \vspace{-0.5em}\and
  \infer[\up L]
  {\B{\Psi \semi} \G, (x_\s : \up A_\l) \semi \D
  \entailpot{\B{q}} \eacquire{x_\l}{x_\s} \semi Q :: (z_m : C)}
  {\B{\Psi \semi} \G \semi \D, (x_\l : A_\l)
  \entailpot{\B{q}} Q :: (z_m : C)}
  \and
  \infer[\up R]
  {\B{\Psi \semi} \G \semi \D \entailpot{\B{q}}
  \eaccept{x_\l}{x_\s} \semi P :: (x_\s : \up A_\l)}
  {\D \plin \and \B{\Psi \semi} \G \semi \D \entailpot{\B{q}} P :: (x_\l : A_\l)}
  \and
  \infer[\down L]
  {\B{\Psi \semi} \G \semi \D, (x_\l : \down A_\s)
  \entailpot{\B{q}} \erelease{x_\s}{x_\l} \semi Q :: (z_m : C)}
  {\B{\Psi \semi} \G, (x_\s : A_\s) \semi \D
  \entailpot{\B{q}} Q :: (z_m : C)}
  \and
  \infer[\down R]
  {\B{\Psi \semi} \G \semi \D \entailpot{\B{q}}
  \edetach{x_\s}{x_\l} \semi P :: (x_\l : \down A_\s)}
  {\D \plin \and \B{\Psi \semi} \G \semi \D \entailpot{\B{q}} P :: (x_\s : A_\s)}
  \end{mathpar}
  \vspace{-2em}
  \caption{Typing rules corresponding to the shared layer.}
  \label{fig:typing-shared}
  \vspace{-1em}
  \end{figure}

\section{Adding a Functional Layer}
\label{sec:fun}
To support general-purpose programming patterns, Nomos combines linear channels with conventional data
structures, such as integers, lists, or dictionaries. To reflect and track different classes of data in the type
system, we take inspiration from prior work~\cite{ToninhoESOP2013,
PfenningFOSSACS2015} and incorporate processes
into a functional core via a \emph{linear contextual monad} that isolates
session-based concurrency. To this end, we introduce a separate
functional context to the typing of a process.
%This context exhibits
%all structural properties of logic, namely weakening, contraction and
%exchange, and differs from the linear context, which exhibits only
%exchange.
The linear contextual monad encapsulates open concurrent
computations, which can be passed in functional computations but also
transferred between processes in the form of \emph{higher-order
  processes}, providing a uniform integration of higher-order
functions and processes.
%
%\jan{@Frank: Can you please briefly discuss related work on combining
%  linear and functional types.}

%We provide a brief overview of the types and constructs introduced
%to the base session types by the functional fragment.
The types are
separated into a functional and concurrent part, mutually dependent on
each other. The functional types $\tau$
are given by the type grammar below.
\begin{center}
\begin{minipage}{0cm}
\begin{tabbing}
$\tau \quad$\=$::= \quad $\=$\tau \to \tau
\mid \tau + \tau \mid \tau \times \tau
\mid \tint \mid \tbool \mid \plist{\tau}{q}$\\
\>$\; \; \mid$\>$\tproc{A_\p}{\overline{A_\p}}_\p \mid
\tproc{A_\s}{\overline{A_\s} \semi \overline{A_\p}}_\s \mid
\tproc{A_\c}{\overline{A_\s} \semi \overline{A}}_\c$
\end{tabbing}
\end{minipage}
\end{center}
The types are standard,
except for the potential annotation $q \in \mathbb{N}$ in list type
$\plist{\tau}{q}$, which we explain in Section~\ref{sec:resource}, and
the contextual monadic types in the last line, which are the topic of this section.
The expressivity of the types and terms in the functional layer
are not important for the development in this paper. Thus, we do not
formally define functional terms $M$ but assume that they have the
expected term formers such as function abstraction and application,
type constructors, and pattern matching. We also define a standard type
judgment for the functional part of the language.
%without providing the rules for the standard constructs.
\begin{center}
\begin{minipage}{0cm}
\begin{tabbing}
$\Psi \exppot{p} M : \tau \quad \text{term $M$ has type $\tau$ in functional
context $\Psi$ (potential $p$ discussed later)}$
\end{tabbing}
\end{minipage}
\end{center}

\paragraph{\textbf{\textit{Contextual Monad}}}
The main novelty in the functional types are the three type formers
for contextual monads, denoting the type of a process expression.  The
type $\tproc{A_\p}{\overline{A_\p}}_\p$ denotes a process offering a
\emph{purely linear} session type $A_\p$ and using the purely linear
vector of types
$\overline{A_\p}$. The corresponding introduction form in the
functional language is the monadic value constructor
$\eproc{c_\p}{P}{\overline{d_\p}}$,
denoting a runnable process offering along channel $c_\p$ that uses
channels $\overline{d_\p}$, all at mode $\p$.
The corresponding typing rule for the monad is (ignore the blue portions)
\begin{center}
\begin{minipage}{0cm}
\begin{tabbing}
$\infer[\{\} I_\p]
{\Psi \exppot{\B{q}} \eproc{x_\p}{P}{\overline{d_\p}} :
\tproc{A}{\overline{D}}_\p}
{\D = \overline{d_\p : D} \qquad
\Psi \semi \cdot \semi \D \entailpot{\B{q}} P :: (x_\p : A)}$
\end{tabbing}
\end{minipage}
\end{center}

The monadic \emph{bind} operation implements process composition and
acts as the elimination form for values of type
$\tproc{A_\p}{\overline{A_\p}}_\p$.  The bind operation, written as
$\ecut{c_\p}{M}{\overline{d_\p}}{Q_c}$, composes the process underlying the
monadic term $M$, which offers along channel $c_\p$ and uses channels
$\overline{d_\p}$, with $Q_c$, which uses $c_\p$. The typing rule for the 
monadic bind is rule $\{\}E_{\p\p}$ in Figure~\ref{fig:typing-functional}.
The linear context is split between the monad $M$ and continuation
$Q$, enforcing linearity. Similarly, the potential in the functional
context is split using the sharing judgment ($\share$),
explained in Section~\ref{sec:resource}. The shared context $\G$ is empty in
accordance with the invariants established in Definition~\ref{def:proc_typ}
\emph{(i)}, since the mode of offered channel $z$ is $\p$.
The effect of executing a bind is the spawn
of the purely linear process corresponding to the monad $M$, and the
parent process continuing with $Q$.
The corresponding operational semantics rule (named $\m{spawn}_{\p\p}$)
is given as follows:
\begin{mathpar}
\proc{d_\p}{\B{w}, \ecut{x_\p}{\eproc{x_\p'}{P_{x_\p',\overline{y}}}
{\overline{y}}}{\overline{a}}{Q}} \step
\proc{c_\p}{\B{0}, P_{c_\p,\overline{a}}}, \proc{d_\p}{\B{w},
[c_\p/x_\p]Q}
{}
\end{mathpar}
The above rule spawns the process $P$ offering along a globally
fresh channel $c_\p$, and using channels $\overline{a}$. The
continuation process $Q$ acts as a client for this fresh channel $c_\p$.
The other two monadic types correspond to spawning a shared process
$\tproc{A_\s}{\overline{A_\s} \semi \overline{A_\p}}_\s$ and a transaction process
$\tproc{A_\c}{\overline{A_\s} \semi \overline{A}}_\c$ at mode $\s$ and
$\c$, respectively. Their rules are analogous to $\{\}I_\p$ and $\{\}E_{\p\p}$.

\begin{figure}[t]
\begin{mathpar}
\raggedright{
\fbox{$\Psi \semi \G \semi \D \entailpot{\B{q}} P :: (x_m : A)$} \quad
\text{Process $P$ uses functional values in $\Psi$,
and provides $A$ along $x$.}}
\and
\inferrule*[right = $\{\}E_{\p\p}$]
{\B{r = p+q} \and
\D = \overline{d_\p : D} \and
\B{\Psi \share (\Psi_1, \Psi_2)} \\
\Psi_1 \exppot{\B{p}} M : \tproc{A}{\overline{D}} \and
\Psi_2 \semi \cdot \semi \D', (x_\p : A) \entailpot{\B{q}} Q :: (z_\p : C)}
{\Psi \semi \cdot \semi \D, \D' \entailpot{\B{r}} \ecut{x_\p}{M}{\overline{d_\p}}{Q} ::
(z_\p : C)}
\and
%%
%\infer[\product R]
%{\Psi \semi \G \semi \D \entailpot{r}
%\esendch{x_m}{M} \semi P :: (x_m : \tau \product A)}
%{r = p+q \qquad
%\Psi \exppot{p} M : \tau \qquad
%\Psi \semi \G \semi \D \entailpot{q}
%P :: (x_m : A)}
%%
%\and
%%
%\infer[\product L]
%{\Psi \semi \G \semi \D, (x_m : \tau \product A)
%\entailpot{q} \erecvch{x_m}{y} \semi Q :: (z_k : C)}
%{\Psi, (y : \tau) \semi \G \semi \D, (x_m : A) \entailpot{q}
%Q :: (z_k : C)}
%%
%\and
%
\infer[\arrow R]
{\Psi \semi \G \semi \D \entailpot{\B{q}}
\erecvch{x_m}{y} \semi P :: (x_m : \tau \arrow A)}
{\Psi, (y : \tau) \semi \G \semi \D \entailpot{\B{q}}
P :: (x_m : A)}
\and
\inferrule*[right = $\arrow L$]
{\B{r = p+q} \qquad
\B{\Psi \share (\Psi_1, \Psi_2)} \qquad
\Psi_1 \exppot{\B{p}} M : \tau \\
\Psi_2 \semi \G \semi \D, (x_m : A) \entailpot{\B{q}}
Q :: (z_k : C)}
{\Psi \semi \G \semi \D, (x_m : \tau \arrow A)
\entailpot{\B{r}} \esendch{x_m}{M} \semi Q :: (z_k : C)}
\end{mathpar}
\vspace{-2em}
\caption{Typing rules corresponding to the functional layer.}
\label{fig:typing-functional}
\vspace{-1.5em}
\end{figure}

\paragraph{\textbf{\textit{Value Communication}}}
Communicating a \emph{value} of the functional language along a channel is expressed
at the type level by adding the following two session types.
\begin{center}
\begin{minipage}{0cm}
\begin{tabbing}
$A ::= \ldots \mid \tau \arrow A \mid \tau \product A$
\end{tabbing}
\end{minipage}
\end{center}
The type $\tau \arrow A$ prescribes receiving a value of type $\tau$
with continuation type $A$, while its dual $\tau \product A$ prescribes
sending a value of type $\tau$ with continuation $A$. The corresponding
typing rules for arrow ($\arrow R, \arrow L$) are given in
Figure~\ref{fig:typing-functional} (rules for $\product$ are inverse).
Receiving a value adds it to the functional context $\Psi$, while
sending it requires proving that the value has type $\tau$.
Semantically, sending a value $M : \tau$ creates a message
predicate along a fresh channel $c_m^+$ containing the value:
\begin{tabbing}
$(\arrow S): $
$\proc{d_k}{\B{w}, \esendch{c_m}{M} \semi P} \step
\msg{c^+_m}{\B{0}, \esendch{c_m}{M} \semi \fwd{c^+_m}{c_m}},
\proc{d_k}{\B{w}, [c^+_m/c_m]P}$
\end{tabbing}
The recipient process substitutes $M$ for $x$, and continues to
offer along the fresh continuation channel received by the message.
This ensures that messages are received in the order they are sent.
The rule is formalized below.
\begin{tabbing}
$(\arrow C): $
$\proc{c_m}{\B{w'}, \erecvch{c_m}{x} \semi Q},
\msg{c^+_m}{\B{w}, \esendch{c_m}{M} \semi \fwd{c^+_m}{c_m}} \step$\\
\hspace{25em}$\proc{c^+_m}{\B{w+w'}, [c^+_m/c_m][M/x]Q}$
\end{tabbing}

\paragraph{\textbf{\textit{Tracking Linear Assets}}}
As an illustration, consider the type $\m{money}$ introduced in
the auction example (Section~\ref{sec:overview}). The type is
an abstraction over funds stored in a process and is described as
\begin{tabbing}
$\m{money} =
\echoice{$\=$\mb{value} : \tint \product \m{money},$
\hspace{12.5em}\=\=\% send value\\
\>$\mb{add} : \m{money} \lolli_\p \m{money},$
\>\% receive money and add it\\
\>$\mb{subtract} : \tint \arrow \ichoice{$\=
$\mb{sufficient} : \m{money} \tensor_\p \m{money},$
\>\% receive int, send money\\
\>\>$\mb{insufficient} : \m{money}}$
\>\% funds insufficient to subtract\\
\>$\mb{coins} : \lc}$
\>\>\% send list of coins
\end{tabbing}
The type supports querying for value, and addition and subtraction.
The type also expresses insufficiency of funds in the case of
subtraction. The provider process only supplies money to the
client if the requested amount is less than the current balance,
as depicted in the $\m{sufficient}$ label. The type is implemented
by a $\mi{wallet}$ process that internally stores a linear list of coins
and an integer representing its value. Since linearity is only enforced
on the list of coins in the linear context, we trust the programmer
updates the integer in the functional context correctly during
transactions. The process is typed and implemented as (modes of channels
$l$ and $m$ is $\p$, skipped in the definition for brevity)
\begin{ntabbing}
\reset
$(n : \tint) \semi (l_\p : \lc) \vdash \mi{wallet} :: (m_\p : \m{money})$
\label{wal-decl}\\
\quad \= $\procdef{\mi{wallet} \; n}{l}{m} =$ \label{wal-def}\\
\>\quad \= $\casedef{m}$ \label{wal-case}
\hspace{14em}\=\=\=\= \% \quad case analyze on label received on $m$\\
\>\> \quad \= $(\labdef{\mb{value}}$
\= $\esendch{m}{n} \semi$ \label{wal-val}
\>\% \quad receive value, send $n$\\
\>\>\>\> $\procdef{\mi{wallet} \; n}{l}{m}$ \label{wal-val-rec}\\
\>\>\> $\mid \labdef{\mb{add}}$
\= $\erecvch{m}{m'} \semi$ \label{wal-add}
\>\% \quad receive $m' : \m{money}$ to add\\
\>\>\>\> $\esendl{m'}{\mb{value}} \semi$ \label{wal-add-val}
\>\% \quad query value of $m'$\\
\>\>\>\> $\erecvch{m'}{v} \semi$ \label{wal-get-val}\\
\>\>\>\> $\esendl{m'}{\mb{coins}} \semi$ \label{wal-send-coins}
\>\% \quad extract list of coins stored in $m'$\\
\>\>\>\> $\procdef{append}{l \; m'}{k} \semi$ \label{wal-append}
\>\% \quad append list received to internal list\\
\>\>\>\> $\procdef{\mi{wallet} \; (n+v)}{k}{m}$ \label{wal-add-rec}\\
\>\>\> $\mid \labdef{\mb{subtract}}$
\= $\erecvch{m}{n'} \semi$ \label{wal-sub}
\>\% \quad receive int to subtract\\
\>\>\>\> $\eif{n' > n} \ethen$ \label{wal-if}\\
\>\>\>\> \quad \=
$\esendl{m}{\mb{insufficient}} \semi$ \label{wal-insuff}
\>\% \quad funds insufficient\\
\>\>\>\>\> $\procdef{\mi{wallet} \; n}{l}{m}$ \label{wal-insuff-rec}\\
\>\>\>\> $\eelse$ \label{wal-else}\\
\>\>\>\> \quad \=
$\esendl{m}{\mb{sufficient}} \semi$ \label{wal-suff}
\>\% \quad funds sufficient\\
\>\>\>\>\> $\procdef{\mi{remove} \; n'}{l}{l'} \semi$ \label{wal-rem}
\>\% \quad remove $n'$ coins from $l$\\
\>\>\>\>\> $\erecvch{l'}{k} \semi$ \label{wal-get-sub}
\>\% \quad and create its own list\\
\>\>\>\>\> $\procdef{\mi{wallet} \; n'}{k}{m'} \semi$ \label{wal-create-wal}
\>\% \quad new $\mi{wallet}$ process for subtracted funds \\
\>\>\>\>\> $\esendch{m}{m'} \semi$ \label{wal-send-wal}
\>\% \quad send new money channel to client\\
\>\>\>\>\> $\procdef{\mi{wallet} \; (n-n')}{l'}{m}$ \label{wal-sub-rec}\\
\>\>\> $\mid \labdef{\mb{coins}}$
\= $\fwd{m}{l})$ \label{wal-coins}
\end{ntabbing}
If the $\mi{wallet}$ process receives the message $\m{value}$, it sends
back the integer $n$, and recurses (lines~\ref{wal-val}
and \ref{wal-val-rec}). If it receives the message $\m{add}$ followed
by a channel of type $\m{money}$ (line~\ref{wal-add}), it
queries the value of the received money $m'$
(line~\ref{wal-add-val}), stores it in $v$ (line~\ref{wal-get-val}),
extracts the coins stored in $m'$ (line~\ref{wal-send-coins}),
and appends them to its internal list of coins (line~\ref{wal-append}).
Similarly, if the $\mi{wallet}$ process receives the message
$\m{subtract}$ followed by an integer, it compares the requested
amount against the stored funds. If the balance is insufficient, it
sends the corresponding label, and recurses (lines~\ref{wal-insuff}
and \ref{wal-insuff-rec}). Otherwise, it removes $n'$ coins using the
$\mi{remove}$ process (line~\ref{wal-rem}), creates a $\m{money}$
abstraction using the $\mi{wallet}$ process (line~\ref{wal-create-wal}),
sends it (line~\ref{wal-send-wal}) and recurses.
Finally, if the $\mi{wallet}$ receives the message $\m{coins}$, it simply
forwards its internal list along the offered channel.

%%% Local Variables:
%%% mode: latex
%%% TeX-master: "pldi19"
%%% End:

\section{Tracking Resource Usage}
\label{sec:resource}

Resource usage is particularly important in digital contracts: Since
multiple parties need to agree on the result of the execution of a
contract, the computation is potentially performed multiple times or
by a trusted third party. This immediately introduces the need to
prevent denial of service attacks and to distribute the cost of the
computation among the participating parties.

The predominant approach for smart contracts on blockchains like
Ethereum is not to restrict the computation model but to introduce a
cost model that defines the \emph{gas} consumption of low level
operations. Any transaction with a smart contract needs to be executed
and validated before adding it to the global distributed ledger, i.e.,
blockchain. This validation is performed by \emph{miners}, who charge
fees based on the gas consumption of the transaction. This fee has to
be estimated and provided by the sender prior to the transaction. If
the provided amount does not cover the gas cost, the money falls to
the miner, the transaction fails, and the state of the contract is
reverted back. Overestimates bare the risk of high losses if the
contract has flaws or vulnerabilities.

It is not trivial to decide on the right amount for the fee since the
gas cost of the contract does not only depend on the requested
transaction but also on the (a priori unknown) state of the blockchain.
Thus, precise and static estimation of gas cost facilitates
transactions and reduces risks. We discuss our approach of tracking
resource usage, both at the functional and process layer.

\paragraph{\textbf{\textit{Functional Layer}}}

Numerous techniques have been proposed to statically derive resource
bounds for functional programs~\cite{LagoG11, AvanziniICFP15,
DannerICFP15, CicekBGGH16, Radicek17}. In
Nomos, we adapt the work on automatic amortized resource analysis
(AARA)~\cite{Jost03, HoffmannAH10} that has been implemented in
Resource Aware ML (RaML)~\cite{HoffmannW15}. RaML can automatically
derive worst-case resource bounds for higher-order polymorphic
programs with user-defined inductive types. The derived bounds are
multivariate resource polynomials of the size parameters of the
arguments. AARA is parametric in the resource metric and can deal with
non-monotone resources like memory that can become available during
the evaluation.

As an illustration, consider the function $\mi{applyInterest}$ that
iterates over a list of balances and applies interest on each element,
multiplying them by a constant $c$. An imperative version of the same
function in Solidity is implemented in Section~\ref{sec:study}.
%
%\jan{is the tick function already introduced earlier? Also: Do we have
%  a name clash between the session type tick and the functional tick?}
We use $\mi{tick}$ annotations to define the resource
usage of an expression in this article. We have annotated the code
to count the number of multiplications. The resource usage of an
evaluation of $\mi{applyInterest} \; b$ is $|b|$.

\begin{lstlisting}[basicstyle=\small\ttfamily, language=caml,
columns=fullflexible]
let applyInterest balances =
  match balances with
  | [] -> []
  | hd::tl -> tick(1); (c*hd)::(applyInterest tl) (* consume unit potential for tick *)
\end{lstlisting}

The idea of AARA is to decorate base types with potential annotations
that define a potential function as in amortized analysis. The typing
rules ensure that the potential before evaluating an expression is
sufficient to cover the cost of the evaluation and the potential
defined by the return type. This posterior potential can then be
used to pay for resource usage in the continuation of the program.
For example, we can derive the following resource-annotated type.
\begin{center}
\begin{minipage}{0cm}
\begin{tabbing}
$\mi{applyInterest} : \plist{\tint}{1} \funtype{0/0} \plist{\tint}{0}$
\end{tabbing}
\end{minipage}
\end{center}
The type $\plist{\tint}{1}$ denotes a list of integers assigning a
unit potential to each element in the list.  The return value, on
the other hand, has no potential. The
annotation on the function arrow indicates that we do not need any
potential to call the function and that no constant potential is left
after the function call has returned.

In a larger program, we might want to call the function $\mi{applyInterest}$
again on the result of a call to the function. In this case, we would
need to assign the type $\plist{\tint}{1}$ to the resulting
list and require $\plist{\tint}{2}$ for the argument. In
general, the type for the function can be described with symbolic
annotations with linear constraints between them. To derive a
worst-case bound for a function the constraints can be solved by an
off-the-shelf LP solver, even if the potential functions are
polynomial~\cite{HoffmannAH10,HoffmannW15}.
%
%In the auction process from Figure~\ref{fig:auction-process}, we would
%like to have a constant worst-case bound for finding a bit. We can
%achieve that by stopping the list traversal after checking 50
%elements (capacity of the auction).
%We can then derive the following type expressing the
%constant worst-case bound.
%\begin{center}
%\begin{minipage}{0cm}
%\begin{tabbing}
%$\mi{findbid} : L^0(\tint \times \tint) \times \tint \funtype{51/0}
%L^0(\tint \times \tint) \times \tint \; \m{option}$
%\end{tabbing}
%\end{minipage}
%\end{center}

In Nomos, we simply adopt the standard typing judgment of AARA for
functional programs.
\begin{center}
\begin{minipage}{0cm}
\begin{tabbing}
$\Psi \exppot{q} M : \tau$
\end{tabbing}
\end{minipage}
\end{center}
It states that under the resource-annotated functional context
$\Psi$, with constant potential $q$, the expression $M$ has the
resource-aware type $\tau$.

The operational \emph{cost} semantics is defined by the judgment
\begin{center}
\begin{minipage}{0cm}
\begin{tabbing}
$\bigeval{M}{V}{\mu}$
\end{tabbing}
\end{minipage}
\end{center}
which states that the closed expression $M$ evaluates to the value $V$
with cost $\mu$. The type soundness theorem states that if
$\cdot \exppot{q} M : \tau$ and $\bigeval{M}{V}{\mu}$ then
$q \geq \mu$. 

More details about AARA can be found in the literature
\cite{Jost03,HoffmannW15} and the appendix.

%\jan{So this semantics is not the standard formulation. Usually, we
%  use an environment. I think, the soundness proof would go through
%  like this but it's not really clear.}

\paragraph{\textbf{\textit{Process Layer}}}

To bound the resource usage of a process, Nomos features
resource-aware session types \cite{DasHP17} for work
analysis.  Resource-aware session types describe resource contracts
for inter-process communication.  The type system supports amortized
analysis by assigning potential to both
messages and processes. The derived resource bounds are functions of
interactions between processes.
As an illustration, consider the following resource-aware list interface
from prior work~\cite{DasHP17}.
\begin{tabbing}
$\tlist{A} = \ichoice{$\=$\pot{\m{nil}}{0} : \pot{\one}{0},
\pot{\m{cons}}{1} : A \tensorpot{0} \tlist{A}}$
\end{tabbing}
The type prescribes that the provider of a list must send one unit
of potential with every $\m{cons}$ message that it sends. Dually, a client
of this list will receive a unit potential with every $\m{cons}$ message.
All other type constructors are marked with potential $0$, and exchanging
the corresponding messages does not lead to transfer of potential.

% Potential stored in a process can be used to perform a \emph{tick},
% just like RAML. Thus, the total initial potential of a system upper bounds
% the total tick operations performed during the computation.

While resource-aware session types in Nomos are equivalent to the existing
formulation~\cite{DasHP17}, our version is simpler and
more streamlined. Instead of requiring every message to carry a potential
(and potentially tagging several messages with 0 potential), we
introduce two new type constructors for exchanging potential.
\begin{center}
\begin{minipage}{0cm}
\begin{tabbing}
$A ::= \ldots \mid \tpaypot{A}{r} \mid \tgetpot{A}{r}$
\end{tabbing}
\end{minipage}
\end{center}
The type $\tpaypot{A}{r}$ requires the provider to pay $r$ units of
potential which are transferred to the client. Dually, the type
$\tgetpot{A}{r}$ requires the client to pay $r$ units of potential
that are received by the provider.  Thus, the reformulated list type becomes
\begin{tabbing}
$\tlist{A} = \ichoice{$\=$\m{nil} : \one,
\m{cons} : \tpaypot{(A \tensor \tlist{A})}{1}}$
\end{tabbing}
The reformulation is more compact since we need to account for potential
in only the typing rules corresponding to $\tpaypot{A}{r}$ and
$\tgetpot{A}{r}$.

With all aspects introduced, the process typing judgment
\begin{center}
\begin{minipage}{0cm}
\begin{tabbing}
$\Psi \semi \G \semi \D \entailpot{q} P :: (x_m : A)$
\end{tabbing}
\end{minipage}
\end{center}
denotes a process $P$ accessing functional variables in $\Psi$,
shared channels in $\G$, linear channels in $\D$, offers
service of type $A$ along channel $x$ at mode $m$ and stores
a non-negative constant potential $q$. Similarly, the expressing
typing judgment
\begin{center}
  \begin{minipage}{0cm}
  \begin{tabbing}
  $\Psi \exppot{p} M : \tau$
  \end{tabbing}
  \end{minipage}
  \end{center}
denotes that expression $M$ has type $\tau$ in the presence of
functional context $\Psi$ and potential $p$.

Figure~\ref{fig:potential_types} shows the rules that interact with
the potential annotations.  In the rule $\getpot R$, process $P$ storing
potential $q$ receives $r$ units along the offered
channel $x_m$ using the \emph{get} construct and the continuation
executes with $p = q+r$ units of potential. 
In the dual rule $\getpot L$, a process storing potential $q = p+r$
sends $r$ units along the channel $x_m$
in its context using the \emph{pay} construct, and the continuation
remains with $p$ units of potential.
The typing rules for the dual constructor $\tpaypot{A}{r}$
are the exact inverse.
Finally, executing the $\etick{r}$ construct consumes $r$ potential from the
stored process potential $q$, and the continuation remains with
$p = q-r$ units, as described in the $\m{tick}$ rule.

\begin{figure}[t]
  \begin{mathpar}
\raggedleft{
\fbox{$\Psi \semi \G \semi \D \entailpot{q} P :: (x_m : A)$} \quad
\text{Process $P$ has potential $q$ and provides type $A$
along channel $x$.}}
%\and
%\infer[\paypot R]
%{\Psi \semi \G \semi \D \entailpot{q} \epay{x_m}{r} \semi P ::
%(x_m : \tpaypot{A}{r})}
%{q = p+r \qquad
%\Psi \semi \G \semi \D \entailpot{p} P :: (x_m : A)}
%%
%\and
%%
%\infer[\paypot L]
%{\Psi \semi \G \semi \D, (x_m : \tpaypot{A}{r}) \entailpot{q}
%\eget{x_m}{r} \semi P :: (z_k : C)}
%{p = q+r \qquad
%\Psi \semi \G \semi \D, (x_m : A) \entailpot{p} P :: (z_k : C)}
%
\and
\infer[\getpot R]
{\Psi \semi \G \semi \D \entailpot{q} \eget{x_m}{r} \semi P ::
(x_m : \tgetpot{A}{r})}
{p = q+r \qquad
\Psi \semi \G \semi \D \entailpot{p} P :: (x_m : A)}
\and
\infer[\getpot L]
{\Psi \semi \G \semi \D, (x_m : \tgetpot{A}{r}) \entailpot{q}
\epay{x_m}{r} \semi P :: (z_k : C)}
{q = p+r \qquad
\Psi \semi \G \semi \D, (x_m : A) \entailpot{p} P :: (z_k : C)}
\and
\infer[\m{tick}]
{\Psi \semi \G \semi \D \entailpot{q} \etick{r} \semi P :: (x_m : A)}
{q = p + r \qquad
\Psi \semi \G \semi \D \entailpot{p} P :: (x_m : A)}
  \end{mathpar}
  \vspace{-2em}
  \caption{Selected typing rules corresponding to potential.}
  \label{fig:potential_types}
  \vspace{-1.5em}
\end{figure}

\paragraph{\textbf{\textit{Integration}}}
Since both AARA for functional programs and resource-aware session types are based on the
integration of the potential method into their type systems, their combination
is natural. The two points of integration of the functional and process
layer are (i) spawning a process, and (ii) sending/receiving a value from
the functional layer. Recall the spawn rule $\{\}E_{\p\p}$ from
Figure~\ref{fig:typing-functional}.
%\[
%\inferrule*[Right = $\{\}E_{PP}$]
%{r = p+q \quad
%\D = \overline{d_\p : D} \quad
%\Psi \exppot{p} M : \tproc{A}{\overline{D}}_P \\
%\Psi \semi \D', (x_\p : A) \entailpot{q} Q :: (z_\p : C)}
%{\Psi \semi \D, \D' \entailpot{r} \ecut{x_\p}{M}{d_\p}{Q} :: (z_\p : C)}
%\]
A process storing potential $r = p+q$ can spawn a process corresponding
to the monadic value $M$, if $M$ needs $p$ units of potential to
evaluate, while the continuation needs $q$ units of potential to execute.
Moreover, the functional context $\Psi$ is shared in the two premises as
$\Psi_1$ and $\Psi_2$ using the judgment $\Psi \share (\Psi_1, \Psi_2)$.
This judgment, already explored in prior work~\cite{HoffmannW15}
describes that the base types in $\Psi$ are copied to both $\Psi_1$ and
$\Psi_2$, but the potential is split up. For instance, $L^{q_1 + q_2}(\tau)
\share (L^{q_1}(\tau), L^{q_2}(\tau))$. The rule $\arrow L$ follows a
similar pattern.
%Recalling the $\product R$ rule,
%\[
%\infer[\product R]
%{\Psi \semi \G \semi \D \entailpot{r}
%\esendch{x_m}{M} \semi P :: (x_m : \tau \product A)}
%{r = p+q \qquad
%\Psi \exppot{p} M : \tau \qquad
%\Psi \semi \G \semi \D \entailpot{q}
%P :: (x_m : A)}
%\]
%A process storing $r = p+q$ units of potential can send an expression
%$M$ if evaluating $M$ needs $p$ units of potential, while the continuation
%needs $q$ units of potential.
Thus, the combination of the two type systems is smooth, assigning
a uniform meaning to potential, both for the functional and process layer.

Remarkably, this technical device of exchanging functional values can
be used to exchange non-constant potential with messages. As an
illustration, we revisit the $\m{auction}$ protocol introduced in
Section~\ref{sec:overview}. Suppose the bids were stored in a list,
instead of a hash map, thus making the cost of collection of
winnings linear in the worst case, rather than constant. A user would
then be required to send a linear potential after acquiring the
contract. This can be done by sending a natural number
$n : \m{nat}^q$, storing potential $q \cdot |n|$ (like a unary list),
where $q$ is the cost of iterating over an element in the list of
bids. The contract would then iterate over the first $n$ elements of
the list and refund the remaining gas if $n$ exceeds the length.
Since the auction state is public, a user can view the size of the list of
bids, compute the required potential, store it in a natural number,
and transfer it.
It would still be possible that a user does not
provide enough fuel to reach the sought-after element in the
list. However, this behavior is clearly visible in the protocol and
code and out-of-gas exceptions are not possible.

\paragraph{\textbf{\textit{Operational Cost Semantics}}}
The resource usage of a process (or message) is tracked in semantic
objects $\proc{c}{w, P}$ and $\msg{c}{w, N}$ using the local counters
$w$. This signifies that the process $P$ (or message $N$) has performed
\emph{work} $w$ so far. The rules of semantics that explicitly affect
the work counter are
\begin{center}
\begin{minipage}{0cm}
\begin{tabbing}
$\infer[\m{internal}]
{\proc{c_m}{w, P[M]} \step \proc{c_m}{w+\mu, P[V]}}
{\bigeval{M}{V}{\mu}}$
\end{tabbing}
\end{minipage}
\end{center}
This rule describes that if an expression $M$ evaluates to $V$ with cost
$\mu$, then the process $P[M]$ depending on monadic expression
$M$ steps to $P[V]$,
while the work counter increments by $\mu$, denoting the total number
of internal steps taken by the process.
%Here, $P[\cdot] = \ecut{c}{[\cdot]}
%{\overline{a_i}}{P_c}$ or $\esendch{c}{[\cdot]} \semi P$, because those
%are the only expressions that integrate the functional and process layer.
%This accounts for the cost incurred at the functional layer.
At the process
layer, the work increments on executing a \emph{tick} operation.
\begin{center}
\begin{minipage}{0cm}
\begin{tabbing}
$\proc{c_m}{w, \etick{\mu} \semi P} \step \proc{c_m}{w+\mu, P}$
\end{tabbing}
\end{minipage}
\end{center}
A new process (or message) is spawned with $w = 0$, and a terminating process
transfers its work to the corresponding message it interacts with
before termination, thus preserving the total work performed by
the system.

%%% Local Variables:
%%% mode: latex
%%% mode: flyspell
%%% TeX-master: "icfp19"
%%% End:

\section{Type Soundness}
\label{sec:sound}
The main theorems that exhibit the connections between our
type system and the operational cost semantics are the usual
\emph{type preservation} and \emph{progress}. First,
Definition~\ref{def:proc_typ} asserts certain invariants
on process typing judgment depending on the mode of the
channel offered by a process. This mode, remains invariant,
as the process evolves. This is ensured by the process typing
rules, which remarkably preserve these invariants despite
being parametric in the mode.

\begin{lemma}[Invariants]\label{lem:invariants}
The typing rules on the judgment $\Psi \semi \G \semi \D
\entailpot{q} (x_m : A)$ preserve the invariants outlined
in Definition~\ref{def:proc_typ},
i.e., if the conclusion satisfies the invariant, so do all the
premises.
\end{lemma}

\paragraph{\textbf{\textit{Configuration Typing}}}
At run-time, a program evolves into a number of processes
and messages, represented by $\m{proc}$ and $\m{msg}$
predicates. This multiset of predicates is referred to as a
\emph{configuration} (abbreviated as $\W$).
\[
\W ::= \cdot \mid \W, \proc{c}{w, P} \mid \W, \msg{c}{w, N}
\]
A key question is how to type these configurations because a configuration both uses and
provides a number of channels. The solution is to have the typing imposes a partial order among the
processes and messages, requiring the provider of a channel
to appear before its client. We stipulate that no two
distinct processes or messages in a well-formed configuration
provide the same channel $c$.

The typing judgment for configurations has the form
$\Sg \semi \G_0 \potconf{E} \W :: (\G \semi \D)$
defining a configuration
$\W$ providing shared channels in $\G$ and linear channels in
$\D$. Additionally, we need to track the mapping between
the shared channels and their linear counterparts offered by a contract
process, switching back and forth between them when the
channel is acquired or released respectively. This mapping, along
with the type of the shared channels, is
stored in $\G_0$. $E$ is a natural number and
stores the sum of the total potential and work as recorded
in each process and message. We call $E$ the energy of the
configuration. The appendix details the configuration
typing rules.

Finally, $\Sg$ denotes a signature storing the type and function
definitions. A signature is well-formed if \emph{(i)} every type
definition $V = A_V$ is \emph{contractive}~\cite{Gay05acta}
and \emph{(ii)} every function definition $f = M : \tau$
is well-typed according to the expression typing judgment
$\Sg \semi \cdot \exppot{p} M : \tau$.
The signature does not contain process definitions;
every process is encapsulated inside a function using the
contextual monad.

\begin{theorem}[Type Preservation]\label{thm:pres}
~
\begin{itemize}[leftmargin=*]
\item If a closed well-typed expression $\cdot \exppot{q} M : \tau$
evaluates to a value, i.e., $\bigeval{M}{V}{\mu}$, then
$q \geq \mu$ and $\cdot \exppot{q-\mu} V : \tau$.

\item Consider a closed well-formed and well-typed configuration
$\W$ such that $\Sg \semi \G_0 \potconf{E} \W :: (\G \semi \D)$.
If the configuration takes a step, i.e. $\W \step \W'$, then there
exist $\G_0', \G'$ such that $\Sg \semi \G_0' \potconf{E} \W' ::
(\G' \semi \D)$, i.e., the resulting configuration is well-typed.
Additionally, $\G_0 \subseteq \G_0'$ and $\G \subseteq \G'$.
\end{itemize}
\end{theorem}

The preservation theorem is standard for expressions
\cite{HoffmannW15}. For processes, we proceed by induction
on the operational cost semantics and inversion on
the configuration and process typing judgment.

To state progress, we need the notion of a \emph{poised} process~\cite{PfenningFOSSACS2015}.  A process $\proc{c_m}{w, P}$ is poised if it is trying to receive a
message on $c_m$. Dually, a message $\msg{c_m}{w, N}$
is poised if it is sending along $c_m$. A configuration is
poised if every message or process in the configuration is
poised. Intuitively, this means that the configuration is
trying to interact with the outside world along a channel
in $\G$ or $\D$. Additionally, a process can be \emph{blocked}~\cite{BalzerICFP2017}
if it is trying to acquire a contract process that has already
been acquired by some process. This can lead
to the possibility of deadlocks.

\begin{theorem}[Progress]\label{thm:prog}
Consider a closed well-formed and well-typed configuration $\W$
such that $\G_0 \potconf{E} \W :: (\G \semi \D)$. Either $\W$ is
poised, or it can take a step, i.e., $\W \step \W'$, or some
process in $\W$ is blocked along $a_\s$ for some shared channel
$a_\s$ and there is a process $\proc{a_\l}{w, P} \in \W$.
\end{theorem}
The progress
theorem is weaker than that for binary linear session types,
where progress guarantees deadlock freedom due to absence
of shared channels.

%%% Local Variables:
%%% mode: latex
%%% TeX-master: "pldi19"
%%% End:

%Jan says: this definition leaks outside of this file

\section{Implementation and Evaluation}
\label{sec:study}
We have developed an open-source prototype
implementation~\cite{NomosGithub} of Nomos in OCaml. This
prototype contains a lexer and parser (369 lines of code), a type
checker (3039 lines of code), a pretty printer (500 lines of code),
and an LP solver interface (914 lines of code).

\paragraph{\textbf{\textit{Syntax}}}
The lexer and parser for Nomos have been implemented in Menhir
\cite{Menhir}, an LR(1) parser generator for OCaml. A
Nomos program is a list of mutually recursive type and process
definitions. To visually separate out functional variables from
session-typed channels, we require that shared channels are prefixed
by $\#$, while linear channels are prefixed by $\$$. This avoids
confusion between the two, both for the programmer and the parser.
We also require the programmer to indicate the
mode of the process being defined: \emph{asset}, \emph{contract} or
\emph{transaction}, assigning the respective modes $\p$, $\s$ and $\c$ to the
offered channel. The modes for all other channels are
inferred automatically (explained later). The initial
potential $\{q\}$ of a process is marked on the turnstile in
the declaration. The syntax for definitions is
\begin{lstlisting}
stype v = A
proc <mode> f : (x1 : T), ($c2 : A), ... |{q}- ($c : A) = M
\end{lstlisting}
In the context, $\m{T}$ is the functional type for variable $\m{x1}$,
while $\m{A}$ is the session type for channel $\m{\$c2}$ and $\m{M}$
is a functional expression
implementing the process. We add syntactic sugar, such as the forms
$\m{let \;x = M ; P}$ and $\m{if \; M \; then \; P_1 \; else \; P_2}$, to the
process layer to ease programming. Finally, a functional expression
can enter the session type monad using $\{\}$, i.e., $\m{M = \{ P \}}$
where $P$ is a session-typed expression.

\paragraph{\textbf{\textit{Type Checking}}}
We implemented a bi-directional~\cite{PierceBDTC} type checker
with a specific focus on the quality of error messages, which include,
for example, \emph{extent} (source code location) information for each
definition and expression. The programmer
provides the initial type of each variable and channel in the declaration
and the definition is checked against it, while reconstructing the
intermediate types. This helps localize the source of a type error
as the point where type reconstruction fails.
Type equality is implemented using a standard co-inductive algorithm
\cite{Gay05acta}. \emph{Type checking is linear time in the size of the program},
which is important in settings where type checking is
part of the attack surface.

\paragraph{\textbf{\textit{Potential and Mode Inference}}}
The potential and mode annotations are the most interesting aspects of
the Nomos type system. Since modes are associated with each channel,
they are tedious to write. Similarly, the exact potential annotations
depend on the cost assigned to each operation and is difficult to
predict statically. Thus, we implemented an inference algorithm of both these
annotations. 

Using ideas from existing techniques for type inference for
AARA~\cite{Jost03,HoffmannW15}, we reduce the reconstruction of 
potential annotations to linear optimization.
To this end, Nomos' type checker uses the Coin-Or LP solver. 
In a Nomos program, the programmer
can indicate unknown potential using $*$. Thus, resource-aware session
types can be marked with $\paypot^*$ and $\getpot^*$, list types can be
marked as $\plist{\tau}{*}$ and process definitions can be marked with
$|\{*\}-$ on the turnstile. The
mode of all the channels is marked as `unknown' while parsing.

The inference engine iterates over the program and substitutes
the star annotations with potential variables and `unknown' with mode
variables. Then, the bidirectional typing rules are applied, approximately
checking the program (modulo potential and mode annotations) while also generating
linear constraints for potential annotations (see Figure~\ref{fig:typing-functional}).
and mode annotations (see Definition~\ref{def:proc_typ}
and Figure~\ref{fig:typing-shared}). Finally, these constraints are
shipped to the LP solver, which is minimizing the value of the
potential annotations to achieve tight bounds. The LP solver either returns
that the constraints are infeasible, or returns a satisfying assignment,
which is then substituted into the program. The final program is pretty
printed for the programmer to view and verify the potential and mode
annotations.

\subsection{Case Studies}
We evaluate the design of Nomos by implementing several smart contract
applications and discussing the typical issues that arise. All the
contracts are implemented and type checked in the prototype implementation
and the potential and mode annotations are derived automatically by
the inference engine. The cost model used for these examples assigns 1 unit
of cost to every atomic internal computation and sending of a message.
We show the contract types from the implementation with the
following ASCII format: i) \lstinline{/\} for $\up$, ii) \lstinline{\/} for $\down$,
iii) \lstinline!<{q}|! for $\getpot^{q}$, iv) \lstinline!|{q}>! for
$\paypot^{q}$, v) \lstinline{^} for $\product$, vi) \lstinline{*[m]}
for $\tensor_m$, vii) \lstinline{-o[m]} for $\lolli_m$.

\paragraph{\textbf{\textit{ERC-20 Token Standard}}}
Tokens are a representation of a particular asset or utility, that
resides on top of an existing blockchain. ERC-20~\cite{ERC-20}
is a technical standard for smart contracts on the Ethereum blockchain
that defines a common list of standard functions that a token contract
has to implement. The majority of tokens on the Ethereum
blockchain are ERC-20 compliant. The standard requires the following
functions to be implemented:
\begin{itemize}
\item $\m{totalSupply()}$ : returns the total number of tokens in
supply as an integer.

\item $\m{balanceOf(id \; \mi{owner})}$ : returns the account balance of
$\mi{owner}$.

\item $\m{transfer(id \; \mi{to}, int \; \mi{value})}$ : transfers $\mi{value}$
tokens from sender's account to identifier $\mi{to}$.

\item $\m{transferFrom(id \; \mi{from}, id \; \mi{to}, int \; \mi{value})}$ :
transfers $\mi{value}$ number of tokens from identifier $\mi{from}$
to identifier $\mi{to}$.

\item $\m{approve(id \; \mi{spender}, int \; \mi{value})}$ : allows $\mi{spender}$
to withdraw from sender's account up to $\mi{value}$ number of tokens.

\item $\m{allowance(id \; \mi{owner}, id \; \mi{spender})}$ : returns the
number of tokens $\mi{spender}$ is allowed to withdraw from $\mi{owner}$.

\end{itemize}

The ERC-20 token contract implements the following session type in Nomos:
\lstset{basicstyle=\ttfamily\footnotesize}
\begin{lstlisting}[columns=fullflexible,]
type erc20token = /\ <{11}| &{ 
  totalSupply : int ^ |{9}> \/ erc20token,
  balanceOf : id -> int ^ |{8}> \/ erc20token,
  transfer : id -> id -> int -> |{0}> \/ erc20token,
  approve : id -> id -> int -> |{6}> \/ erc20token,
  allowance : id -> id -> int ^ |{6}> \/ erc20token }
\end{lstlisting}
% \begin{tabbing}
% $\m{erc20token} = \up \echoice{$\=
% $\mb{totalSupply} : \tint \product \down
% \m{erc20token},$\\
% \hspace{4em}\=$\mb{balanceOf} : \m{id} \arrow \tint \product \down
% \m{erc20token},$\\
% \>$\mb{transfer} : \m{id} \arrow \m{id} \arrow \tint \arrow \down
% \m{erc20token},$\\
% \>$\mb{transferFrom} : \m{id} \arrow \m{id} \arrow \tint \arrow \down
% \m{erc20token},$\\
% \>$\mb{approve} : \m{id} \arrow \m{id} \arrow \tint \arrow \down
% \m{erc20token},$\\
% \>$\mb{allowance} : \m{id} \arrow \m{id} \arrow \tint \product \down
% \m{erc20token}}$
% \end{tabbing}
The type ensures that the token implements the protocol underlying
the ERC-20 standard. To query the total number of
tokens in supply, a client sends the $\m{totalSupply}$ label, and
the contract sends back an integer. If the contract
receives the $\m{balanceOf}$ label followed by the owner's identifier,
it sends back an integer corresponding to the owner's balance. A
balance transfer can be initiated by sending the $\m{transfer}$ label to
the contract followed by sender's and receiver's identifier, and the
amount to be transferred. 
If the contract receives $\m{approve}$, it receives the
two identifiers and the value, and updates the allowance internally. Finally, this
allowance can be checked by issuing the $\m{allowance}$ label,
and sending the owner's and spender's identifier.

A programmer can design their own implementation (contract) of the
$\m{erc20token}$ session type. Internally, the contract relies on
custom coins created and named by its owner and used exclusively
for exchanges among private accounts. These coins can be minted
by a special transaction that can only be issued by the owner and that
creates coins out of thin air (consuming gas to create coins).
Depending on the functionality intended by the owner, they can employ
different types to represent their coins. For instance, choosing type $\one$, the 
multiplicative unit from linear logic, will allow both creation and destruction
of coins ``for free". A $\mi{mint-one}$ process, typed as $\cdot \vdash
\mi{mint-one} :: (c : \one)$, can create coin $c$ out of thin air (by
closing channel $c$) and a
$\mi{burn-one}$ process, typed as $(c : \one) \vdash \mi{burn-one} ::
(d : \one)$, will destroy the coin $c$ (by waiting on channel $c$).
Nomos' linear type system enforces that the coins are
treated linearly modulo minting and burning. Any transaction that does
not involve minting or burning ensures linearity of these coins.

One specific implementation of the $\m{erc20token}$ session type can be
achieved by storing two lists, one for the balance of each account, and
one for the allowance between each pair of accounts. The account
balance needs to be treated linearly, hence we place this balance list
in the linear context, while we store the allowance list in the functional
context. In this contract, we call the custom coin $\plcoin$, and use
$\plcoins$ to mean $\tlist{\plcoin}$.
The account balance is abstracted using the $\m{account}$
type:
\begin{tabbing}
$\m{account} = \echoice{$\=$\mb{addr} : \m{id} \product \m{account},$
\hspace{8em}\= \% \quad send identifier\\
\>$\mb{add} : \plcoins \lolli \m{account},$
\>\% \quad receive \plcoinsn{} and add internally\\
\>$\mb{subtract} : \tint \arrow \plcoins \tensor \m{account}}$
\>\% \quad receive integer, send \plcoinsn{}
\end{tabbing}
This allows a client to query for the identifier stored in the account,
as well as add and subtract from the account balance. We ignore
the resource consumption as it is not relevant to the example.
The $\mi{balance}$ process provides the $\m{account}$ abstraction.
It internally
stores the identifier in its functional context and \plcoinsn{} in its linear
context, and offers along the linear $\m{account}$ type.
\begin{tabbing}
$(r : \m{id}) \semi (M : \plcoins) \vdash \mi{balance} ::
(acc : \m{account})$
\end{tabbing}
Finally, the contract process stores the allowances as a list of triples
storing the owner's and sender's address and allowance value,
typed as $\m{id} \product \m{id} \product \tint$. Thus, the
$\mi{plcontract}$ process stores the allowance in the functional
context, and the list of accounts in its linear context and offers along
the $\m{erc20token}$ type introduced earlier.
\begin{tabbing}
$(allow : \tlist{\m{id} \product \m{id} \product \tint}) \semi
(accs : \tlist{\m{account}}) \vdash \mi{plcontract} ::
(st : \m{erc20token})$
\end{tabbing}
As an illustration, we show the part of the implementation for initiating a
transfer.
\begin{ntabbing}
\reset
$\procdef{\mi{plcontract} \; allow}{accs}{st} =$ \label{erc-def}\\
\quad \= $\eaccept{lt}{st} \semi$ \label{erc-accept}
\hspace{10em}\=\=\=\% \quad accept a client acquire request\\
\> $\casedef{lt} \ldots$ \label{erc-case}
\>\% \quad switch on label on $lt$\\
\> \quad \= $\mid \labdef{\mb{transfer}}$
\= $\erecvch{lt}{s} \semi$ \label{erc-recv-s}
\>\% \quad receive sender's identifier $s : \m{id}$\\
\>\>\> $\erecvch{lt}{r} \semi$ \label{erc-recv-r}
\>\% \quad receive receiver's identifier $r : \m{id}$\\
\>\>\> $\erecvch{lt}{v} \semi$ \label{erc-recv-v}
\>\% \quad receive transfer value $v : \tint$\\
\>\>\> $\edetach{st}{lt} \semi$ \label{erc-detach}
\>\% \quad detach from client\\
\>\>\> $\ldots$ \label{ecr-code1}
\>\% \quad extract sender and receiver's account \ldots\\
\>\>\> $\ldots$ \label{ecr-code1}
\>\% \quad and store in $sa$ and $ra$ resp.\\
\>\>\> $\esendl{sa}{\mb{subtract}} \semi$ \label{erc-sub}
\>\% \quad subtract \plcoinsn{} corresponding to \ldots\\
\>\>\> $\esendl{sa}{v} \semi$ \label{erc-send-v}
\>\% \quad $v$ from account channel $sa$\\
\>\>\> $\erecvch{sa}{m}$ \label{erc-recv-m}
\>\% \quad receive transfer amount $m$\\
\>\>\> $\esendl{ra}{\mb{add}} \semi$ \label{erc-add}
\>\% \quad add $m : \plcoins$ to \ldots\\
\>\>\> $\esendch{ra}{m} \semi$ \label{erc-send-m}
\>\% \quad account channel $ra$\\
\>\>\> $\procdef{\mi{plcontract} \; allow}{accs}{st}$
\label{erc-recurse}
\end{ntabbing}
The contract first receives the sender and receiver's identifiers
(lines~\ref{erc-recv-s} and \ref{erc-recv-r}) and the transfer
value $v$. The contract then detaches from the client (line
\ref{erc-detach}). We skip the code of extracting the sender's
and receiver's account from the list $accs$ and store them
in $sa$ and $ra$ of type $\m{account}$, respectively. The
contract then subtracts the \plcoinsn{} from account $sa$
(lines~\ref{erc-sub} and \ref{erc-send-v}) and receives and
stores them in $m$ (line~\ref{erc-recv-m}). This balance is then
added to $ra$'s account (lines~\ref{erc-add} and \ref{erc-send-m}).
An important point here is that Nomos enforces linearity
of the transfer transaction. Since $m : \plcoins$ is typed
as a linear asset, it cannot be discarded or modified. The
amount deducted from sender must be transferred to the receiver
(since no minting is involved here).

\paragraph{\textbf{\textit{Hacker Gold (HKG) Token}}}
The HKG token is one particular implementation of the ERC-20
token specification. Recently, a vulnerability was discovered
in the HKG token smart contract based on a typographical
error leading to a re-issuance of the entire token~\cite{HKG-News}.

The typographical error in the contract came about when updating the
receiver's balance during a transfer. Instead of
writing \lstinline{balance += value}, the programmer mistakenly
wrote \lstinline{balance =+ value} (semantically meaning
\lstinline{balance = value}). Moreover, while testing
this error was missed, because the first transfer always succeeds
(since the two statements are semantically equivalent when
\lstinline{balance = 0}). Nomos' type system would have
caught the linearity violation in the latter statement that
drops the existing balance in the recipient's account.

\paragraph{\textbf{\textit{Puzzle Contract}}}
This contract, taken from prior work~\cite{Oyente-CCS16}
rewards users who solve a computational puzzle and submit the solution.
The contract allows two functions, one that allows the owner
to update the reward, and the other that allows a user to submit
their solution and collect the reward.

In Nomos, this contract is implemented to offer the type
\begin{lstlisting}
type puzzle = /\ <{14}| &{
  update : id -> money -o[R] |{0}> \/ puzzle,
  submit : int ^ &{
    success : int -> money *[R] |{5}> \/ puzzle,
    failure : |{9}> \/ puzzle } }
\end{lstlisting}
% \begin{tabbing}
% $\m{puzzle} = \B{\up} \R{\getpot^{14}} \echoice{$
% \=$\mb{update} : \m{id} \arrow \m{money} \lolli \B{\down} \m{puzzle},$\\
% $\mb{submit} : \tint \product \echoice{$
% \=$\mb{failure} : \R{\paypot^{9}} \B{\down} \m{puzzle},$\\
% \>\hspace{-3em}$\mb{success} : \m{solution} \arrow \m{money} \tensor_{\p} \R{\paypot^{5}}
% \B{\down} \m{puzzle}}}$
% \end{tabbing}
The contract still supports the two transactions. To update the
reward, it receives the $\m{update}$ label and an identifier, verifies that the sender is
the owner, receives money from the sender, and acts like a puzzle again.
The transaction to submit a solution has a \emph{guard}
associated with it. First, the contract sends an integer corresponding
to the reward amount, the user then verifies that the reward matches
the expected reward (the guard condition). If this check succeeds,
the user sends the $\m{success}$ label, followed by the solution,
receives the winnings, and the session terminates. If the guard
fails, the user issues the $\m{failure}$ label and immediately
terminates the session. Thus, the contract implementation guarantees
that the user submitting the solution receives their expected winnings.

\paragraph{\textbf{\textit{Voting}}}
The voting contract provides a $\m{ballot}$ type.
\begin{lstlisting}
type ballot = /\ <{22}| +{
  open : id -> +{ vote : id -> |{0}> \/ ballot,
                    novote : |{14}> \/ ballot },
  closed : id ^ |{19}> \/ ballot }
\end{lstlisting}
% \begin{tabbing}
% $\ballot = \B{\up} \R{\getpot^{22}} \ichoice{$\=$\mb{open} : \m{id} \arrow
% \ichoice{$\=$\mb{vote} : \m{id} \arrow \B{\down} \ballot,$\\
% \>\>$\mb{novote} : \R{\paypot^{14}} \B{\down} \ballot},$\\
% \>$\mb{closed} : \m{id} \product \R{\paypot^{19}} \B{\down} \ballot}$
% \end{tabbing}
This contract allows voting when the election is $\mb{open}$ by
sending the candidate's $\mi{id}$,
and prevents double voting by checking if the voter has already voted
(the $\mb{novote}$
label). Once the election closes, the contract can be acquired to
check the winner. We use two implementations for the contract:
the first (voting in Table~\ref{tab:study}) stores a counter for each candidate that is updated after
each vote is cast; the second (voting-aa in Table~\ref{tab:study}) does not use a counter but stores potential
inside the vote list that is consumed for counting the votes at the end.
This stored potential is provided by the voter to amortize the cost of counting.
% \jan{what's the int after vote?}  \jan{add this
% amortization idea for ballot counting. However, isn't it a bit
% contrived as you could count keep a counter? Maybe you could use a
% linear type to keep track of votes instead of integers? That could
% be used as a motivation for counting at the end, maybe.}

\paragraph{\textbf{\textit{Escrow}}}
A contract can act as a reliable third party for custody of
a bond that takes effect once both the buyer and the seller approve.
\begin{lstlisting}
type escrow = /\ <{7}| &{
  approve : id -> |{0}> \/ escrow,
  cancel : id -> |{0}> \/ escrow,
  deposit : id -> bond -o[R] |{4}> \/ escrow,
  withdraw : id -> bond *[R] |{0}> \/ escrow }
\end{lstlisting}
% \begin{tabbing}
% $\m{escrow} = \B{\up} \R{\getpot^{7}} \echoice{$
% \=$\mb{approve} : \m{id} \arrow \B{\down} \m{escrow},$\\
% \>$\mb{cancel} : \m{id} \arrow \B{\down} \m{escrow},$\\
% \>$\mb{deposit} : \m{id} \lolli \m{bond} \lolli_{\p} \R{\paypot^{4}}
% \B{\down} \m{escrow},$\\
% \>$\mb{withdraw} : \m{id} \lolli \m{bond} \tensor_{\p}
% \B{\down} \m{escrow}}$
% \end{tabbing}
This session type describes the implementation of an escrow, allowing the seller to
deposit the bond, the buyer to withdraw the bond, and both the
buyer and seller to approve or cancel the whole transaction. The
withdrawal succeeds only after the bond has been deposited, and
both the buyer and seller approve it.
% \jan{I don't understand this protocol.}

\paragraph{\textbf{\textit{Experimental Evaluation}}}
We implemented 8 case studies in Nomos. 
We have already discussed auction (Section~\ref{sec:overview}), ERC
20, puzzle, and voting. The other case studies are:
\begin{itemize}[leftmargin=*]
  \item A bank account that allows users to register, make deposits and withdrawals
  and check the balance.
  \item An escrow to exchange bonds between two parties.
  \item A wallet allowing users to store money on the blockchain.
  \item An insurance contract that processes flight delay insurance claims
  after verifying them with a trusted third party.
  This contract involves inter-contract communication since the insurance
  and the third-party verifier are implemented as separate contracts.
\end{itemize}
Table~\ref{tab:study} contains a compilation of our experiments with
the case studies and the prototype implementation. The experiments
were run on an Intel Core i5 2.7 GHz processor with 16 GB 1867 MHz DDR3
memory.
It presents the
contract name, its lines of code (LOC), the type checking time (T
(ms)), number of potential and mode variables introduced (Vars),
number of potential and mode constraints that were generated while
type checking (Cons) and the time the LP solver took to infer their
values (I (ms)).
The last column describes the maximal gap between the
static gas bound inferred and the actual runtime
gas cost. It accounts for the difference in the gas cost in different
program paths. However, this waste is clearly marked in the program by
explicit $\mi{tick}$
instructions so the programmer is aware of this runtime gap, based on
the program path executed.

\begin{table}[t]
\centering
\begin{tabular}{@{}l r r r r r r}
\textbf{Contract} & \textbf{LOC} & \textbf{T (ms)} &
\textbf{Vars} & \textbf{Cons} &
\textbf{I (ms)} & \textbf{Gap} \\
\midrule
auction & 176 & 0.558 & 229 & 730 & 5.225 & 3 \\
ERC 20 & 136 & 0.579 & 161 & 561 & 4.317 & 6 \\
puzzle & 108 & 0.410 & 126 & 389 & 8.994 & 8 \\
voting & 101 & 0.324 & 109 & 351 & 3.664 & 0 \\
voting-aa & 101 & 0.346 & 140 & 457 & 3.926 & 0 \\
escrow & 85 & 0.404 & 95 & 321 & 3.816 & 3 \\
insurance & 56 & 0.299 & 76 & 224 & 8.289 & 0 \\
bank & 147 & 0.663 & 173 & 561 & 4.549 & 0 \\
wallet & 30 & 0.231 & 32 & 102 & 3.224 & 0 \\
\midrule
\end{tabular}
\caption{Evaluation of Nomos with Case Studies.
   LOC = lines of code; 
   T (ms) =  the type checking time in ms;
   Vars = \#variables generated during type inference; 
   Cons = \#constraints generated during type inference;
   I (ms) = type inference time in ms;
   Gap =  maximal gas bound gap.}
\label{tab:study}
\vspace{-2em}
\end{table}
%%% Local Variables:
%%% mode: latex
%%% TeX-master: "la"
%%% End:

The evaluation shows that the type-checking overhead is less than a
millisecond for case studies. This indicates that Nomos is applicable
to settings like distributed blockchains in which type checking could
add significant overhead and could be part of the attack surface. Type
inference is also efficient but an order of magnitude slower than type
checking. This is acceptable since inference is only performed once
during deployment of the contract. Gas bounds are tight in most
cases. Loose gas bounds are caused by conditional branches with different
gas cost. In practice, this is not a major concern since the Nomos semantics
tracks the exact gas cost, and a user will not be overcharged for their
transaction. However, Nomos' type system can be easily modified to only
allow contracts with tight bounds.

Our implementation experience revealed that describing the session type
of a contract crystallizes the important aspects of its protocol. Once
the type is defined, the implementation simply \emph{follows} the type
protocol. The error messages from the type checker were helpful in ensuring
linearity of assets, and using $*$ for potential annotations meant
we could remain unaware of the exact gas cost of operations.

%%% Local Variables:
%%% mode: latex
%%% mode: flyspell
%%% TeX-master: "pldi20.tex"
%%% End:

\section{Blockchain Integration}
\label{sec:model}

%\jan{.7 pages\\
%- transactions (equi-sync)\\
%- client code\\
%- blockchain\\
%- coins cannot be generated\\
%- name generation\\
%- funs. in CIC (couldn't make sense of this point when transcribing)
%}
Although Nomos has been designed to be applicable for implementing
general digital contracts, the standard semantics needs some
adaptation for a contract to be run on a blockchain. To integrate with a blockchain,
we need a mechanism to \emph{(i)} represent the contracts and their
addresses in the current blockchain state, \emph{(ii)} create and send
transactions to the appropriate addresses, and most importantly,
\emph{(iii)} construct the global distributed ledger, which stores
the history of all transactions. This section addresses these challenges
and also highlights the main limitation
of the language.

\paragraph{\textbf{\textit{Nomos on a Blockchain}}}
To describe a possible blockchain
implementation of Nomos, we assume a blockchain like Ethereum that
contains a set of Nomos contracts $C_1,\ldots,C_n$ together with their
type information
$\Psi^i \semi \G^i \semi \D^i_\p \entailpot{q_i} C_i :: (x^i_\s :
A^i_\s)$. The functional contexts $\Psi^i$ type the contract data, while
the shared contexts $\G^i$ type the shared contracts that $C_i$
refer to, and the linear contexts $\D^i_\p$ type the contract's
linear assets. We allow contracts to carry potential given by the
annotations $q_i$ and the potential defined by the annotations in
$\Psi^i$ and $\D^i_\p$. This potential is useful to amortize gas cost
over different transactions. If this behavior is not desired then 
one can require $q_i = 0$ for every $i$. Together, these contracts define the
blockchain state.
% Together, these contracts form a configuration typed as
% \[
%   \Sg \semi \G \potconf{E} \proc{x^1_\s}{w_1, C_1} \ldots
% \proc{x^n_\s}{w_n, C_n} :: (\G \semi \cdot)
% \]
% where $\G = (x^1_\s : A^1_\s), \ldots, (x^n_\s : A^n_\s)$ denotes the
% set of all contract channels while $E = \Sg_{i=1}^{n} q_i + w_i$ is the total
% energy stored in the contracts.
The channel name $x^i_\s$ of
a contract is its address and has to be globally unique. We assume the existence
of a deterministic mechanism that produces fresh names.
%\jan{should there be some requirement on $\Gamma_i$ reflecting the order of contract creation?}
%\ankush{Don't think so. The transaction order is maintained in the bc-server,
% creation order can be inferred from that. }

%\jan{what about the deterministic execution issue?}

To perform a transaction with a contract, an external user
submits a script that is well-typed with respect to the existing
contracts using the judgment
\begin{center}
\begin{minipage}{0cm}
\begin{tabbing}
$\Psi \semi \G \semi \cdot \entailpot{q} Q :: (x_\c : \one)$
\end{tabbing}
\end{minipage}
\end{center}
Here, $\G \subseteq x^1_\s : A^1_\s, \ldots, x^n_\s : A^n_\s$ stores
references to the Nomos contracts accessible by the transaction.
$\Psi$ stores the functional part of the script, and since the script cannot
refer to linear data, its linear context is empty. Additionally, we
mandate that the transaction offers along a channel of type $\one$,
and that it terminates by sending a $\m{close}$ message on its offered
channel. For instance, the transaction $Q$ must end with the operation
$(\eclose{x_\c})$. This ensures that transactions are sequentialized
and executed in the order they are queued (explained below).

A transaction script is connected to the blockchain state using
a server process. This process, named $\m{bc{-}server}$ stores
the entire transaction history and offers along channel $bc : 
\m{tx\_interface}$ where the transaction code is received and
relayed to the blockchain state. It is defined as follows.
\begin{ntabbing}
  \reset
  $\m{type \; tx\_code} = \{ \one \}$ \hspace{1.5em}
  $\m{type \; tx\_queue} = \m{list \; tx\_code}$ \label{server_types}\\
  $\m{stype \; tx\_interface} = \m{tx\_code} \arrow \m{tx\_interface}$
  \label{server_stypes}\\
  $(txns : \m{tx\_queue}) \semi \cdot \semi \cdot \entailpot{0}
  \m{bc{-}server} :: (bc : \m{tx\_interface})$ \label{server_decl}\\
  \quad\=$\procdefna{\m{bc{-}server} \; txns}{bc} =$ \label{server_def}\\
  \>\quad\=$\erecvch{bc}{tx} \semi \procdefna{tx}{x_\c} \semi
  \ewait{x_\c} \semi$ \label{server_recv}\\
  \>\>$\procdefna{\m{bc{-}server} \; (tx :: txns)}{bc}$ \label{server_recurse}
\end{ntabbing}
%\jan{Discuss spawning of contracts and leaving the chain in a clean state
% after the transaction.}
The transaction script is packaged as a value of the contextual monadic
type introduced in Section~\ref{sec:fun}. For instance, the transaction
$Q$ is packaged as $\{x_\c \leftarrow Q\} : \{ \one \} = \m{tx\_code}$.
The $\m{bc{-}server}$ process receives this code, spawns a process
corresponding to it and waits for the transaction to terminate (line
\ref{server_recv}). Note that the transaction is required to terminate
with a $(\eclose{x_\c})$ message which matches with the $(\ewait{x_\c})$
being executed by the server, ensuring the execution order of the
transactions. Finally, the latest transaction is added to the queue
of transactions $txns : \m{type \; tx\_queue} = \m{list \; tx\_code}$,
and the $\m{bc{-}server}$ process recurses.

A transaction can either update the state of existing contracts, or create
new ones. In the former case, it \emph{acquires} the contracts it wishes
to interact with, followed by an update in the contracts' internal state
and \emph{releases} them. Since
the contract types are equi-synchronizing, they remain unchanged at
the end of transaction execution. This ensures that the subsequent
transactions can access the same contracts at the same type. In the future
we plan to allow \emph{sub-synchronizing} types that enable
a client to release a contract channel not at the same type, but a
\emph{subtype}. The subtype can then describe the phase of
the contract. For instance, the ended phase of auction contract
will be a subtype of the running phase. In the latter case, new contracts
are added to the blockchain state, making
them visible in the type of the configuration for subsequent
transactions to access. Thus, in either case, the blockchain state
remains well-formed between transactions.
A successful execution of a transaction will lead to the $\m{bc{-}server}$
process recursing and accepting further transactions.

Concurrent execution of transactions is missing from blockchain systems
today~\cite{HerlihyCACM19}. To reconstruct the blockchain state, each
miner must re-execute every transaction sequentially; simply
executing them in parallel is unsafe when contracts depend on each
other. However, Nomos naturally has a concurrent semantics, and
we can support concurrent transactions with a slight modification
to the $\m{tx\_interface}$ type. One caveat is that we need to
ensure deterministic execution of a transaction. The only source of
non-determinism in the Nomos semantics is the \emph{acquire-accept} pair.
A contract executing an $\m{accept}$ can attach with any process that tries
to acquire it. One approach to resolve this non-determinism is
\emph{record-and-replay}~\cite{Ronsse-RecPlay99, LidburyPLDI19}.
The miner records the order in which the contracts are acquired
in the ledger, which is then replayed by others to compute
the current blockchain state. Another promising approach
is \emph{speculation}~\cite{DickersonPODC17} where transactions are
executed in parallel and their read and write sets are tracked.
If there is a conflict in these sets, then they are sequentialized
and this schedule is repeated by validators.
This speculative technique is known to provide speed-ups to the
overall throughput of the blockchain system~\cite{Speculative19}.

When selecting a request, a miner first creates a configuration,
and then type checks the transaction script $Q$ against
its submitted type information and the existing types of the
contracts $C_1, \ldots, C_n$ and the server process. If type
checking were too costly here, that can lead to yet another source of
denial-of-service attacks. In Nomos however, since the type of
transaction script is provided by the programmer, this form
of bi-directional type checking is linear time
in the size of the script.
%\jan{somewhere: discuss the cost of type checking; it should be linear time}
The gas cost of the transaction is statically bounded by
the potential given by $q$ and $\Psi$.
If we allow amortization then the potential in the contracts $C_i$'s is
also available to cover the gas cost.
This internal potential is not available to the user but can only
be accessed according to the protocol that is given in the contract
session type.

\paragraph{\textbf{\textit{Miner's Transaction Fee}}}
Mining rewards in blockchains like Ethereum are realized by special
transactions that transfer coins to the miner at the beginning of a
block. In Nomos, such a transaction could, for example, be represented
by an interaction with a special \emph{mining reward contract} that
sends linear coins to every client who requests them. Like in
Ethereum, a block with transactions is only valid if only the first
transaction interacts with the reward contract. This can be ensured by
the miner with a dynamic check or statically by removing the reward
contract from the list of available contracts before executing user
transactions.

\paragraph{\textbf{\textit{Deadlocks}}}

The only language specific reason a transaction can fail is a
deadlock in the transaction code. Our progress theorem accounts for the
possibility of deadlocks. Deadlocks may arise due to cyclic
interdependencies on the contracts that a transaction attempts to acquire.
While it is of course desirable to rule out deadlocks, we felt that this is
orthogonal to the design of Nomos.  Any extensions for shared session
types that prevent deadlocks (e.g., \cite{Balzer-ESOP19}) will be
readily transferable to our setting. Another possibility is to employ
dynamic deadlock detection~\cite{MitchellDeadlock84, ChandyDeadlock83}
and abort the transaction if a deadlock is detected.

%\paragraph{\textbf{\textit{Surface Syntax and Client Code}}}
%
%In this paper, we did not focus on the usability of Nomos. However, we
%do not neglect this point and plan to work with the blockchain
%community to develop a more intuitive surface syntax. One point that
%we would like to make is that we do not expect users to write a new
%client process for every interaction she wants to have with a
%contract. We rather envision that a contract developer would create a
%contract together with several boilerplate clients that a user would
%then instantiate with the fitting arguments; for example the bid for
%an auction.

%%% Local Variables:
%%% mode: latex
%%% mode: flyspell
%%% TeX-master: "icfp19"
%%% End:

\section{Related Work}
\label{sec:related}
We classify the related work into 3 categories - i) new programming
languages for smart contracts, ii) static
analysis techniques for existing languages and bytecode, and
iii) session-typed and type-based resource analysis systems
technically related to Nomos.

\paragraph{\textbf{\textit{Smart Contract Languages}}}
Existing smart contracts on Ethereum are predominantly implemented in
Solidity~\cite{Auction-Solidity}, a statically typed object-oriented language
influenced by Python and Javascript. However,
the language provides no information about the resource usage of a
contract. Languages like Vyper~\cite{Vyper} address resource usage by
disallowing recursion and infinite-length loops, thus making estimation of
gas usage decidable. However, both languages still suffer from re-entrancy
vulnerabilities. Bamboo~\cite{Bamboo}, on the other hand, makes state
transitions explicit and avoids re-entrance by design. In contrast to our
work, none of these languages use linear type systems to track assets
stored in a contract.

Domain specific languages have also been designed for other blockchains
apart from Ethereum. Typecoin~\cite{Typecoin-PLDI15} uses affine logic
to solve the peer-to-peer affine commitment problem using a
generalization of Bitcoin where transactions deal in types rather than
numbers. Although Typecoin does not provide a mechanism for
expressing protocols, it also uses a linear type system to prevent
resources from being discarded or duplicated. Rholang~\cite{Rholang}
is formally modeled by the
$\rho$-calculus, a reflective higher-order extension of the $\pi$-calculus.
Michelson~\cite{Michelson} is a purely functional stack-based language
that has no side effects. Scilla~\cite{Scilla}
is an intermediate-level language where contracts are structured as
communicating automata providing a continuation-passing style
computational model to the language semantics. However, none of
these languages describe and enforce communication protocols statically.

\paragraph{\textbf{\textit{Static Analysis}}}
Analysis of smart contracts has received substantial
attention recently due to their security vulnerabilities that can be
exploited by malicious users. KEVM~\cite{kevm-CSF18}
creates a program verifier based on reachability logic
that given an EVM program and specification, tries to automatically
prove the corresponding reachability theorems. However, the verifier
requires significant manual intervention, both in specification and proof
construction. Oyente~\cite{Oyente-CCS16} is a symbolic execution
tool that checks for 4 kinds of security bugs in smart contracts,
transaction-order dependence, timestamp dependence, mishandled
exceptions and re-entrancy vulnerabilities. MadMax
\cite{MadMax-OOPSLA18} automatically detects gas-focused
vulnerabilities with high confidence. The analysis is based on a decompiler
that extracts control and data flow information from EVM bytecode,
and a logic-based analysis specification that produces a high-level
program model. \citet{BhargavanF*} translate Ethereum contracts to F*
to prove runtime safety and functional correctness, although they do
not support all syntactic features. \textsc{VeriSol}~\cite{Verisol}
is a highly-automated formal verifier for Solidity that can produce
proofs as well as counterexamples and proves semantic conformance of
smart contracts against a state machine model with access-control policy.
However, in contrast to Nomos, where guarantees
are proved by a soundness proof of the type system, static analysis
techniques often do not explore all program paths,
can report false positives that need to be manually filtered, and
miss bugs due to timeouts and other sources of incompleteness.

\paragraph{\textbf{\textit{Session types and Resource analysis}}}
Session types were introduced by Honda~\cite{HondaCONCUR1993}
as a typed formalism for inter-process dyadic interaction.
They have been integrated into a functional language in prior work
\cite{ToninhoESOP2013}. However, this integration does not account for
resource usage or sharing. Sharing in session types has also been explored
in prior work~\cite{BalzerICFP2017}, but with the strong restriction that
shared processes cannot rely on linear resources that we lift in Nomos.
Shared session types were also never integrated with a functional layer
or tracked for resource usage. While we consider binary session types
that express local interactions, global protocols can be expressed using
multi-party session types~\cite{HondaPOPL2008, Alceste-POPL19}.
Automatic amortized resource analysis (AARA) has been introduced
as a type system to derive linear~\cite{Jost03} and polynomial
bounds~\cite{HoffmannW15} for functional programming languages.
Resource usage has also previously been explored separately for the
purely linear process layer~\cite{DasHP17}, but were never combined
with shared session types or integrated with the functional layer.

%%% Local Variables:
%%% mode: latex
%%% mode: flyspell
%%% TeX-master: "pldi19"
%%% End:

\section{Conclusion}
\label{sec:conclusion}
We have described the programming language Nomos, its
type-theoretic foundation, a prototype implementation and
evaluated its feasibility on several real world smart contract
applications. Nomos builds on linear logic, shared
session types, and automatic amortized resource analysis to
address the challenges that programmers are faced
with when implementing digital contracts. Our main contributions are the design
and implementation of Nomos' multi-layered resource-aware type system and its type
soundness proof.

In future work, we plan to explore refinement session types
for expressing and verifying functional correctness of contracts against
their specifications and to target open questions regarding a blockchain integration.
These include the exact cost model, fluctuation of gas prices, and potential
compilation to a lower-level language. Since Nomos has a concurrent semantics, we
also plan to support parallel execution of transactions using speculation techniques
\cite{Speculative19}.

\appendix
\section{Overview}
This appendix supplements the tech report ``Resource-Aware Session Types
for Digital Contracts''. The main contributions of the appendix
are as follows.

\begin{itemize}
\item Appendix~\ref{sec:code} presents the Nomos code for standard smart contract
applications.

\item Appendix~\ref{sec:grammar} presents the type grammar.

\item Appendix~\ref{sec:proc-typing} presents the process typing rules, concerning
the judgment $\Psi \semi \G \semi \D \entailpot{q} P :: (x_m : A)$.
This judgment types a process in state $P$ providing service of type $A$
along channel $x$ at mode $m$. Moreover, the process uses functional
variables from $\Psi$, shared channels from $\G$ and linear channels from
$\D$. Finally, the process stores potential $q$.

\item Appendix~\ref{sec:semantics} presents the rules of the operational
cost semantics. These discuss the behavior of the semantic objects
$\proc{c_m}{w, P}$ and $\msg{c_m}{w, N}$ defining a process $P$
(or message $N$) offering along channel $c$ at mode $m$ which has
performed work $w$ so far.

\item Appendix~\ref{sec:conf-typing} presents the rules corresponding to
configuration typing and other helper judgments. The configuration
typing judgment $\G_0 \potconf{E} \W :: (\G \semi \D)$ describes a
well-typed configuration $\W$ which offers shared channels in $\G$
and linear channels in $\D$.

\item Appendix~\ref{sec:type-safety} is the main contribution of the
supplementary material. It presents and proves the main theorem
of type safety of our language. This is split into a type preservation
and a progress theorem. The appendix also proves the lemmas
necessary for the type safety theorems.

\end{itemize}

\section{Implementation of Smart Contract Applications in Nomos}\label{sec:code}

\subsection{Auction}
\lstset{basicstyle=\ttfamily\footnotesize}
\begin{lstlisting}
type money = &{ value : <{2}| int ^ money,
                coins : <{0}| lcoin }
type lcoin = 1
proc asset emp : . |{1}- ($l[R] : lcoin) = 
  {
    work ;
    close $l[R]
  }
proc asset empty_wallet : . |{3}- ($m[R] : money) = 
  {
    $l[R] <- emp <-  ;
    work ;
    let n = (tick  ; 0) ;
    $m[R] <- wallet <- n $l[R]
  }
proc asset wallet : (n : int), ($l[R] : lcoin) |- ($m[R] : money) = 
  {
    case $m[R] ( value => get $m[R] {2};
                          work ;
                          send $m[R] ((tick  ; n)) ;
                          $m[R] <- wallet <- n $l[R]
               | coins => get $m[R] {0};
                          $m[R] <- $l[R] )
  }
type dictionary = &{ add : <{5}| int -> money -o[R] dictionary,
                     delete : <{6}| int -> money *[R] dictionary,
                     check : <{4}| int -> bool ^ dictionary,
                     size : <{2}| int ^ dictionary }
proc asset dummy : (n : int) |- ($d[R] : dictionary) = 
  {
    case $d[R] ( add => get $d[R] {5};
                        key = recv $d[R] ;
                        work ;
                        $v[R] <- recv $d[R] ;
                        $v[R].coins ;
                        pay $v[R] {0};
                        work ;
                        wait $v[R] ;
                        let n = (tick  ; (tick  ; n) + (tick  ; 1)) ;
                        $d[R] <- dummy <- n
               | delete => get $d[R] {6};
                           key = recv $d[R] ;
                           $v[R] <- empty_wallet <-  ;
                           send $d[R] $v[R] ;
                           let n = (tick  ; (tick  ; n) - (tick  ; 1)) ;
                           $d[R] <- dummy <- n
               | check => get $d[R] {4};
                          key = recv $d[R] ;
                          if (tick  ; (tick  ; key) > (tick  ; 0))
                          then
                            send $d[R] ((tick  ; true)) ;
                            $d[R] <- dummy <- n
                          else
                            send $d[R] ((tick  ; false)) ;
                            $d[R] <- dummy <- n
               | size => get $d[R] {2};
                         work ;
                         send $d[R] ((tick  ; n)) ;
                         $d[R] <- dummy <- n )
  }
type lot = 1
proc asset addbid : (r : int), ($m[R] : money), ($bs[R] : dictionary)
                              |{8}- ($newbs[R] : dictionary) = 
  {
    work ;
    $bs[R].add ;
    pay $bs[R] {5};
    work ;
    send $bs[R] ((tick  ; r)) ;
    send $bs[R] $m[R] ;
    $newbs[R] <- $bs[R]
  }
type auction = /\ <{22}|
  +{ running : &{ bid : int -> money -o[R] |{0}> \/ auction,
                  cancel : |{21}> \/ auction },
     ended : &{ collect : int -> +{ won : lot *[R] |{0}> \/ auction,
                                    lost : money *[R] |{7}> \/ auction },
                cancel : |{21}> \/ auction } }
proc contract run : (T : int), (w : int), (v : int),
                    ($b[R] : dictionary), ($l[R] : lot)
                                              |- (#sa[S] : auction) = 
  {
    $la[L] <- accept #sa[S] ;
    get $la[L] {22};
    work ;
    $la[L].running ;
    case $la[L] ( bid => r = recv $la[L] ;
                         work ;
                         $m[R] <- recv $la[L] ;
                         pay $la[L] {0};
                         #sa[S] <- detach $la[L] ;
                         $m[R].value ;
                         pay $m[R] {2};
                         work ;
                         bv = recv $m[R] ;
                         $newb[R] <- addbid <- r $m[R] $b[R] ;
                         if (tick  ; (tick  ; bv) > (tick  ; v))
                         then
                           #sa[S] <- check <- T r bv $newb[R] $l[R]
                         else
                           #sa[S] <- check <- T w v $newb[R] $l[R]
                | cancel => pay $la[L] {21};
                            #sa[S] <- detach $la[L] ;
                            #sa[S] <- run <- T w v $b[R] $l[R] )
  }
proc contract check : (T : int), (w : int), (v : int),
                      ($b[R] : dictionary), ($l[R] : lot)
                                          |{6}- (#sa[S] : auction) = 
  {
    work ;
    $b[R].size ;
    pay $b[R] {2};
    n = recv $b[R] ;
    if (tick  ; (tick  ; n) = (tick  ; T))
    then
      #sa[S] <- end_lot <- T w $b[R] $l[R]
    else
      #sa[S] <- run <- T w v $b[R] $l[R]
  }
proc asset removebid : (r : int), ($bs[R] : dictionary)
                            |{10}- ($newbs[R] : money *[R] dictionary) = 
  {
    work ;
    $bs[R].delete ;
    pay $bs[R] {6};
    work ;
    send $bs[R] ((tick  ; r)) ;
    work ;
    $m[R] <- recv $bs[R] ;
    send $newbs[R] $m[R] ;
    $newbs[R] <- $bs[R]
  }
proc contract end_lot : (T : int), (w : int),
                        ($b[R] : dictionary), ($l[R] : lot)
                                              |- (#sa[S] : auction) = 
  {
    $la[L] <- accept #sa[S] ;
    get $la[L] {22};
    work ;
    $la[L].ended ;
    case $la[L] ( collect => r = recv $la[L] ;
                             if (tick  ; (tick  ; w) = (tick  ; r))
                             then
                               $la[L].won ;
                               send $la[L] $l[R] ;
                               pay $la[L] {0};
                               #sa[S] <- detach $la[L] ;
                               #sa[S] <- end_nolot <- T w $b[R]
                             else
                               $la[L].lost ;
                               $newb[R] <- removebid <- r $b[R] ;
                               work ;
                               $m[R] <- recv $newb[R] ;
                               send $la[L] $m[R] ;
                               pay $la[L] {7};
                               #sa[S] <- detach $la[L] ;
                               #sa[S] <- end_lot <- T w $newb[R] $l[R]
                | cancel => pay $la[L] {21};
                            #sa[S] <- detach $la[L] ;
                            #sa[S] <- end_lot <- T w $b[R] $l[R] )
  }
proc contract end_nolot : (T : int), (w : int), ($b[R] : dictionary)
                                                |{18}- (#sa[S] : auction) = 
  {
    $la[L] <- accept #sa[S] ;
    get $la[L] {22};
    work ;
    $la[L].ended ;
    case $la[L] ( collect => r = recv $la[L] ;
                             $la[L].lost ;
                             $newb[R] <- removebid <- r $b[R] ;
                             work ;
                             $m[R] <- recv $newb[R] ;
                             send $la[L] $m[R] ;
                             pay $la[L] {7};
                             #sa[S] <- detach $la[L] ;
                             work {3};
                             #sa[S] <- end_nolot <- T w $newb[R]
                | cancel => pay $la[L] {21};
                            #sa[S] <- detach $la[L] ;
                            work {0};
                            #sa[S] <- end_nolot <- T w $b[R] )
  }
TC time: 2.0809173584
Inference time: 9.94420051575
# Vars = 229
# Constraints = 730
% compilation successful!
% runtime successful!
\end{lstlisting}

\subsection{Bank Account}

\begin{lstlisting}
type money = &{ value : <{2}| int ^ money,
                coins : <{0}| lcoin,
                check_pwd : <{4}| int -> bool ^ money }
type lcoin = 1
proc asset emp : . |{1}- ($l[R] : lcoin) = 
  {
    work ;
    close $l[R]
  }
proc asset empty_wallet : (pwd : int) |{3}- ($m[R] : money) = 
  {
    $l[R] <- emp <-  ;
    work ;
    let n = (tick  ; 0) ;
    $m[R] <- wallet <- pwd n $l[R]
  }
proc asset wallet : (pwd : int), (n : int), ($l[R] : lcoin) |- ($m[R] : money) = 
  {
    case $m[R] ( value => get $m[R] {2};
                          work ;
                          send $m[R] ((tick  ; n)) ;
                          $m[R] <- wallet <- pwd n $l[R]
               | coins => get $m[R] {0};
                          $m[R] <- $l[R]
               | check_pwd => get $m[R] {4};
                              p = recv $m[R] ;
                              if (tick  ; (tick  ; p) = (tick  ; pwd))
                              then
                                send $m[R] ((tick  ; true)) ;
                                $m[R] <- wallet <- pwd n $l[R]
                              else
                                send $m[R] ((tick  ; false)) ;
                                $m[R] <- wallet <- pwd n $l[R] )
  }
type dictionary = &{ add : <{5}| int -> money -o[R] dictionary,
                     delete : <{6}| int -> money *[R] dictionary,
                     check : <{5}| int -> int -> bool ^ dictionary,
                     size : <{2}| int ^ dictionary }
proc asset dummy : (n : int) |- ($d[R] : dictionary) = 
  {
    case $d[R] ( add => get $d[R] {5};
                        key = recv $d[R] ;
                        work ;
                        $v[R] <- recv $d[R] ;
                        $v[R].coins ;
                        pay $v[R] {0};
                        work ;
                        wait $v[R] ;
                        let n = (tick  ; (tick  ; n) + (tick  ; 1)) ;
                        $d[R] <- dummy <- n
               | delete => get $d[R] {6};
                           key = recv $d[R] ;
                           $v[R] <- empty_wallet <- key ;
                           send $d[R] $v[R] ;
                           let n = (tick  ; (tick  ; n) - (tick  ; 1)) ;
                           $d[R] <- dummy <- n
               | check => get $d[R] {5};
                          key = recv $d[R] ;
                          work ;
                          pwd = recv $d[R] ;
                          if (tick  ; (tick  ; pwd) > (tick  ; 0))
                          then
                            send $d[R] ((tick  ; true)) ;
                            $d[R] <- dummy <- n
                          else
                            send $d[R] ((tick  ; false)) ;
                            $d[R] <- dummy <- n
               | size => get $d[R] {2};
                         work ;
                         send $d[R] ((tick  ; n)) ;
                         $d[R] <- dummy <- n )
  }
type account = /\ <{29}|
    &{ signup : int -> int -> |{19}> \/ account,
       login : int -> int ->
          +{ failure : |{19}> \/ account,
             success : &{ deposit : money -o[R] |{11}> \/ account,
                          balance : int ^ |{0}> \/ account,
                          withdraw : int -> money *[R] |{9}> \/ account } } }
proc contract bank : ($accts[R] : dictionary) |- (#sa[S] : account) = 
  {
    $la[L] <- accept #sa[S] ;
    get $la[L] {29};
    case $la[L] ( signup => id = recv $la[L] ;
                            work ;
                            pwd = recv $la[L] ;
                            $m[R] <- empty_wallet <- pwd ;
                            $accts[R].add ;
                            pay $accts[R] {5};
                            send $accts[R] ((tick  ; id)) ;
                            send $accts[R] $m[R] ;
                            pay $la[L] {19};
                            #sa[S] <- detach $la[L] ;
                            #sa[S] <- bank <- $accts[R]
                | login => id = recv $la[L] ;
                           work ;
                           pwd = recv $la[L] ;
                           $accts[R].check ;
                           pay $accts[R] {5};
                           send $accts[R] ((tick  ; id)) ;
                           send $accts[R] ((tick  ; pwd)) ;
                           work ;
                           r = recv $accts[R] ;
                           if (tick  ; r)
                           then
                             $la[L].success ;
                             work ;
                             case $la[L]
                              ( deposit => work ;
                                           $m[R] <- recv $la[L] ;
                                           $accts[R].add ;
                                           pay $accts[R] {5};
                                           send $accts[R] ((tick  ; id)) ;
                                           send $accts[R] $m[R] ;
                                           pay $la[L] {11};
                                           #sa[S] <- detach $la[L] ;
                                           #sa[S] <- bank <- $accts[R]
                              | balance => $accts[R].delete ;
                                           pay $accts[R] {6};
                                           send $accts[R] ((tick  ; id)) ;
                                           work ;
                                           $m[R] <- recv $accts[R] ;
                                           $m[R].value ;
                                           pay $m[R] {2};
                                           work ;
                                           val = recv $m[R] ;
                                           send $la[L] ((tick  ; val)) ;
                                           $accts[R].add ;
                                           pay $accts[R] {5};
                                           send $accts[R] ((tick  ; id)) ;
                                           send $accts[R] $m[R] ;
                                           pay $la[L] {0};
                                           #sa[S] <- detach $la[L] ;
                                           #sa[S] <- bank <- $accts[R]
                              | withdraw => $accts[R].delete ;
                                            pay $accts[R] {6};
                                            send $accts[R] ((tick  ; id)) ;
                                            work ;
                                            v = recv $la[L] ;
                                            work ;
                                            $m[R] <- recv $accts[R] ;
                                            send $la[L] $m[R] ;
                                            pay $la[L] {9};
                                            #sa[S] <- detach $la[L] ;
                                            #sa[S] <- bank <- $accts[R] )
                           else
                             $la[L].failure ;
                             pay $la[L] {19};
                             #sa[S] <- detach $la[L] ;
                             #sa[S] <- bank <- $accts[R] )
  }
TC time: 0.648975372314
Inference time: 4.99391555786
# Vars = 173
# Constraints = 561
% compilation successful!
% runtime successful!
\end{lstlisting}

\subsection{ERC-20 Token}

\begin{lstlisting}
type money = &{ add : <{8}| money -o[R] money,
                subtract : <{6}| int -> +{ sufficient : money *[R] money,
                                           insufficient : money },
                value : <{2}| int ^ money,
                coins : <| 1 }
proc asset wallet : (n : int) |- ($m[R] : money) = 
  {
    case $m[R] ( add => get $m[R] {8};
                        $m1[R] <- recv $m[R] ;
                        work ;
                        $m1[R].value ;
                        pay $m1[R] {2};
                        n1 = recv $m1[R] ;
                        $m1[R].coins ;
                        pay $m1[R] ;
                        work ;
                        wait $m1[R] ;
                        let n = (tick  ; (tick  ; n) + (tick  ; n1)) ;
                        $m[R] <- wallet <- n
               | subtract => get $m[R] {6};
                             n1 = recv $m[R] ;
                             if (tick  ; (tick  ; n) > (tick  ; n1))
                             then
                               $m[R].sufficient ;
                               $m1[R] <- wallet <- n1 ;
                               send $m[R] $m1[R] ;
                               let n = (tick  ; (tick  ; n) - (tick  ; n1)) ;
                               work {0};
                               $m[R] <- wallet <- n
                             else
                               $m[R].insufficient ;
                               work {3};
                               $m[R] <- wallet <- n
               | value => get $m[R] {2};
                          work ;
                          send $m[R] ((tick  ; n)) ;
                          $m[R] <- wallet <- n
               | coins => get $m[R] ;
                          work ;
                          close $m[R] )
  }
type erc20token = /\ <{11}|
   &{ totalSupply : int ^ |{9}> \/ erc20token,
      balanceOf : int -> int ^ |{8}> \/ erc20token,
      transfer : int -> int -> int -> |{0}> \/ erc20token,
      transferFrom : int -> int -> int -> |{0}> \/ erc20token,
      approve : int -> int -> int -> |{6}> \/ erc20token,
      allowance : int -> int -> int ^ |{6}> \/ erc20token }
type balance_dict = &{ get_balance : int -> int ^ balance_dict,
                       transfer : int -> int -> int -> balance_dict }
type allowance_dict = &{ get_allowance : int -> int -> int ^ allowance_dict,
                         set_allowance : int -> int -> int -> allowance_dict }
proc contract erc20contract : ($allows[R] : allowance_dict),
                              ($bals[R] : balance_dict), (N : int)
                                          |- (#se[S] : erc20token) = 
  {
    $le[L] <- accept #se[S] ;
    get $le[L] {11};
    case $le[L] ( totalSupply => work ;
                                 send $le[L] ((tick  ; N)) ;
                                 pay $le[L] {9};
                                 #se[S] <- detach $le[L] ;
                                 #se[S] <- erc20contract <- $allows[R] $bals[R] N
                | balanceOf => addr = recv $le[L] ;
                               $bals[R].get_balance ;
                               send $bals[R] ((tick  ; addr)) ;
                               work ;
                               val = recv $bals[R] ;
                               send $le[L] ((tick  ; val)) ;
                               pay $le[L] {8};
                               #se[S] <- detach $le[L] ;
                               #se[S] <- erc20contract <- $allows[R] $bals[R] N
                | transfer => from = recv $le[L] ;
                              work ;
                              to = recv $le[L] ;
                              work ;
                              amt = recv $le[L] ;
                              $allows[R].get_allowance ;
                              send $allows[R] ((tick  ; from)) ;
                              send $allows[R] ((tick  ; to)) ;
                              work ;
                              allowance = recv $allows[R] ;
                              if (tick  ; (tick  ; amt) > (tick  ; allowance))
                              then
                                pay $le[L] {0};
                                #se[S] <- detach $le[L] ;
                                work {3};
                                #se[S] <- erc20contract <- $allows[R] $bals[R] N
                              else
                                $bals[R].transfer ;
                                send $bals[R] ((tick  ; from)) ;
                                send $bals[R] ((tick  ; to)) ;
                                send $bals[R] ((tick  ; amt)) ;
                                pay $le[L] {0};
                                #se[S] <- detach $le[L] ;
                                work {0};
                                #se[S] <- erc20contract <- $allows[R] $bals[R] N
              | transferFrom => from = recv $le[L] ;
                                work ;
                                to = recv $le[L] ;
                                work ;
                                amt = recv $le[L] ;
                                $allows[R].get_allowance ;
                                send $allows[R] ((tick  ; from)) ;
                                send $allows[R] ((tick  ; to)) ;
                                work ;
                                allowance = recv $allows[R] ;
                                if (tick  ; (tick  ; amt) > (tick  ; allowance))
                                then
                                  pay $le[L] {0};
                                  #se[S] <- detach $le[L] ;
                                  work {3};
                                  #se[S] <- erc20contract <- $allows[R] $bals[R] N
                                else
                                  $bals[R].transfer ;
                                  send $bals[R] ((tick  ; from)) ;
                                  send $bals[R] ((tick  ; to)) ;
                                  send $bals[R] ((tick  ; amt)) ;
                                  pay $le[L] {0};
                                  #se[S] <- detach $le[L] ;
                                  work {0};
                                  #se[S] <- erc20contract <- $allows[R] $bals[R] N
                | approve => from = recv $le[L] ;
                             work ;
                             to = recv $le[L] ;
                             work ;
                             allowance = recv $le[L] ;
                             $allows[R].set_allowance ;
                             send $allows[R] ((tick  ; from)) ;
                             send $allows[R] ((tick  ; to)) ;
                             send $allows[R] ((tick  ; allowance)) ;
                             pay $le[L] {6};
                             #se[S] <- detach $le[L] ;
                             #se[S] <- erc20contract <- $allows[R] $bals[R] N
                | allowance => from = recv $le[L] ;
                               work ;
                               to = recv $le[L] ;
                               $allows[R].get_allowance ;
                               send $allows[R] ((tick  ; from)) ;
                               send $allows[R] ((tick  ; to)) ;
                               work ;
                               allowance = recv $allows[R] ;
                               send $le[L] ((tick  ; allowance)) ;
                               pay $le[L] {6};
                               #se[S] <- detach $le[L] ;
                               #se[S] <- erc20contract <- $allows[R] $bals[R] N )
  }
TC time: 1.89590454102
Inference time: 4.7709941864
# Vars = 161
# Constraints = 561
% compilation successful!
% runtime successful!
\end{lstlisting}

\subsection{Escrow}

\begin{lstlisting}
type escrow = /\ <{7}| &{ approve : int -> |{0}> \/ escrow,
                          cancel : int -> |{0}> \/ escrow,
                          deposit : int -> bond -o[R] |{4}> \/ escrow,
                          withdraw : int -> bond *[R] |{0}> \/ escrow }
type bond = 1
proc asset emp : . |{1}- ($l[R] : bond) = 
  {
    work ;
    close $l[R]
  }
proc contract escrow_con : (buyer : int), (seller : int),
                           (buyerOk : bool), (sellerOk : bool),
                           ($l[R] : bond)
                                        |- (#se[S] : escrow) = 
  {
    $le[L] <- accept #se[S] ;
    get $le[L] {7};
    case $le[L]
       ( approve =>
          r = recv $le[L] ;
          if (tick  ; (tick  ; r) = (tick  ; buyer))
          then
            let buyerOk = (tick  ; true) ;
            pay $le[L] {0};
            #se[S] <- detach $le[L] ;
            work {3};
            #se[S] <- escrow_con <- buyer seller buyerOk sellerOk $l[R]
          else
            if (tick  ; (tick  ; r) = (tick  ; seller))
            then
              let sellerOk = (tick  ; true) ;
              pay $le[L] {0};
              #se[S] <- detach $le[L] ;
              work {0};
              #se[S] <- escrow_con <- buyer seller buyerOk sellerOk $l[R]
            else
              pay $le[L] {0};
              #se[S] <- detach $le[L] ;
              work ;
              #se[S] <- escrow_con <- buyer seller buyerOk sellerOk $l[R]
      | cancel =>
          r = recv $le[L] ;
          if (tick  ; (tick  ; r) = (tick  ; buyer))
          then
            let buyerOk = (tick  ; false) ;
            pay $le[L] {0};
            #se[S] <- detach $le[L] ;
            work {3};
            #se[S] <- escrow_con <- buyer seller buyerOk sellerOk $l[R]
          else
            if (tick  ; (tick  ; r) = (tick  ; seller))
            then
              let sellerOk = (tick  ; false) ;
              pay $le[L] {0};
              #se[S] <- detach $le[L] ;
              work {0};
              #se[S] <- escrow_con <- buyer seller buyerOk sellerOk $l[R]
            else
              pay $le[L] {0};
              #se[S] <- detach $le[L] ;
              work ;
              #se[S] <- escrow_con <- buyer seller buyerOk sellerOk $l[R]
      | deposit =>
          r = recv $le[L] ;
          work ;
          $m[R] <- recv $le[L] ;
          let seller = (tick  ; r) ;
          pay $le[L] {4};
          #se[S] <- detach $le[L] ;
          work ;
          wait $m[R] ;
          work {0};
          #se[S] <- escrow_con <- buyer seller buyerOk sellerOk $l[R]
      | withdraw =>
          r = recv $le[L] ;
          if (tick  ; (tick  ; r) = (tick  ; buyer))
          then
            send $le[L] $l[R] ;
            $l[R] <- emp <-  ;
            pay $le[L] {0};
            #se[S] <- detach $le[L] ;
            work {3};
            #se[S] <- escrow_con <- buyer seller buyerOk sellerOk $l[R]
          else
            $m[R] <- emp <-  ;
            send $le[L] $m[R] ;
            pay $le[L] {0};
            #se[S] <- detach $le[L] ;
            work {3};
            #se[S] <- escrow_con <- buyer seller buyerOk sellerOk $l[R] )
  }
TC time: 1.9428730011
Inference time: 5.47099113464
# Vars = 95
# Constraints = 321
% compilation successful!
% runtime successful!
\end{lstlisting}

\subsection{Insurance}

\begin{lstlisting}
type insurance = /\ <{6}|
    &{ submit : int -> +{ success : money *[R] |{0}> \/ insurance,
       failure : |> \/ insurance } }
type verifier = /\ <{3}| &{ verify : int -> +{ valid : |{0}> \/ verifier,
                                               invalid : |{0}> \/ verifier } }
proc contract verify : . |- (#sv[S] : verifier) = 
  {
    $lv[L] <- accept #sv[S] ;
    get $lv[L] {3};
    case $lv[L] ( verify => claim = recv $lv[L] ;
                            if (tick  ; (tick  ; claim) > (tick  ; 0))
                            then
                              $lv[L].valid ;
                              pay $lv[L] {0};
                              #sv[S] <- detach $lv[L] ;
                              #sv[S] <- verify <- 
                            else
                              $lv[L].invalid ;
                              pay $lv[L] {0};
                              #sv[S] <- detach $lv[L] ;
                              #sv[S] <- verify <-  )
  }
type money = &{ subtract : money *[R] money }
proc contract insurer : (#sv[S] : verifier), ($m[R] : money)
                                          |- (#si[S] : insurance) = 
  {
    $li[L] <- accept #si[S] ;
    get $li[L] {6};
    case $li[L] ( submit => claim = recv $li[L] ;
                            $lv[L] <- acquire #sv[S] ;
                            pay $lv[L] {3};
                            $lv[L].verify ;
                            send $lv[L] ((tick  ; claim)) ;
                            work ;
                            case $lv[L]
                               ( valid => get $lv[L] {0};
                                          $li[L].success ;
                                          $m[R].subtract ;
                                          work ;
                                          $r[R] <- recv $m[R] ;
                                          send $li[L] $r[R] ;
                                          pay $li[L] {0};
                                          #sv[S] <- release $lv[L] ;
                                          #si[S] <- detach $li[L] ;
                                          #si[S] <- insurer <- #sv[S] $m[R]
                               | invalid => get $lv[L] {0};
                                            #li[L].failure ;
                                            pay $li[L] ;
                                            #sv[S] <- release $lv[L] ;
                                            #si[S] <- detach $li[L] ;
                                            #si[S] <- insurer <- #sv[S] $m[R] ) )
  }
TC time: 1.36709213257
Inference time: 3.58390808105
# Vars = 76
# Constraints = 224
% compilation successful!
% runtime successful!
\end{lstlisting}

\subsection{Puzzle}

\begin{lstlisting}
type puzzle = /\ <{14}|
    &{ update : int -> money -o[R] |{0}> \/ puzzle,
       submit : int ^ &{ success : int -> money *[R] |{5}> \/ puzzle,
                         failure : |{9}> \/ puzzle } }
type money = &{ value : <{2}| int ^ money,
                coins : <{0}| lcoin }
type lcoin = 1
proc asset join : ($m[R] : lcoin), ($n[R] : lcoin) |{1}- ($o[R] : lcoin) = 
  {
    wait $m[R] ;
    wait $n[R] ;
    work ;
    close $o[R]
  }
proc asset consume : ($m[R] : money) |{1}- ($o[R] : 1) = 
  {
    work ;
    $m[R].coins ;
    pay $m[R] {0};
    $o[R] <- $m[R]
  }
proc asset add : ($m[R] : money), ($n[R] : money) |{10}- ($o[R] : money) = 
  {
    work ;
    $m[R].value ;
    pay $m[R] {2};
    mval = recv $m[R] ;
    $n[R].value ;
    pay $n[R] {2};
    work ;
    nval = recv $n[R] ;
    let oval = (tick  ; (tick  ; mval) + (tick  ; nval)) ;
    $m[R].coins ;
    pay $m[R] {0};
    $n[R].coins ;
    pay $n[R] {0};
    $ocoin[R] <- join <- $m[R] $n[R] ;
    $o[R] <- wallet <- oval $ocoin[R]
  }
proc asset wallet : (n : int), ($l[R] : lcoin) |- ($m[R] : money) = 
  {
    case $m[R] ( value => get $m[R] {2};
                          work ;
                          send $m[R] ((tick  ; n)) ;
                          $m[R] <- wallet <- n $l[R]
               | coins => get $m[R] {0};
                          $m[R] <- $l[R] )
  }
proc asset emp : . |{1}- ($l[R] : lcoin) = 
  {
    work ;
    close $l[R]
  }
proc asset empty_wallet : . |{3}- ($m[R] : money) = 
  {
    $l[R] <- emp <-  ;
    work ;
    let n = (tick  ; 0) ;
    $m[R] <- wallet <- n $l[R]
  }
proc contract game : (addr : int), ($m[R] : money) |- (#sp[S] : puzzle) = 
  {
    $lp[L] <- accept #sp[S] ;
    get $lp[L] {14};
    case $lp[L] ( update => n = recv $lp[L] ;
                            work ;
                            $r[R] <- recv $lp[L] ;
                            if (tick  ; (tick  ; n) = (tick  ; addr))
                            then
                              $newm[R] <- add <- $m[R] $r[R] ;
                              pay $lp[L] {0};
                              #sp[S] <- detach $lp[L] ;
                              #sp[S] <- game <- addr $newm[R]
                            else
                              $tmp[R] <- consume <- $r[R] ;
                              work ;
                              wait $tmp[R] ;
                              pay $lp[L] {0};
                              #sp[S] <- detach $lp[L] ;
                              work {8};
                              #sp[S] <- game <- addr $m[R]
                | submit => work ;
                            $m[R].value ;
                            pay $m[R] {2};
                            mval = recv $m[R] ;
                            send $lp[L] ((tick  ; mval)) ;
                            work ;
                            case $lp[L]
                               ( success => work ;
                                            sol = recv $lp[L] ;
                                            send $lp[L] $m[R] ;
                                            pay $lp[L] {5};
                                            #sp[S] <- detach $lp[L] ;
                                            $emp[R] <- empty_wallet <-  ;
                                            #sp[S] <- game <- addr $emp[R]
                               | failure => pay $lp[L] {9};
                                            #sp[S] <- detach $lp[L] ;
                                            #sp[S] <- game <- addr $m[R] ) )
  }
TC time: 1.64389610291
Inference time: 4.714012146
# Vars = 126
# Constraints = 389
% compilation successful!
% runtime successful!
\end{lstlisting}

\subsection{Amortized Voting}

\begin{lstlisting}
type ballot = /\ <{16}| +{ open : int -> +{ vote : int -> |{0}> \/ ballot,
                                            novote : |{9}> \/ ballot },
                           closed : int ^ |{13}> \/ ballot }
type vote_list = +{ cons : |{4}> vote_list,
                    nil : 1 }
proc asset cons : ($t[P] : vote_list) |{5}- ($l[P] : vote_list) = 
  {
    work ;
    $l[P].cons ;
    pay $l[P] {4};
    $l[P] <- $t[P]
  }
type voters = &{ check : <{0}| int -> +{ success : voters,
                                         failure : voters },
                 size : <{0}| int ^ voters }
proc contract open_election : (T : int), ($vs[P] : voters),
                              ($c1[P] : vote_list), ($c2[P] : vote_list)
                                                  |{14}- (#sb[S] : ballot) = 
  {
    $lb[L] <- accept #sb[S] ;
    get $lb[L] {16};
    work ;
    $lb[L].open ;
    v = recv $lb[L] ;
    $vs[P].check ;
    pay $vs[P] {0};
    send $vs[P] ((tick  ; v)) ;
    work ;
    case $vs[P] ( success => $lb[L].vote ;
                             work ;
                             c = recv $lb[L] ;
                             if (tick  ; (tick  ; c) > (tick  ; 0))
                             then
                               $c1n[P] <- cons <- $c1[P] ;
                               pay $lb[L] {0};
                               #sb[S] <- detach $lb[L] ;
                               #sb[S] <- check <- T $vs[P] $c1n[P] $c2[P]
                             else
                               $c2n[P] <- cons <- $c2[P] ;
                               pay $lb[L] {0};
                               #sb[S] <- detach $lb[L] ;
                               #sb[S] <- check <- T $vs[P] $c1[P] $c2n[P]
                | failure => $lb[L].novote ;
                             pay $lb[L] {9};
                             #sb[S] <- detach $lb[L] ;
                             #sb[S] <- check <- T $vs[P] $c1[P] $c2[P] )
  }
proc asset count_helper : (n : int), ($c[P] : vote_list) |{2}- ($s[P] : int ^ 1) = 
  {
    case $c[P] ( cons => get $c[P] {4};
                         work ;
                         let n = (tick  ; (tick  ; n) + (tick  ; 1)) ;
                         $s[P] <- count_helper <- n $c[P]
               | nil => wait $c[P] ;
                        work ;
                        send $s[P] ((tick  ; n)) ;
                        close $s[P] )
  }
proc asset count_list : ($c[P] : vote_list) |{4}- ($s[P] : int ^ 1) = 
  {
    work ;
    let n = (tick  ; 0) ;
    $s[P] <- count_helper <- n $c[P]
  }
proc contract count : (T : int), ($vs[P] : voters),
                      ($c1[P] : vote_list), ($c2[P] : vote_list)
                                            |{14}- (#sb[S] : ballot) = 
  {
    $s1[P] <- count_list <- $c1[P] ;
    $s2[P] <- count_list <- $c2[P] ;
    s1 = recv $s1[P] ;
    work ;
    s2 = recv $s2[P] ;
    work ;
    wait $s1[P] ;
    work ;
    wait $s2[P] ;
    if (tick  ; (tick  ; s1) > (tick  ; s2))
    then
      #sb[S] <- closed_election <- s1 $vs[P]
    else
      #sb[S] <- closed_election <- s2 $vs[P]
  }
proc contract check : (T : int), ($vs[P] : voters),
                      ($c1[P] : vote_list), ($c2[P] : vote_list)
                                              |{18}- (#sb[S] : ballot) = 
  {
    work ;
    $vs[P].size ;
    pay $vs[P] {0};
    n = recv $vs[P] ;
    if (tick  ; (tick  ; n) = (tick  ; T))
    then
      #sb[S] <- count <- T $vs[P] $c1[P] $c2[P]
    else
      #sb[S] <- open_election <- T $vs[P] $c1[P] $c2[P]
  }
proc contract closed_election : (w : int), ($vs[P] : voters)
                                            |- (#sb[S] : ballot) = 
  {
    $lb[L] <- accept #sb[S] ;
    get $lb[L] {16};
    work ;
    $lb[L].closed ;
    work ;
    send $lb[L] ((tick  ; w)) ;
    pay $lb[L] {13};
    #sb[S] <- detach $lb[L] ;
    #sb[S] <- closed_election <- w $vs[P]
  }
TC time: 0.34499168396
Inference time: 11.1479759216
# Vars = 140
# Constraints = 457
% compilation successful!
% runtime successful!
\end{lstlisting}

\subsection{Wallet}

\begin{lstlisting}
type coin = 1
type lcoin = +{ cons : coin *[R] lcoin,
                nil : 1 }
type money = /\ &{ value : <{2}| int ^ \/ money,
                   coins : <{5}| lcoin *[R] \/ money }
proc asset emp : . |{2}- ($l[R] : lcoin) = 
  {
    work ;
    $l[R].nil ;
    work ;
    close $l[R]
  }
proc contract wallet : (n : int), ($l[R] : lcoin) |- (#sm[S] : money) = 
  {
    $m[L] <- accept #sm[S] ;
    case $m[L] ( value => get $m[L] {2};
                          work ;
                          send $m[L] ((tick  ; n)) ;
                          work {0};
                          #sm[S] <- detach $m[L] ;
                          #sm[S] <- wallet <- n $l[R]
               | coins => get $m[L] {5};
                          work {0};
                          work ;
                          send $m[L] $l[R] ;
                          $l[R] <- emp <-  ;
                          work ;
                          let n = (tick  ; 0) ;
                          #sm[S] <- detach $m[L] ;
                          #sm[S] <- wallet <- n $l[R] )
  }
TC time: 0.217914581299
Inference time: 6.70695304871
# Vars = 32
# Constraints = 102
% compilation successful!
% runtime successful!
\end{lstlisting}

\section{Types}\label{sec:grammar}
First, I present the grammar for ordinary functional types $\tau$
with potential.
\[
\begin{array}{rccl}
\tau & ::= & & t \mid \tau \to \tau
\mid \tau + \tau \mid \tau \times \tau \\
& & \mid & \tint \mid \tbool \mid \plist{\tau}{q} \\
& & \mid & \tproc{A_\p}{\overline{A_\p}}_\p \mid
\tproc{A_\s}{\overline{A_\s} \semi \overline{A_\p}}_\s \mid
\tproc{A_\c}{\overline{A_\s} \semi \overline{A}}_\c
\end{array}
\]
Next, I define the purely linear session types.
\[
\begin{array}{rccl}
A_\p & ::= & & V \mid \ichoice{\ell : A_\p}_{\ell \in L} \mid
\echoice{\ell : A_\p}_{\ell \in L} \mid A_m \lolli_m A_\p
\mid A_m \tensor_m A_\p \mid \one \\
 & & \mid & \tau \arrow A_\p \mid \tau \product A_\p
\mid \tpaypot{A_\p}{r} \mid \tgetpot{A_\p}{r}
\end{array}
\]
Next, the shared linear session types.
\[
\begin{array}{rccl}
A_\l & ::= & & V \mid \ichoice{\ell : A_\l}_{\ell \in L} \mid
\echoice{\ell : A_\l}_{\ell \in L} \mid A_m \lolli_m A_\l
\mid A_m \tensor_m A_\l \\
 & & \mid & \tau \arrow A_\l \mid \tau \product A_\l
\mid \tpaypot{A_\l}{r} \mid \tgetpot{A_\l}{r} \\
& & \mid & \down A_\s
\end{array}
\]
Finally, the shared session type.
\[
\begin{array}{rcl}
A_\s & ::= & \up A_\l
\end{array}
\]
The client linear types follow the same grammar as purely
linear types.
The combined type is represented using $A$ which denotes
the type of either a client or contract process in linear mode.
\[
\begin{array}{rcl}
A_\c & ::= & A_\p \\
A & ::= & A_\c \mid A_\l
\end{array}
\]
First, the expressions at the functional layer are as follows
(usual terms from a functional language).
\[
\begin{array}{rcclr}
M, N & ::= & & \lam{x}{\tau}{M_x} \mid M \; N \\
& & \mid & \inl{M} \mid \inr{M}
\mid \case{M}{M_l}{M_r} \\
& & \mid & \pair{M}{N} \mid \projl{M} \mid \projr{M} \\
& & \mid & n \mid \m{true} \mid \m{false} \\
& & \mid & [] \mid M::N \mid
\match{M}{M_1}{x::xs}{M_2} \\
& & \mid & \eproc{c_\p}{P_{c_\p,\overline{a}}}{\overline{a}}
\mid \eproc{c_\s}{P_{c_\s,\overline{a},\overline{d}}}
{\overline{a} \semi \overline{d}}
\mid \eproc{c_\c}{P_{c_\c,\overline{a}, \overline{b}}}
{\overline{a} \semi \overline{b}}
\end{array}
\]
The processes (proof terms) are as follows.
\[
\begin{array}{rccll}
P,Q & ::= & & \ecut{c}{M}{\overline{a}}{P_c} & \quad
\text{spawn process computed by $M$ and continue with} \\
& & & & \quad \text{$P_a$, both communicating along fresh channel $a$} \\
& & \mid & \fwd{x}{y} & \quad \text{forward between $x$ and $y$} \\
& & \mid & \esendl{x}{l_k} \semi P & \quad
\text{send label $l_k$ along x} \\
& & \mid & \ecase{x}{l_i}{P} & \quad
\text{branch on received label along $x$} \\
& & \mid & \esendch{x}{w} \semi P & \quad
\text{send channel/value $w$ along $x$} \\
& & \mid & \erecvch{x}{y} \semi P & \quad
\text{receive channel/value along $x$ and bind it to $y$} \\
& & \mid & \eclose{x} & \quad \text{close channel $x$} \\
& & \mid & \ewait{x} \semi P & \quad
\text{wait on closing channel $x$} \\
& & \mid & \ework{p} \semi P & \quad
\text{do work $p$, continue with $P$} \\
& & \mid & \eget{x}{p} \semi P & \quad
\text{get potential $p$ on channel $x$} \\
& & \mid & \epay{x}{p} \semi P & \quad
\text{pay potential $p$ on channel $x$} \\
& & \mid & \eacquire{x_\l}{x_\s} \semi P_{x_\l} & \quad
\text{send an acquire request along $x_\s$} \\
& & \mid & \eaccept{x_\l}{x_\s} \semi P_{x_\l} & \quad
\text{accept an acquire request along $x_\s$} \\
& & \mid & \edetach{x_\s}{x_\l} \semi P_{x_\s} & \quad
\text{send a detach request along $x_\l$} \\
& & \mid & \erelease{x_\s}{x_\l} \semi P_{x_\s} & \quad
\text{receive a detach request along $x_\l$} \\
\end{array}
\]

\section{Type System}\label{sec:proc-typing}
We first define the judgments we use in our type system.
\[
\begin{array}{ll}
\Psi \exppot{q} M : \tau & \text{term $M$ has type $\tau$} \\
& \text{and needs potential $q$ for evaluation} \\
\Psi \semi \G \semi \D \entailpot{q} P :: (c_m : A) &
\text{process $P$ offers service of type $A$}\\
& \text{along channel $c$ at mode $m = (\m{\s, \l, \c, \p})$}\\
& \text{and uses shared channels from $\G$}\\
& \text{and linear channels from $\D$}\\
& \text{and functional variables from $\Psi$}\\
& \text{and stores potential $q$}
\end{array}
\]
Mode $\s$ stands for channels in shared mode.
Mode $\l$ stands for shared channels in their
linear mode.
Mode $\c$ stands for linear channels that internally
depend on shared processes.
Mode $\p$ stands for purely linear channels offered
by purely linear processes.

\subsection{Monad}
First, I present the rules concerning the monad.
\subsubsection*{Introduction Rules}
\[
\infer[\{\} I_\p]
{\Psi \exppot{q} \eproc{x_\p}{P}{\overline{d_\p}} :
\tproc{A_\p}{\overline{D_\p}}_\p}
{\D = \overline{d_\p : D_\p} \qquad
\Psi \semi \cdot \semi \D \entailpot{q} P :: (x_\p : A_\p)}
\]

\[
\infer[\{\} I_\s]
{\Psi \exppot{q} \eproc{x_\s}{P}
{\overline{a_\s} \semi \overline{d_\p}} :
\tproc{A}{\overline{A_\s} \semi \overline{D_\p}}_\s}
{\G = \overline{a_\s : A_\s} \qquad
\D = \overline{d_\p : D_\p} \qquad
\Psi \semi \G \semi \D \entailpot{q} P :: (x_\s : A)}
\]

\[
\infer[\{\} I_\c]
{\Psi \exppot{q} \eproc{x_\c}{P}
{\overline{a_\s} \semi \overline{d}} :
\tproc{A}{\overline{A_\s} \semi \overline{D}}_\c}
{\G = \overline{a_\s : A_\s} \qquad
\D = \overline{d : D} \qquad
\Psi \semi \G \semi \D \entailpot{q} P :: (x_\c : A)}
\]

\subsubsection*{Elimination Rules}
\[
\inferrule*[right = $\{\}E_{\p m(\in \{\p,\s,\l,\c\})}$]
{r = p+q \qquad
\D = \overline{d_\p : D_\p} \qquad
\Psi \share (\Psi_1, \Psi_2) \\\\
\Psi_1 \exppot{p} M : \tproc{A}{\overline{D_\p}}_\p \qquad
\Psi_2 \semi \G \semi \D', (x_\p : A) \entailpot{q} Q :: (z_m : C)}
{\Psi \semi \G \semi \D, \D' \entailpot{r}
\ecut{x_\p}{M}{\overline{d_\p}}{Q} :: (z_m : C)}
\]

\[
\inferrule*[right = $\{\}E_{\s m(\in \{\s,\l,\c\})}$]
{r = p+q \qquad
\G \supseteq \overline{a_\s : A_\s} \qquad
\D = \overline{d_\p : D_\p} \qquad
(A_\s, A_\s) \esync \qquad
\Psi \share (\Psi_1, \Psi_2) \\\\
\Psi_1 \exppot{p} M : \tproc{A}{\overline{A_\s} \semi
\overline{D_\p}}_\s \qquad
\Psi_2 \semi \G, (x_\s : A) \semi \D' \entailpot{q} Q :: (z_m : C)}
{\Psi \semi \G \semi \D, \D' \entailpot{r}
\ecut{x_\s}{M}{\overline{d_\p}}{Q} :: (z_m : C)}
\]

\[
\inferrule*[right = $\{\}E_{\c m(\in \{\l,\c \})}$]
{r = p+q \qquad
\G \supseteq \overline{a_\s : A_\s} \qquad
\D = \overline{d : D} \qquad
\Psi \share (\Psi_1, \Psi_2) \\\\
\Psi_1 \exppot{p} M : \tproc{A}
{\overline{A_\s}\semi \overline{D}}_\c \qquad
\Psi_2 \semi \G \semi (x_\c : A), \D'
\entailpot{q} Q :: (z_m : C)}
{\Psi \semi \G \semi \D, \D' \entailpot{r}
\ecut{x_\c}{M}{\overline{a_\s}
\semi \overline{d}}{Q} :: (z_m : C)}
\]

The rest of the rules for expressions in the functional layer
are standard. We skip them and discuss the process layer.

\subsection{Forwarding}
\[
\infer[\m{fwd}_{m(\in\{\p,\c\})}]
{\Psi \semi \G \semi (y_m : A) \entailpot{q}
\fwd{x_m}{y_m} :: (x_m : A)}
{q = 0}
\]

\subsection{Labels and Branching}
\[
\infer[\oplus R]
{\Psi \semi \G \semi \D \entailpot{q}
\esendl{x_m}{k} \semi P ::
(x_m : \ichoice{\ell : A_{\ell}}_{\ell \in L})}
{\Psi \semi \G \semi \D \entailpot{q}
P :: (x_m : A_k) \qquad (k \in L)}
\]
\[
\infer[\oplus L]
{\Psi \semi \G \semi \D, (x_m : \ichoice{\ell : A_{\ell}}_{\ell \in L})
\entailpot{q} \ecase{x_m}{\ell}{Q_{\ell}}_{\ell \in L} :: (z_k : C)}
{\Psi \semi \G \semi \D, (x_m : A_{\ell}) \entailpot{q}
Q_{\ell} :: (z_k : C)
\qquad (\forall \ell \in L)}
\]

\[
\infer[\with R]
{\Psi \semi \G \semi \D \entailpot{q}
\ecase{x_m}{\ell}{P_{\ell}}_{\ell \in L} ::
(x_m : \echoice{\ell : A_{\ell}}_{\ell \in L})}
{\Psi \semi \G \semi \D \entailpot{q}
P :: (x_m : A_{\ell}) \qquad (\forall \ell \in L)}
\]
\[
\infer[\with L]
{\Psi \semi \G \semi \D, (x_m : \echoice{\ell : A_{\ell}}_{\ell \in L})
\entailpot{q} \esendl{x_m}{k} \semi P :: (z_k : C)}
{\Psi \semi \G \semi \D, (x_m : A_{\ell}) \entailpot{q}
Q_{\ell} :: (z_k : C)
\qquad (k \in L)}
\]

\subsection{Linear Channel Communication}
\[
\infer[\tensor_n R]
{\Psi \semi \G \semi \D, (w_n : A) \entailpot{q}
\esendch{x_m}{w_n} \semi P :: (x_m : A \tensor_n B)}
{\Psi \semi \G \semi \D \entailpot{q}
P :: (x_m : B)}
\]
\[
\infer[\tensor_n L]
{\Psi \semi \G \semi \D, (x_m : A \tensor_n B)
\entailpot{q} \erecvch{x_m}{y_n} \semi Q :: (z_k : C)}
{\Psi \semi \G \semi \D, (y_n : A), (x_m : B) \entailpot{q}
Q :: (z_k : C)}
\]

\[
\infer[\lolli_n R]
{\Psi \semi \G \semi \D \entailpot{q}
\erecvch{x_m}{y_n} \semi P :: (x_m : A \lolli B)}
{\Psi \semi \G \semi \D, (y_n : A) \entailpot{q}
P :: (x_m : B)}
\]
\[
\infer[\lolli_n L]
{\Psi \semi \G \semi \D, (w_n : A), (x_m : A \lolli B)
\entailpot{q} \esendch{x_m}{w_n} \semi Q :: (z_k : C)}
{\Psi \semi \G \semi \D, (x_m : B) \entailpot{q}
Q :: (z_k : C)}
\]

\subsection{Value Communication}
\[
\infer[\product R]
{\Psi \semi \G \semi \D \entailpot{r}
\esendch{x_m}{M} \semi P :: (x_m : \tau \product A)}
{r = p+q \qquad
\Psi \share (\Psi_1, \Psi_2) \qquad
\Psi_1 \exppot{p} M : \tau \qquad
\Psi_2 \semi \G \semi \D \entailpot{q}
P :: (x_m : A)}
\]
\[
\infer[\product L]
{\Psi \semi \G \semi \D, (x_m : \tau \product A)
\entailpot{q} \erecvch{x_m}{y} \semi Q :: (z_k : C)}
{\Psi, (y : \tau) \semi \G \semi \D, (x_m : A) \entailpot{q}
Q :: (z_k : C)}
\]

\[
\infer[\arrow R]
{\Psi \semi \G \semi \D \entailpot{q}
\erecvch{x_m}{y} \semi P :: (x_m : \tau \arrow A)}
{\Psi, (y : \tau) \semi \G \semi \D \entailpot{q}
P :: (x_m : B)}
\]
\[
\infer[\arrow L]
{\Psi \semi \G \semi \D, (x_m : \tau \arrow A)
\entailpot{r} \esendch{x_m}{M} \semi Q :: (z_k : C)}
{r = p+q \qquad
\Psi \share (\Psi_1, \Psi_2) \qquad
\Psi_1 \exppot{p} M : \tau \qquad
\Psi_2 \semi \G \semi \D, (x_m : A) \entailpot{q}
Q :: (z_k : C)}
\]

\subsection{Termination}
\[
\infer[\one R]
{\Psi \semi \G \semi \cdot \entailpot{q}
\eclose{x_m} :: (x_m : \one)}
{q = 0}
\hspace{5em}
\infer[\one L]
{\Psi \semi \G \semi \D, (x_m : \one) \entailpot{q}
\ewait{x_m} \semi Q :: (z_k : C)}
{\Psi \semi \G \semi \D \entailpot{q} Q :: (z_k : C)}
\]

\subsection{Potential}
\[
\infer[\m{work}]
{\Psi \semi \G \semi \D \entailpot{q} \etick{r} \semi P :: (x_m : A)}
{q = p + r \qquad
\Psi \semi \G \semi \D \entailpot{p} P :: (x_m : A)}
\]

\[
\infer[\paypot R]
{\Psi \semi \G \semi \D \entailpot{q} \epay{x_m}{r} \semi P ::
(x_m : \tpaypot{A}{r})}
{q = p+r \qquad
\Psi \semi \G \semi \D \entailpot{p} P :: (x_m : A)}
\]
\[
\infer[\paypot L]
{\Psi \semi \G \semi \D, (x_m : \tpaypot{A}{r}) \entailpot{q}
\eget{x_m}{r} \semi P :: (z_k : C)}
{p = q+r \qquad
\Psi \semi \G \semi \D, (x_m : A) \entailpot{p} P :: (z_k : C)}
\]

\[
\infer[\getpot R]
{\Psi \semi \G \semi \D \entailpot{q} \eget{x_m}{r} \semi P ::
(x_m : \tgetpot{A}{r})}
{p = q+r \qquad
\Psi \semi \G \semi \D \entailpot{p} P :: (x_m : A)}
\]
\[
\infer[\getpot L]
{\Psi \semi \G \semi \D, (x_m : \tgetpot{A}{r}) \entailpot{q}
\epay{x_m}{r} \semi P :: (z_k : C)}
{q = p+r \qquad
\Psi \semi \G \semi \D, (x_m : A) \entailpot{p} P :: (z_k : C)}
\]

\subsection{Acquiring and Releasing}
\[
\infer[\up R]
{\Psi \semi \G \semi \D \entailpot{q}
\eaccept{x_\l}{x_\s} \semi P :: (x_\s : \up A_\l)}
{\D \plin \qquad
\Psi \semi \G \semi \D \entailpot{q} P :: (x_\l : A_\l)}
\]
\[
\infer[\up L_{m(=\l,\c)}]
{\Psi \semi \G, (x_\s : \up A_\l) \semi \D
\entailpot{q} \eacquire{x_\l}{x_\s} \semi Q :: (z_m : C)}
{\Psi \semi \G \semi \D, (x_\l : A_\l)
\entailpot{q} Q :: (z_m : C)}
\]

\[
\infer[\down R]
{\Psi \semi \G \semi \D \entailpot{q}
\edetach{x_\s}{x_\l} \semi P :: (x_\l : \down A_\s)}
{\D \plin \qquad
\Psi \semi \G \semi \D \entailpot{q} P :: (x_\s : A_\s)}
\]
\[
\infer[\down L_{m(=\l,\c)}]
{\Psi \semi \G \semi \D, (x_\l : \down A_\s)
\entailpot{q} \erelease{x_\s}{x_\l} \semi Q :: (z_m : C)}
{\Psi \semi \G, (x_\s : A_\s) \semi \D
\entailpot{q} Q :: (z_m : C)}
\]

\section{Operational Cost Semantics}\label{sec:semantics}
First, we define the judgments for expressions. The first
judgment is a small step semantics for expressions,
$M \step M'$ and $M \val$. Finally,
we introduce another judgment for processes,
$\proc{c_m}{w, P} \step \proc{c_m'}{w', P'}$ and a new predicate
$\msg{c_m}{w, M}$ to denote a message. Additionally, we define
processes with a hole for a compact representation of
the cost semantics.

\[
\begin{array}{rccl}
P [\cdot] & ::= & & \ecut{c}{[\cdot]}{a_i}{P_c} \\
& & \mid & \esendch{c}{[\cdot]} \semi P
\end{array}
\]

\[
\infer[\m{internal}]
{\proc{c_m}{w, P[N]} \step \proc{c_m}{w+\mu, P[V]}}
{\bigeval{N}{V}{\mu}}
\]

\begin{mathpar}
\inferrule*[right=$\{\}E_{\p m}$]
{\fresh{c_\p}}
{\proc{d_m}{w, \ecut{x_\p}{\eproc{x_\p'}{P_{x_\p',\overline{y}}}
{\overline{y}}}{\overline{a}}{Q}}
\step \\ \proc{c_\p}{0, P_{c_\p,\overline{a}}} \quad \proc{d_m}{w,
[c_\p/x_\p]Q}}
\end{mathpar}

\[
\inferrule*[right = $\{\}E_{\s m}$]
{\fresh{c_\s}}
{\proc{d_m}{w, \ecut{x_\s}{\eproc{x_\s'}
{P_{x_\s',\overline{y}, \overline{z}}}
{\overline{y} \semi \overline{z}}}{\overline{a}
\semi \overline{b}}{Q}} \step \\
\proc{c_\s}{0, P_{c_\s,\overline{a}, \overline{b}}} \quad \proc{d_m}{w,
[c_\s/x_\s]Q}}
\]

\[
\inferrule*[right = $\{\}E_{\c\c}$]
{\fresh{c_\c}}
{\proc{d_\c}{w, \ecut{x_\c}{\eproc{x_\c'}{P_{x_\c',\overline{y},
\overline{z}}}
{\overline{y} \semi \overline{z}}}
{\overline{a} \semi \overline{b}}{Q}} \step \\
\proc{c_\c}{0, P_{c_\c,\overline{a}, \overline{b}}} \quad
\proc{d_\c}{w, [c_\c/x_\l]Q}}
\]

\[
\infer[\m{fwd}^+]
{\msg{d_m}{w', M} \quad
\proc{c_m}{w, \fwd{c_m}{d_m}} \step
\msg{c_m}{w + w', [c_m/d_m]M}}
{}
\]
\[
\infer[\m{fwd}^-]
{\proc{c_m}{w, \fwd{c_m}{d_m}} \quad
\msg{e_l}{w', M(c_m)} \step
\msg{e_l}{w + w', M(d_m)}}
{}
\]

\[
\infer[\oplus C_s]
{\proc{c_m}{w, \esendl{c_m}{\ell} \semi P} \step
\proc{c^+_m}{w, [c^+_m/c_m]P} \quad
\msg{c_m}{0, \esendl{c_m}{\ell} \semi \fwd{c_m}{c^+_m}}}
{\fresh{c^+_m}}
\]
\begin{mathpar}
\inferrule*[right=$\oplus C_r$]
{\mathstrut}
{\msg{c_m}{w, \esendl{c_m}{\ell} \semi \fwd{c_m}{c^+_m}} \quad
\proc{d_k}{w', \ecase{c_m}{l}{Q_l}_{l \in L}} \step \\
\proc{d_k}{w+w', [c^+_m/c_m]Q_{\ell}}}
\end{mathpar}

\[
\infer[\with C_s]
{\proc{d_k}{w, \esendl{c_m}{\ell} \semi P} \step
\msg{c^+_m}{0, \esendl{c_m}{\ell} \semi \fwd{c^+_m}{c_m}} \quad
\proc{d_k}{w, [c^+_m/c_m]P}
}
{\fresh{c^+_m}}
\]
\begin{mathpar}
\inferrule*[right=$\with C_r$]
{\mathstrut}
{\proc{c_m}{w', \ecase{c_m}{l}{Q_l}_{l \in L}} \quad
\msg{c^+_m}{0, \esendl{c_m}{\ell} \semi \fwd{c^+_m}{c_m}} \step \\
\proc{c^+_m}{w+w', [c^+_m/c_m]Q_{\ell}}}
\end{mathpar}

\[
\infer[\tensor_n C_s]
{\proc{c_m}{w, \esendch{c_m}{e_n} \semi P} \step
\proc{c^+_m}{w, [c^+_m/c_m]P} \quad
\msg{c_m}{0, \esendch{c_m}{e_n} \semi \fwd{c_m}{c^+_m}}}
{\fresh{c^+_m}}
\]
\[
\inferrule*[right = $\tensor_n C_r$]
{\mathstrut}
{\msg{c_m}{w, \esendch{c_m}{e_n} \semi \fwd{c_m}{c^+_m}} \quad
\proc{d_k}{w', \erecvch{c_m}{x_n} \semi Q} \step \\
\proc{d_k}{w+w', [c^+_m/c_m][e_n/x_n]Q}}
\]

\[
\infer[\lolli_n C_s]
{\proc{d_k}{w, \esendch{c_m}{e_n} \semi P} \step
\msg{c^+_m}{0, \esendch{c_m}{e_n} \semi \fwd{c^+_m}{c_m}} \quad
\proc{d_k}{w, [c^+_m/c_m]P}}
{\fresh{c^+_m}}
\]
\[
\inferrule*[right = $\lolli_n C_r$]
{\mathstrut}
{\proc{c_m}{w', \erecvch{c_m}{x_n} \semi Q} \quad
\msg{c^+_m}{w, \esendch{c_m}{e_n} \semi \fwd{c^+_m}{c_m}}
\step \\ \proc{c^+_m}{w+w', [c^+_m/c_m][e_n/x_n]Q}}
\]

\[
\infer[\product C_s]
{\proc{c_m}{w, \esendch{c_m}{N} \semi P} \step
\proc{c^+_m}{w, [c^+_m/c_m]P} \quad
\msg{c_m}{0, \esendch{c_m}{N} \semi \fwd{c_m}{c^+_m}}}
{\fresh{c^+_m} \qquad N \val}
\]
\[
\inferrule*[right = $\product C_r$]
{\mathstrut}
{\msg{c_m}{w, \esendch{c_m}{N} \semi \fwd{c_m}{c^+_m}} \quad
\proc{d_k}{w', \erecvch{c_m}{x} \semi Q} \step \\
\proc{d_k}{w+w', [c^+_m/c_m][N/x]Q}}
\]

\[
\infer[\arrow C_s]
{\proc{d_k}{w, \esendch{c_m}{N} \semi P} \step
\msg{c^+_m}{0, \esendch{c_m}{N} \semi \fwd{c^+_m}{c_m}} \quad
\proc{d_k}{w, [c^+_m/c_m]P}}
{\fresh{c^+_m} \qquad N \val}
\]
\[
\inferrule*[right = $\arrow C_r$]
{\mathstrut}
{\proc{c_m}{w', \erecvch{c_m}{x} \semi Q} \quad
\msg{c^+_m}{w, \esendch{c_m}{N} \semi \fwd{c^+_m}{c_m}} \step \\
\proc{c^+_m}{w+w', [c^+_m/c_m][N/x]Q}}
\]

\[
\infer[\one C_s]
{\proc{c_m}{w, \eclose{c_m}} \step
\msg{c_m}{w, \eclose{c_m}}}
{\mathstrut}
\]
\[
\infer[\one C_r]
{\msg{c_m}{w, \eclose{c_m}} \quad
\proc{d_k}{w', \ewait{c_m} \semi Q} \step \proc{d_k}{w + w', Q}}
{\mathstrut}
\]

\[
\infer[\m{tick}]
{\proc{c_m}{w + \mu, P}}
{\proc{c_m}{w, \etick{\mu} \semi P}}
\]

\[
\infer[\paypot C_s]
{\proc{c_m}{w, \epay{c_m}{r} \semi P} \step
\proc{c^+_m}{w, [c^+_m/c_m]P} \quad
\msg{c_m}{0, \epay{c_m}{r} \semi \fwd{c_m}{c^+_m}}}
{\fresh{c^+_m}}
\]
\begin{mathpar}
\inferrule*[right=$\paypot C_r$]
{\mathstrut}
{\msg{c_m}{w, \epay{c_m}{r} \semi \fwd{c_m}{c^+_m}} \quad
\proc{d_k}{w', \eget{c_m}{r} \semi Q} \step \\
\proc{d_k}{w+w', [c^+_m/c_m]Q}}
\end{mathpar}

\[
\infer[\getpot C_s]
{\proc{d_k}{w, \epay{c_m}{r} \semi P} \step
\msg{c^+_m}{0, \epay{c_m}{r} \semi \fwd{c^+_m}{c}} \quad
\proc{d_k}{w, [c^+_m/c_m]P}}
{\fresh{c^+_m}}
\]
\begin{mathpar}
\inferrule*[right=$\getpot C_r$]
{\mathstrut}
{\proc{c_m}{w', \eget{c_m}{r} \semi Q} \quad
\msg{c^+_m}{w, \epay{c_m}{r} \semi \fwd{c^+_m}{c_m}} \step \\
\proc{c_m}{w+w', [c^+_m/c_m]Q}}
\end{mathpar}

\[
\inferrule*[right = $\up C$]
{\fresh{a_\l}}
{\proc{a_\s}{w', \eaccept{x_\l}{a_\s} \semi P_{x_\l}} \quad
\proc{c_m}{w, \eacquire{x_\l}{a_\s} \semi Q_{x_\l}} \step \\\\
\proc{a_\l}{w', P_{a_\l}} \quad
\proc{c_m}{w, Q_{a_\l}}}
\]

\[
\inferrule*[right = $\down C$]
{\mathstrut}
{\proc{a_\l}{w', \edetach{x_\s}{a_\l} \semi P_{x_\s}} \quad
\proc{c_m}{w, \erelease{x_\s}{a_\l} \semi Q_{x_\s}} \step \\\\
\proc{a_\s}{w', P_{a_\s}} \quad
\proc{c_m}{w, Q_{a_\s}}}
\]

\section{Configuration Typing}\label{sec:conf-typing}

\[
\infer[\m{emp}]
{\G_0 \potconf{0} (\cdot) :: (\cdot \semi \cdot)}
{}
\]

\[
\infer[\m{proc}_\p]
{\G_0 \potconf{E + q + w} \W, \proc{x_\p}{w, P} ::
(\G \semi \D, (x_\p : A_\p))}
{\G_0 \potconf{E} \W :: (\G \semi \D, \D_\p') \qquad
\cdot \semi \cdot \semi \D_\p' \entailpot{q} P :: (x_\p : A_\p)}
\]

\[
\infer[\m{proc}_\s]
{\G_0 \potconf{E + q + w} \W, \proc{x_\s}{w, P} :: (\G, (x_\s : A_\s)
\semi \D)}
{(x_\s : A_\s) \in \G_0 \qquad
(A_\s, A_\s) \; \m{esync} \qquad
\G_0 \potconf{E} \W :: (\G \semi \D,\D_\p') \qquad
\cdot \semi \G_0 \semi \D_\p' \entailpot{q} P :: (x_\s : A_\s)}
\]

\[
\infer[\m{proc}_\l]
{\G_0 \potconf{E + q + w} \W, \proc{x_\l}{w, P} :: (\G, (x_\s : A_\s)
\semi \D, (x_\l : A_\l))}
{(x_\s : A_\s) \in \G_0 \qquad
(A_\l, A_\s) \; \m{esync}  \qquad
\G_0 \potconf{E} \W :: (\G \semi \D, \D') \qquad
\cdot \semi \G_0 \semi \D' \entailpot{q} P :: (x_\l : A_\l)}
\]

\[
\infer[\m{proc}_\c]
{\G_0 \potconf{E+q+w} \W, \proc{x_\c}{w, P} :: (\G \semi
\D, (x_\c : A_\c))}
{\G_0 \potconf{E} \W :: (\G \semi \D, \D') \qquad
\cdot \semi \G_0 \semi \D' \entailpot{q} P :: (x_\c : A_\c) }
\]

\[
\infer[\m{msg}]
{\G_0 \potconf{E+q+w} \W, \msg{x_m}{w, M} :: (\G \semi
\D, (x_m : A))}
{\G_0 \potconf{E} \W :: (\G \semi \D, \D') \qquad
\cdot \semi \cdot \semi \D' \entailpot{q} M :: (x_m : A)}
\]

%\[
%\infer[\m{unavail}]
%{\G_0 \potconf{E} \W, \unavail{x_\s} :: (\G, (x_\s : A_\s) \semi
%\D)}
%{(x_\s : A_\s) \in \G_0 \qquad
%(A_\s, A_\s) \; \m{esync} \qquad
%\G_0 \potconf{E} \W :: (\G \semi \D)}
%\]

In addition, for a well-typed configuration $\G_0 \potconf{E} \W ::
(\G \semi \D)$,
we need the following well-formedness conditions.

\begin{itemize}
\item All channels in $\G_0, \G$ and $\D$ are unique.

\item $\G \subseteq \G_0$.
\end{itemize}

\subsection{Equi-Synchronizing}
\[
\infer[\oplus]
{(\ichoice{\ell : A_{\ell}}_{\ell \in L}, C_\s) \esync}
{(A_{\ell}, C_\s) \esync \quad (\forall \ell \in L)}
\qquad
\infer[\with]
{(\echoice{\ell : A_{\ell}}_{\ell \in L}, C_\s) \esync}
{(A_{\ell}, C_\s) \esync \quad (\forall \ell \in L)}
\]

\[
\infer[\tensor]
{(A \tensor B, C_\s) \esync}
{(B, C_\s) \esync}
\qquad
\infer[\lolli]
{(A \lolli B, C_\s) \esync}
{(B, C_\s) \esync}
\]

\[
\infer[\product]
{(\tau \product B, C_\s) \esync}
{(B, C_\s) \esync}
\qquad
\infer[\arrow]
{(\tau \arrow B, C_\s) \esync}
{(B, C_\s) \esync}
\]

\[
\infer[\paypot]
{(\tpaypot{A}{r}, C_\s) \esync}
{(A, C_\s) \esync}
\qquad
\infer[\getpot]
{(\tgetpot{A}{r}, C_\s) \esync}
{(A, C_\s) \esync}
\]

\[
\infer[\up]
{(\up A_\l, \up A_\l) \esync}
{(A_\l, \up A_\l) \esync}
\qquad
\infer[\down]
{(\down A_\s, A_\s) \esync}
{(A_\s, A_\s) \esync}
\]

\subsection{Purely Linear Context}

\begin{mathpar}
\infer[\m{emp}]
{\cdot \plin}
{}
\and
\infer[\m{step}]
{x_\p : A_\p, \D \plin}
{x_\p : A_\p \qquad \D \plin}
\end{mathpar}

\section{Type Safety}\label{sec:type-safety}

\begin{lemma}[Renaming]\label{lem:renaming}
The following renamings are allowed.

\begin{itemize}
\item If $\Psi \semi \G, (x_\s : A_\s) \semi \D \entailpot{q} P_{x_\s} :: (z_k : C)$ is
well-typed, so is $\G, (c_\s : A_\s) \semi \D \entailpot{q} P_{c_\s} :: (z_k : C)$.

\item If $\Psi \semi \G \semi \D, (x_m : A) \entailpot{q} P_{x_m} :: (z_k : C)$ is
well-typed, so is $\G \semi \D, (c_m : A) \entailpot{q} P_{c_m} :: (z_k : C)$.

\item If $\Psi \semi \G \semi \D \entailpot{q} P_{z_k} :: (z_k : C)$ is well-typed, so is
$\G \semi \D \entailpot{q} P_{c_k} :: (c_k : C)$.
\end{itemize}
\end{lemma}

\begin{lemma}[Invariants]\label{lem:invariants}
The process typing judgment $\Psi \semi \G \semi \D \entailpot{q} P ::
(x_m : A)$
preserves the following invariants.

\begin{itemize}
\item[$(\p)$] $\Psi \semi \cdot \semi \D_\p \entailpot{q} P :: (x_\p : A_\p)$

\item[$(\s/\l)$] $\Psi \semi \G \semi \D_\p \entailpot{q} P :: (x_\s : A_\s)$ or
$\Psi \semi \G \semi \D \entailpot{q} P :: (x_\l : A_\l)$

\item[$(\c)$] $\Psi \semi \G \semi \D \entailpot{q} P :: (x_\c : A_\c)$
\end{itemize}
\end{lemma}

\begin{proof}
The elimination rules preserve the invariant trivially because they can only
be applied when the invariant is maintained and the premise in each rule
maintains the same invariant.

\begin{itemize}
\item Case ($E_{\p\p}$) : This rule can only be applied when the context is
purely linear. And then adding $x_\p$ to the context will keep it purely linear.

\item Case ($E_{\p \s}, E_{\p \l}$) : This rule can only be applied if offering channel
is either in $\s$ or $\l$ mode and the context is purely linear. Hence, adding
$x_\p$ to the context is allowed.

\item Case ($E_{\p \c}$) : The context is mixed linear, hence adding a purely
linear channel is valid.

\item Case ($E_{\s\s}$, $E_{\s \l}$, $E_{\s\c}$) : The context has shared channels
in each case, hence adding another shared channel is valid.

\item Case ($E_{\c\c}$) : Adding a client linear channel to a mixed context is valid.

\item Case ($\m{fwd}$) :
\begin{itemize}
\item[$(\p)$] : $\D_\p = (y_\p : A_\p)$ which is valid since $\D_\p$ is purely
linear and there are no premises.

\item[$(\s/\l)$] : This rule cannot be applied since the $\m{fwd}$ rule applies
only when the offering mode is $\p$. Hence, there is a mode mismatch.

\item[$(\c$)] : Analogous to ($\s/ \l$).
\end{itemize}

\item Case ($\oplus R$) :
\begin{itemize}
\item[$(\p)$] : 
\[
\infer[\oplus R]
{\Psi \semi \cdot \semi \D_\p \entailpot{q}
(\esendl{x_\p}{k} \semi P) ::
(x_\p : \ichoice{\ell : A_{\ell}}_{\ell \in L})}
{\Psi \semi \cdot \semi \D_\p \entailpot{q}
P :: (x_\p : A_k) \qquad (k \in L)}
\]
The context doesn't change,
and the type of the offered channel remains purely linear.

\item[$(\s/\l)$] :
\[
\infer[\oplus R]
{\Psi \semi \G \semi \D \entailpot{q}
(\esendl{x_\l}{k} \semi P) ::
(x_\l : \ichoice{\ell : A_{\ell}}_{\ell \in L})}
{\Psi \semi \G \semi \D \entailpot{q}
P :: (x_\l : A_k) \qquad (k \in L)}
\]
The context doesn't change,
and the type of the offered channel remains shared linear.
Also, the mode of $x$
cannot be $\s$ because the type doesn't allow that.

\item[$(\c$)] :
\[
\infer[\oplus R]
{\Psi \semi \G \semi \D \entailpot{q}
(\esendl{x_\c}{k} \semi P) ::
(x_\c : \ichoice{\ell : A_{\ell}}_{\ell \in L})}
{\Psi \semi \G \semi \D \entailpot{q}
P :: (x_\c : A_k) \qquad (k \in L)}
\]
The context doesn't change, and
the type of the offered channel remains client linear.
\end{itemize}

\item Case ($\oplus L$) :
\begin{itemize}
\item[$(\p)$] : 
\[
\infer[\oplus L]
{\Psi \semi \cdot \semi \D_\p, (x_\p : \ichoice{\ell : A_{\ell}}_{\ell \in L})
\entailpot{q} \ecase{x_\p}{\ell}{Q_{\ell}}_{\ell \in L} :: (z_\p : C)}
{\Psi \semi \cdot \semi \D_\p, (x_\p : A_{\ell}) \entailpot{q}
Q_{\ell} :: (z_\p : C)
\qquad (\forall \ell \in L)}
\]
The context remains purely linear,
and the offered channel doesn't change.

\item[$(\s/\l)$] :
\[
\infer[\oplus L]
{\Psi \semi \G \semi \D, (x_m : \ichoice{\ell : A_{\ell}}_{\ell \in L})
\entailpot{q} \ecase{x_m}{\ell}{Q_{\ell}}_{\ell \in L} :: (z_k : C)}
{\Psi \semi \G \semi \D, (x_m : A_{\ell}) \entailpot{q}
Q_{\ell} :: (z_k : C)
\qquad (\forall \ell \in L)}
\]
The mode of $x_m$ doesn't change,
and the offered channel doesn't change.

\item[$(\c$)] :
\[
\infer[\oplus L]
{\Psi \semi \G \semi \D, (x_m : \ichoice{\ell : A_{\ell}}_{\ell \in L})
\entailpot{q} \ecase{x_m}{\ell}{Q_{\ell}}_{\ell \in L} :: (z_\c : C)}
{\Psi \semi \G \semi \D, (x_m : A_{\ell}) \entailpot{q}
Q_{\ell} :: (z_\c : C)
\qquad (\forall \ell \in L)}
\]
The mode of the channel $x_m$ doesn't change,
and the offered channel doesn't change.
\end{itemize}

\item Case ($\lolli_n R$) :
\begin{itemize}
\item[$(\p)$] : 
\[
\infer[\lolli_\p R]
{\Psi \semi \cdot \semi \D_\p \entailpot{q}
\erecvch{x_\p}{y_\p} \semi P :: (x_\p : A \lolli_\p B)}
{\Psi \semi \cdot \semi \D_\p, (y_\p : A) \entailpot{q}
P :: (x_\p : B)}
\]
A process offering a purely linear channel only allows exchanging
purely linear channels.
This channel gets added to the purely linear context,
and the type of the offered channel remains purely linear.

\item[$(\s/\l)$] :
\[
\infer[\lolli_n R]
{\Psi \semi \G \semi \D \entailpot{q}
\erecvch{x_\l}{y_n} \semi P :: (x_\l : A \lolli_n B)}
{\Psi \semi \G \semi \D, (y_n : A) \entailpot{q}
P :: (x_\l : B)}
\]
A linear channel gets added to the mixed linear context,
and the type of the offered channel remains shared linear.
Also, the mode of $x$
cannot be $\s$ because the type doesn't allow that.

\item[$(\c$)] :
\[
\infer[\lolli_n R]
{\Psi \semi \G \semi \D \entailpot{q}
\erecvch{x_\c}{y_n} \semi P :: (x_\c : A \lolli_n B)}
{\Psi \semi \G \semi \D, (y_n : A) \entailpot{q}
P :: (x_\c : B)}
\]
A linear channel gets added to the mixed linear context,
and the type of the offered channel remains client linear.
\end{itemize}

\item Case ($\lolli_n L$) :
\begin{itemize}
\item[$(\p)$] : 
\[
\infer[\lolli_\p L]
{\Psi \semi \cdot \semi \D_\p, (w_\p : A), (x_\p : A \lolli_\p B)
\entailpot{q} \esendch{x_\p}{w_\p} \semi Q :: (z_\p : C)}
{\Psi \semi \cdot \semi \D_\p, (x_\p : B) \entailpot{q}
Q :: (z_\p : C)}
\]
A purely linear channel is allowed in a purely linear context.
The context remains purely linear,
and the offered channel doesn't change.

\item[$(\s/\l)$] :
\[
\infer[\lolli_n L]
{\Psi \semi \G \semi \D, (w_n : A), (x_m : A \lolli_n B)
\entailpot{q} \esendch{x_m}{w_n} \semi Q :: (z_k : C)}
{\Psi \semi \G \semi \D, (x_m : B) \entailpot{q}
Q :: (z_k : C)}
\]
A linear channel is allowed in a mixed linear context.
The mode of the channel $x_m$ doesn't change,
and the offered channel doesn't change.

\item[$(\c$)] :
\[
\infer[\lolli_n L]
{\Psi \semi \G \semi \D, (w_n : A), (x_m : A \lolli_n B)
\entailpot{q} \esendch{x_m}{w_n} \semi Q :: (z_k : C)}
{\Psi \semi \G \semi \D, (x_m : B) \entailpot{q}
Q :: (z_k : C)}
\]
A linear channel is allowed in a mixed linear context.
The mode of the channel $x_m$ doesn't change,
and the offered channel doesn't change.
\end{itemize}

\item Case ($\up R$) : 
\begin{itemize}
\item[$(\p)$] : This rule cannot be applied since the offered channel in
this case should be purely linear, which is not the case for $\up R$ rule.

\item[$(\s/\l)$] :
\[
\infer[\up R]
{\Psi \semi \G \semi \D \entailpot{q}
\eaccept{x_\l}{x_\s} \semi P :: (x_\s : \up A_\l)}
{\D \plin \qquad \Psi \semi \G \semi \D \entailpot{q} P :: (x_\l : A_\l)}
\]
The context doesn't change and the offered channel switches its mode
from $\s$ to $\l$. Moreover, the rule cannot be applied if the offered channel
is in $\l$ mode, since there will be a mode mismatch.

\item[$(\c)$] : This rule cannot be applied since the offered channel
should be in $\c$ mode, which doesn't match with $\s$.
\end{itemize}

\item Case ($\down R$) : Analogous to $\up R$.

\item Case ($\up L$) :
\begin{itemize}
\item[$(\p)$] : This rule cannot be applied since the context should be
purely linear, which is not the case for $\up L$ rule.

\item[$(\s/\l)$] :
\[
\infer[\up L_\l]
{\Psi \semi \G, (x_\s : \up A_\l) \semi \D
\entailpot{q} \eacquire{x_\l}{x_\s} \semi Q :: (z_\l : C)}
{\Psi \semi \G \semi \D, (x_\l : A_\l)
\entailpot{q} Q :: (z_\l : C)}
\]
A shared linear channel is allowed in a mixed linear context. The mode of the
offering channel is unchanged. A shared channel is removed from the shared context,
but the new context is still shared.

\item[$(\c)$] :
\[
\infer[\up L]
{\Psi \semi \G, (x_\s : \up A_\l) \semi \D
\entailpot{q} \eacquire{x_\l}{x_\s} \semi Q :: (z_\c : C)}
{\Psi \semi \G \semi \D, (x_\l : A_\l)
\entailpot{q} Q :: (z_\c : C)}
\]
A shared linear channel gets added to the mixed linear context, which is
allowed. A shared channel is removed from the shared context,
but the new context is still shared. Moreover, the offered channel
remains at the same mode.
\end{itemize}

\item Case ($\down L$) : Analogous to $\up L$ rule.

\end{itemize}
\end{proof}

\begin{lemma}[Configuration Weakening]\label{lem:conf-weakening}
If we have a well-typed configuration, $\G_0 \potconf{E} \W :: (\G \semi \D)$,
then for a shared channel $c_\s : B_\s \notin \G_0$, we can weaken $\G_0$
and get $\G_0, (c_\s : B_\s) \potconf{E}\W :: (\G \semi \D)$.
\end{lemma}

\begin{proof}
We case analyze on the configuration typing judgment.

\begin{itemize}
\item Case ($\m{emp}$) : We have $\G_0 \potconf{0} (\cdot) ::
(\cdot \semi \cdot)$. But, since there is no premise, we use the
$\m{emp}$ rule to get $\G_0, (c_\s : B_\s) \potconf{0} (\cdot) ::
(\cdot \semi \cdot)$.

\item Case ($\m{proc}_\p$) : We have $\G_0 \potconf{E + q + w}
\W, \proc{x_\p}{w, P} :: (\G \semi \D, (x_\p : A_\p))$. Inverting
the $\m{proc}_\p$ rule,
\[
\infer[\m{proc}_\p]
{\G_0 \potconf{E + q + w} \W, \proc{x_\p}{w, P} ::
(\G \semi \D, (x_\p : A_\p))}
{\G_0 \potconf{E} \W :: (\G \semi \D, \D_\p') \qquad
\cdot \semi \cdot \semi \D_\p' \entailpot{q} P :: (x_\p : A_\p)}
\]
we get $\G_0 \potconf{E} \W ::
(\G \semi \D, \D_\p')$. By the induction hypothesis,
$\G_0, (c_\s : B_\s) \potconf{E} \W :: (\G \semi \D, \D_\p')$.
Applying the $\m{proc}_\p$ rule,
\[
\infer[\m{proc}_\p]
{\G_0, (c_\s : B_\s) \potconf{E + q + w} \W, \proc{x_\p}{w, P} ::
(\G \semi \D, (x_\p : A_\p))}
{\G_0, (c_\s : B_\s) \potconf{E} \W :: (\G \semi \D, \D_\p') \qquad
\cdot \semi \cdot \semi \D_\p' \entailpot{q} P :: (x_\p : A_\p)}
\]

\item Case ($\m{proc}_\s$) : We have
$\G_0 \potconf{E + q + w} \W, \proc{x_\s}{w, P} ::
(\G, (x_\s : A_\s))$. Inverting the $\m{proc}_\s$ rule,
\[
\infer[\m{proc}_\s]
{\G_0 \potconf{E + q + w} \W, \proc{x_\s}{w, P} ::
(\G, (x_\s : A_\s)
\semi \D)}
{(x_\s : A_\s) \in \G_0 \qquad
(A_\s, A_\s) \; \m{esync} \qquad
\G_0 \potconf{E} \W :: (\G \semi \D,\D_\p') \qquad
\cdot \semi \D_\p' \entailpot{q} P :: (x_\s : A_\s)}
\]
we get $\G_0 \potconf{E} \W :: (\G \semi \D,\D_\p')$. By the
induction hypothesis,
$\G_0, (c_\s : B_\s) \potconf{E} \W :: (\G \semi \D,\D_\p')$.
Also, by Lemma~\ref{lem:proc-weakening}, we get
$\cdot \semi \G_0, (c : B_\s) \semi \D_\p' \entailpot{q} P :: (x_\s : A_\s)$.
Applying the $\m{proc}_\s$ rule back,
\[
\inferrule*[right = $\m{proc}_\s$]
{(x_\s : A_\s) \in \G_0, (c_\s : B_\s) \qquad
(A_\s, A_\s) \; \m{esync} \qquad
\G_0, (c_\s : B_\s) \potconf{E} \W :: (\G \semi \D,\D_\p') \\
\cdot \semi \G_0, (c : B_\s) \semi \D_\p' \entailpot{q} P :: (x_\s : A_\s)}
{\G_0, (c_\s : B_\s) \potconf{E + q + w} \W, \proc{x_\s}{w, P} ::
(\G, (x_\s : A_\s)
\semi \D)}
\]

\item Case ($\m{proc}_\l$) : We have
$\G_0 \potconf{E + q + w} \W, \proc{x_\l}{w, P} ::
(\G, (x_\s : A_\s) \semi \D, (x_\l : A_\l))$.
Inverting the $\m{proc}_\l$ rule,
\[
\infer[\m{proc}_\l]
{\G_0 \potconf{E + q + w} \W, \proc{x_\l}{w, P} ::
(\G, (x_\s : A_\s) \semi \D, (x_\l : A_\l))}
{(x_\s : A_\s) \in \G_0 \qquad
(A_\l, A_\s) \; \m{esync}  \qquad
\G_0 \potconf{E} \W :: (\G \semi \D, \D') \qquad
\cdot \semi \G_0 \semi \D' \entailpot{q} P :: (x_\l : A_\l)}
\]
we get $\G_0 \potconf{E} \W :: (\G \semi \D, \D')$.
Applying the induction hypothesis, we get
$\G_0, (c_\s : B_\s) \potconf{E} \W ::
(\G \semi \D, \D')$. Using Lemma~\ref{lem:proc-weakening},
we get $\cdot \semi \G_0, (c_\s : B_\s) \semi \D' \entailpot{q}
P :: (x_\l : A_\l)$. Applying the $\m{proc}_\l$ rule back,
\[
\inferrule*[right = $\m{proc}_\l$]
{(x_\s : A_\s) \in \G_0, (c_\s : B_\s) \qquad
(A_\l, A_\s) \; \m{esync}  \qquad
\G_0, (c_\s : B_\s) \potconf{E} \W :: (\G \semi \D, \D') \\
\cdot \semi \G_0, (c_\s : B_\s) \semi \D' \entailpot{q} P :: (x_\l : A_\l)}
{\G_0, (c_\s : B_\s) \potconf{E + q + w} \W, \proc{x_\l}{w, P} ::
(\G, (x_\s : A_\s) \semi \D, (x_\l : A_\l))}
\]

\item Case ($\m{proc}_\c$) : We have
$\G_0 \potconf{E+q+w} \W, \proc{x_\c}{w, P} :: (\G \semi
\D, (x_\c : A_\c))$. Inverting the $\m{proc}_\c$ rule,
\[
\infer[\m{proc}_\c]
{\G_0 \potconf{E+q+w} \W, \proc{x_\c}{w, P} :: (\G \semi
\D, (x_\c : A_\c))}
{\G_0 \potconf{E} \W :: (\G \semi \D, \D') \qquad
\cdot \semi \G_0 \semi \D' \entailpot{q} P :: (x_\c : A_\c) }
\]
we get $\G_0 \potconf{E} \W :: (\G \semi \D, \D')$. By the
induction hypothesis, we get
$\G_0, (c_\s : B_\s) \potconf{E} \W :: (\G \semi \D, \D')$.
Also, using Lemma~\ref{lem:proc-weakening}, we get
$\cdot \semi \G_0, (c_\s : B_\s) \semi \D' \entailpot{q} P :: (x_\c : A_\c)$
Applying the $\m{proc}_\c$ rule back,
\[
\infer[\m{proc}_\c]
{\G_0, (c_\s : B_\s) \potconf{E+q+w} \W, \proc{x_\c}{w, P} ::
(\G \semi \D, (x_\c : A_\c))}
{\G_0, (c_\s : B_\s) \potconf{E} \W :: (\G \semi \D, \D') \qquad
\cdot \semi \G \semi \D' \entailpot{q} P :: (x_\c : A_\c) }
\]

\item Case ($\m{msg}$) : We have 
$\G_0 \potconf{E+q+w} \W, \msg{x_m}{w, M} :: (\G \semi
\D, (x_m : A))$. Inverting the $\m{msg}$ rule,
\[
\infer[\m{msg}]
{\G_0 \potconf{E+q+w} \W, \msg{x_m}{w, M} :: (\G \semi
\D, (x_m : A))}
{\G_0 \potconf{E} \W :: (\G \semi \D, \D') \qquad
\cdot \semi \cdot \semi \D' \entailpot{q} M :: (x_m : A)}
\]
$\G_0 \potconf{E} \W :: (\G \semi \D, \D')$. By the
induction hypothesis, $\G_0, (c_\s : B_\s) \potconf{E+q+w}
\W, \msg{x_m}{w, M} :: (\G \semi \D, (x_m : A))$. Applying
the $\m{msg}$ rule back,
\[
\infer[\m{msg}]
{\G_0, (c_\s : B_\s) \potconf{E+q+w} \W, \msg{x_m}{w, M} ::
(\G \semi \D, (x_m : A))}
{\G_0, (c_\s : B_\s) \potconf{E} \W :: (\G \semi \D, \D') \qquad
\cdot \semi \cdot \semi \D' \entailpot{q} M :: (x_m : A)}
\]
\end{itemize}
\end{proof}

\begin{lemma}[Process Weakening]\label{lem:proc-weakening}
For a well-typed process $\G \semi \D \entailpot{q} P :: (x_\c : A)$ and for
a shared channel $c_\s : A_\s \notin \G$, we have $\G, (c_\s : A_\s) \semi \D
\entailpot{q} P :: (x_\c : A)$.
\end{lemma}

\begin{proof}
Analogous to Lemma~\ref{lem:conf-weakening}.
\end{proof}

\begin{lemma}[Permutation-Message]\label{lem:perm-msg}
Consider a well-typed configuration typed by the judgment
$\G_0 \potconf{E} \W_1,$
$\msg{c_m}{w, M}, \W_2, \proc{d_k}{w', P(c_m)} ::
(\G \semi \D)$. Then, the message can be moved right such that
the configuration $\G_0 \potconf{E} \W_1, \W_2, \msg{c_m}{w, M},
\proc{d_k}{w', P(c_m)} :: (\G \semi \D)$ is well-typed.
\end{lemma}

\begin{proof}
We case analyze on the structure of the message.

\begin{itemize}
\item Case ($\tensor_n$) : We have $\G_0 \potconf{E} \W_1,
\msg{c_m}{w, \esendch{c_m}{e_n} \semi \fwd{c_m}{c^+_m}}, \W_2,
\proc{d_k}{w', P(c_m)} :: (\G \semi \D)$. First, we type the message
\[
\cdot \semi \cdot \semi (c^+_m : B), (e_n : A) \entailpot{q}
\esendch{c_m}{e_n} \semi \fwd{c_m}{c^+_m} :: (c_m : A \tensor_n B)
\]
Next, we invert the $\m{msg}$ rule,
\[
\inferrule*[right = $\m{msg}$]
{\G_0 \potconf{E} \W_1 :: (\G \semi \D, (c^+_m : B), (e_n : A)) \\
\cdot \semi \cdot \semi (c^+_m : B), (e_n : A) \entailpot{q}
\esendch{c_m}{e_n} \semi \fwd{c_m}{c^+_m} :: (c_m : A \tensor_n B)}
{\G_0 \potconf{E+q+w} \W_1, \msg{c_m}{w, \esendch{c_m}{e_n}} ::
(\G \semi \D, (c_m : A \tensor_n B))}
\]
Since the channel $c_m$ is only used by $\proc{d_k}{w', P(c_m)}$,
we know that none of the processes or messages in $\W_2$ can use it.
Hence, we can move the message just left of the process
$\proc{d_k}{w', P(c_m)}$.
\end{itemize}
\end{proof}

\begin{lemma}[Permutation-Process]\label{lem:perm-proc}
Consider a well-typed configuration typed by the judgment
$\G_0 \potconf{E} \W_1, \proc{c_m}{w, P}, \W_2,
\msg{c^+_m}{w', M(c_m)} :: (\G \semi \D)$. Then, the
process can be moved right such that
the configuration $\G_0 \potconf{E} \W_1, \W_2, \proc{c_m}{w, P},
\msg{c^+_m}{w', M(c_m)} :: (\G \semi \D)$ is well-typed.
\end{lemma}

\begin{proof}
We case analyze on the structure of the message.
\begin{itemize}
\item Case ($\lolli_n$) : We have $\G_0 \potconf{E} \W_1,
\proc{c_m}{w, P}, \W_2, \msg{c^+_m}{w', \esendch{c_m}{e_n} \semi
\fwd{c^+_m}{c_m}} :: (\G \semi \D)$. First, we type the message
\[
\cdot \semi \cdot \semi (e_n : A), (c_m : A \lolli_n B) \entailpot {q}
\esendch{c_m}{e_n} \semi \fwd{c^+_m}{c_m} :: (c^+_m : B)
\]
Since the message is the only provider of channel $c_m$ offered
by $\proc{c_m}{w, P}$, we know that none of the processes in
$\W_2$ can depend on it. Thus, the process can be moved to the
without affecting the invariant for any process in $\W_2$.
\end{itemize}
\end{proof}

\begin{lemma}[Permutation-Acquire]\label{lem:perm-acquire}
Consider a well-typed configuration typed by the judgment
$\G_0 \potconf{E} \W_1, \proc{c_m}{w', \eacquire{a_\l}{a_\s} \semi Q},
\W_2, \proc{a_\s}{w, \eaccept{a_\l}{a_\s} \semi P}, \W_3 :: (\G \semi \D)$.
Then, the acquiring process can be moved right such that
the configuration $\G_0 \potconf{E} \W_1, \W_2,
\proc{a_\s}{w, \eaccept{a_\l}{a_\s} \semi P},$
$\proc{c_m}{w', \eacquire{a_\l}{a_\s} \semi Q}, \W_3 ::
(\G \semi \D)$ is well-typed.
\end{lemma}

\begin{proof}
Due to independence, we know that $\proc{a_\s}{w,
\eaccept{a_\l}{a_\s} \semi P}$ can only depend on any channels
at mode $\s$ or $\p$. On the other hand, $m$ can only be $\c$
or $\l$. In particular, the shared process cannot depend on channel
$c_m$, thus the acquiring process can be moved to the right of the
shared process.
\end{proof}

\begin{lemma}[Permutation-Release]\label{lem:perm-release}
Consider a well-typed configuration typed by the judgment
$\G_0 \potconf{E} \W_1, \proc{c_m}{w', \erelease{a_\s}{a_\l} \semi Q},
\W_2, \proc{a_\l}{w, \edetach{a_\s}{a_\l} \semi P}, \W_3 :: (\G \semi \D)$.
Then, the releasing process can be moved right such that
the configuration $\G_0 \potconf{E} \W_1, \W_2,
\proc{a_\l}{w, \edetach{a_\s}{a_\l} \semi P},$
$\proc{c_m}{w', \erelease{a_\s}{a_\l} \semi Q}, \W_3 ::
(\G \semi \D)$ is well-typed.
\end{lemma}

\begin{proof}
Due to independence, we know that $\proc{a_\l}{w,
\edetach{a_\s}{a_\l} \semi P}$ can only depend on any channels
at mode $\s$ or $\p$. On the other hand, $m$ can only be $\c$
or $\l$. In particular, the shared process cannot depend on channel
$c_m$, thus the releasing process can be moved to the right of the
detaching process.
\end{proof}

\begin{lemma}[Shared-Substitution]\label{lem:subst}
If the process $\G, (b_\s : B_\s), (x_\s : B_\s) \semi \D
\entailpot{q} P_{x_\s} :: (z_m : C)$ is well-typed, then
$\G, (b_\s : B_\s) \semi \D \entailpot{q} P_{b_\s} :: (z_m : C)$
is also well-typed.
\end{lemma}

\begin{proof}
We apply induction on the process typing judgment.
\begin{itemize}
\item Case ($\{\}E_{\c\c}$) : 
\[
\inferrule*[right = $\{\}E_{\c\c}$]
{r = p+q \qquad
\G, (b_\s : B_\s), (x_\s : B_\s) \supseteq \overline{a_\s : A} \qquad
\D = \overline{d : D} \\\\
\Psi \exppot{p} M : \tproc{A}{\overline{A} \semi \overline{D}}_\c
\qquad \Psi \semi \G, (b_\s : B_\s), (x_\s : B_\s) \semi \D', (y_\c : A)
\entailpot{q} Q_{x_\s} :: (z_\c : C)}
{\Psi \semi \G, (b_\s : B_\s), (x_\s : B_\s) \semi \D, \D' \entailpot{r}
\ecut{y_\c}{M}{a_\s \semi d}{Q_{x_\s}} :: (z_\c : C)}
\]
By the induction hypothesis,
$\Psi \semi \G, (b_\s : B_\s) \semi \D', (y_\c : A)
\entailpot{q} Q_{b_\s} :: (z_\c : C)$. We simply substitute $b_\s$
for $x_\s$ in $\overline{a_\s : A}$. Hence, $\G, (b_\s : B_\s)
\supseteq [b_\s/x_\s]\overline{a_\s : A}$. Applying the $\{\}E_{\c\c}$
rule back
\[
\inferrule*[right = $\{\}E_{\c\c}$]
{r = p+q \qquad
\G, (b_\s : B_\s) \supseteq [b_\s/x_\s]\overline{a_\s : A} \qquad
\D = \overline{d : D} \\\\
\Psi \exppot{p} M : \tproc{A}{\overline{A} \semi \overline{D}}_\c
\qquad \Psi \semi \G, (b_\s : B_\s) \semi \D', (y_\c : A)
\entailpot{q} Q_{b_\s} :: (z_\c : C)}
{\Psi \semi \G, (b_\s : B_\s) \semi \D, \D' \entailpot{r}
\ecut{y_\c}{M}{[b_\s/x_\s]a_\s \semi d}{Q_{x_\s}} :: (z_\c : C)}
\]

\item Case ($\m{fwd}$) : 
\[
\Psi \semi \G, (b_\s : B_\s), (x_\s : B_\s) \semi (y_k : A)
\entailpot{q} \fwd{z_m}{y_k} :: (z_m : A)
\]
Here, the lemma holds trivially since $x_\s$ doesn't occur
in $P_{x_\s}$. Therefore, $P_{x_\s} = P_{b_\s}$ and
\[
\Psi \semi \G, (b_\s : B_\s) \semi (y_k : A)
\entailpot{q} \fwd{z_m}{y_k} :: (z_m : A)
\]

\item Case ($\lolli_n R$) : 
\[
\infer[\lolli_n R]
{\Psi \semi \G, (b_\s : B_\s), (x_\s : B_\s) \semi \D \entailpot{q}
\erecvch{z_m}{y_n} \semi P_{x_\s} :: (z_m : A \lolli_n B)}
{\Psi \semi \G, (b_\s : B_\s), (x_\s : B_\s) \semi \D, (y_n : A)
\entailpot{q} P_{x_\s} :: (z_m : B)}
\]
By the induction hypothesis,
$\Psi \semi \G, (b_\s : B_\s) \semi \D, (y_n : A)
\entailpot{q} P_{b_\s} :: (z_m : B)$. Applying the $\lolli R$ rule,
\[
\infer[\lolli_n R]
{\Psi \semi \G, (b_\s : B_\s) \semi \D \entailpot{q}
\erecvch{z_m}{y_n} \semi P_{b_\s} :: (z_m : A \lolli_n B)}
{\Psi \semi \G, (b_\s : B_\s) \semi \D, (y_n : A)
\entailpot{q} P_{b_\s} :: (z_m : B)}
\]

\item Case ($\lolli_n L$) : 
\[
\infer[\lolli L]
{\Psi \semi \G, (b_\s : B_\s), (x_\s : B_\s) \semi
\D, (w_n : A), (y_k : A \lolli B)
\entailpot{q} \esendch{y_k}{w_n} \semi Q_{x_\s} :: (z_m : C)}
{\Psi \semi \G, (b_\s : B_\s), (x_\s : B_\s) \semi
\D, (y_k : B) \entailpot{q} Q_{x_\s} :: (z_m : C)}
\]
By the induction hypothesis,
$\Psi \semi \G, (b_\s : B_\s) \semi
\D, (y_k : B) \entailpot{q} Q_{b_\s} :: (z_m : C)$.
Applying the $\lolli_n L$ rule,
\[
\infer[\lolli_n L]
{\Psi \semi \G, (b_\s : B_\s) \semi
\D, (w_n : A), (y_k : A \lolli_n B)
\entailpot{q} \esendch{y_k}{w_n} \semi Q_{b_\s} :: (z_m : C)}
{\Psi \semi \G, (b_\s : B_\s) \semi
\D, (y_k : B) \entailpot{q} Q_{b_\s} :: (z_m : C)}
\]

\item Case ($\up L$) :
\[
\infer[\up L]
{\Psi \semi \G, (b_\s : \up A_\l), (x_\s : \up A_\l) \semi \D
\entailpot{q} \eacquire{x_\l}{x_\s} \semi Q :: (z_m : C)}
{\Psi \semi \G, (b_\s : \up A_\l) \semi \D, (x_\l : A_\l)
\entailpot{q} Q :: (z_m : C)}
\]
The lemma holds trivially since $x_\s$ doesn't occur in $Q$.
Hence, $[b_\s/x_\s]Q = Q$. Applying the $\up L$ rule,
\[
\infer[\up L]
{\Psi \semi \G, (b_\s : \up A_\l) \semi \D
\entailpot{q} \eacquire{x_\l}{b_\s} \semi Q :: (z_m : C)}
{\Psi \semi \G, (b_\s : \up A_\l) \semi \D, (x_\l : A_\l)
\entailpot{q} Q :: (z_m : C)}
\]

\item Case ($\down L$) :
\[
\infer[\down L]
{\Psi \semi \G, (b_\s : B_\s), (x_\s : B_\s) \semi \D,
(y_\l : \down A_\s) \entailpot{q} \erelease{y_\s}{y_\l}
\semi Q_{x_\s} :: (z_m : C)}
{\Psi \semi \G, (b_\s : B_\s), (x_\s : B_\s), (y_\s : A_\s) \semi \D
\entailpot{q} Q_{x_\s} :: (z_m : C)}
\]
By the induction hypothesis,
$\Psi \semi \G, (b_\s : A_\s), (y_\s : A_\s) \semi \D
\entailpot{q} Q_{b_\s} :: (z_m : C)$. Applying the $\down L$ rule,
\[
\infer[\down L]
{\Psi \semi \G, (b_\s : B_\s) \semi \D,
(y_\l : \down A_\s) \entailpot{q} \erelease{y_\s}{y_\l}
\semi Q_{b_\s} :: (z_m : C)}
{\Psi \semi \G, (b_\s : A_\s), (y_\s : A_\s) \semi \D
\entailpot{q} Q_{b_\s} :: (z_m : C)}
\]
\end{itemize}
\end{proof}

\begin{lemma}[Variable Substitution]\label{lem:var_subst}
To substitute value for a variable from the functional context,
we need the following two lemmas.
\begin{itemize}
\item If $V \val$ and $\cdot \exppot{p} V : \tau$ and
$\Psi, (x : \tau) \exppot{q} M : \sigma$, then
$\Psi \exppot{p+q} [V/x]M : \sigma$.

\item If $V \val$ and $\cdot \exppot{p} V : \tau$ and
$\Psi, (x : \tau) \semi \G \semi \D \exppot{q} P :: (c : A)$, then
$\Psi \semi \G \semi \D \exppot{p+q} [V/x]P :: (c : A)$
\end{itemize}
\end{lemma}

\begin{theorem}[Expression Preservation]\label{thm:exp_pres}
If a well-typed expression $\cdot \exppot{q} N : \tau$ takes a step, i.e.,
$\bigeval{N}{V}{\mu}$, then $V \val$ and $q \geq \mu$ and
$\cdot \exppot{q-\mu} V : \tau$.
\end{theorem}

\begin{theorem}[Process Preservation]\label{thm:proc_pres}
Consider a closed well-formed and well-typed configuration $\W$
such that $\G_0 \potconf{E} \W :: (\G \semi \D)$. If the
configuration takes a step, i.e. $\W \step \W'$, then there exist
$\G_0', \G'$ such that $\G_0' \potconf{E} \W' :: (\G' \semi \D)$, i.e., the
resulting configuration is well-typed.
\end{theorem}

\begin{proof}
We case analyze on the semantics.

\begin{itemize}
\item Case ($\m{internal}$) : $\W = \dc, \proc{c_m}{w, P[N]}$
and $\W' = \dc, \proc{c_m}{w + \mu, P[V]}$. We case analyze
on $P[N]$.

\begin{itemize}

\item Case ($\arrow \m{send}$) :  $P[N] = \esendch{d_k}{N} \semi P$ and
$P[V] = \esendch{d_k}{V} \semi P$, where $\bigeval{N}{V}{\mu}$.
Suppose, $\G_0 \potconf{E+r+w} \dc, \proc{c_m}{w, \esendch{d_k}{N}
\semi P} :: (\G \semi \D, (c_m : C))$. Inverting the $\m{proc}_m$ rule,
\[
\inferrule*[right = $\m{proc}_m$]
{\G_0 \potconf{E} \dc :: (\G \semi \D_1, (d_k : \tau \arrow A), \D) \\
\infer[\arrow L]
{\cdot \semi \G_0 \semi \D_1, (d_k : \tau \arrow A) \entailpot{r}
\esendch{d_k}{N} \semi P :: (c_m : C)}
{r = p+q \quad
\cdot \exppot{p} N : \tau \quad
\cdot \semi \G_0 \semi \D, (d_k : A) \entailpot{q} P :: (c_m : C)}}
{\G_0 \potconf{E+r+w} \dc,\proc{c_m}{w, \esendch{d_k}{N} \semi P}
:: (\G \semi \D, (c_m : C))}
\]
By Theorem~\ref{thm:exp_pres}, we get that
$\cdot \exppot{p - \mu} V : \tau$.
Finally, we apply the same derivation again to get
\[
\inferrule*[right = $\m{proc}_m$]
{\G_0 \potconf{E} \dc :: (\G \semi \D_1, (d_k : \tau \arrow A), \D) \\
\inferrule*[right = $\arrow L$]
{r' = p-\mu+q \\
\cdot \exppot{p-\mu} V : \tau \\\\
\cdot \semi \G_0 \semi \D, (d_k : A) \entailpot{q} P :: (c_m : C)}
{\cdot \semi \G_0 \semi \D_1, (d_k : \tau \arrow A) \entailpot{r'}
\esendch{d_k}{V} \semi P :: (c_m : C)}
}
{\G_0 \potconf{E+r'+w+\mu} \proc{c_m}{w+\mu, \esendch{d_k}{N} \semi P},
\dc :: (\G \semi \D, (c_m : C))}
\]
and the proof succeeds since $r'+w+\mu = p-\mu+q+w+\mu=p+q+w=r+w$.

\item Case ($\product \m{send}$) : Analogous to $\arrow \m{send}$.

\item Case ($E_{\s m}$) : $\W = \dc, \dc, \proc{c_m}{w, \ecut{d_\s}{N}
{\overline{a_\s} \semi \overline{a_\p}}{Q}}$
and $\W' = \dc, \proc{c_m}{w+\mu, \ecut{d_\s}{V}{\overline{a_\s} \semi
\overline{a_\p}}{Q}}$
where $\bigeval{N}{V}{\mu}$. Inverting the $\m{proc_m}$ rule,
\[
\inferrule*[right=$\m{proc}_m$]
{\G_0 \potconf{E} \dc :: (\G \semi \D, \D_1, \D_2) \\
\inferrule*[right = $E_{\s m}$]
{r = p+q \quad
\G_0 \supseteq \overline{a_\s : A} \quad
\D_1 = \overline{a_\p : D} \\\\
\cdot \exppot{p} N : \tproc{A_\s}{\overline{A} \semi \overline{D}}_\s \quad
\cdot \semi \G_0, (d_\s : A_\s) \semi \D_2 \entailpot{q} Q :: (c_m : C)}
{\cdot \semi \G_0 \semi \D_1, \D_2 \entailpot{r}
\ecut{d_\s}{N}{\overline{a_\s} \semi \overline{a_\p}}{Q} :: (c_m : C)}
}
{\G_0 \potconf{E+r+w}
\dc, \proc{c_m}{w, \ecut{d_\s}{N}{\overline{a_\s} \semi \overline{a_\p}}{Q}}
:: (\G \semi \D, c_m : C)}
\]
By Theorem~\ref{thm:exp_pres}, $\cdot \exppot{p-\mu} V :
\tproc{A_\s}{\overline{D}}_\s$. Applying the same derivation back,
\[
\inferrule*[right=$\m{proc}_m$]
{\G_0 \potconf{E} \dc :: (\G \semi \D, \D_1, \D_2) \\
\inferrule*[right = $E_{\s m}$]
{r' = p-\mu+q \quad
\G_0 \supseteq \overline{a_\s : A} \quad
\D_1 = \overline{a_\p : D} \\\\
\cdot \exppot{p-\mu} V : \tproc{A_\s}{\overline{A} \semi
\overline{D}}_\s \quad
\cdot \semi \G_0, (d_\s : A_\s) \semi \D_2 \entailpot{q} Q :: (c_m : C)}
{\cdot \semi \G_0 \semi \D_1, \D_2 \entailpot{r'}
\ecut{d_\s}{V}{\overline{a_\s} \semi \overline{a_\p}}{Q} :: (c_m : C)}
}
{\G_0 \potconf{E+r'+w+\mu}
\dc, \proc{c_m}{w+\mu, \ecut{d_\s}{V}{\overline{a_\s} \semi \overline{a_\p}}
{Q}} :: (\G \semi \D, c_m : C)}
\]
and the proof succeeds since $r'+w+\mu = p-\mu+q+w+\mu =
p+q+w = r+w$.

\item Case ($E_{\p m}, E_{\c\c}$) : Analogous to $E_{\s m}$.
\end{itemize}

\item Case ($\{\}E_{\s\c}$) : $\W = \dc, \proc{d_\c}{w, \ecut{x_\s}{\eproc{x_\s'}
{P_{x_\s',\overline{y},\overline{z}}}{\overline{y} \semi \overline{z}}}
{\overline{a} \semi \overline{b}}{Q}}$ and $\W' = \dc,
\proc{c_\s}{0, P_{c_\s,\overline{a},\overline{b}}},$ 
$\proc{d_\c}{w, [c_\s/x_\s]Q}$.
Inverting the $\m{proc}_\c$ rule,
\[
\inferrule*[right = $\m{proc}_\c$]
{\G_0 \potconf{E} \dc :: (\G \semi \D, \D_1, \D_2) \\
\inferrule*[right = $\{\}E_{\s C}$]
{\infer[\{\}I_\s]
{\cdot \exppot{p} \eproc{x_\s'}
{P_{x_\s',\overline{y},\overline{z}}}{\overline{y} \semi \overline{z}} :
\tproc{A_\s}{\overline{A} \semi \overline{D}}_\s}
{\G_y = \overline{y : A} \qquad
\D_z = \overline{z : D} \qquad \cdot \semi \G_y \semi \D_z \entailpot{p}
P_{x_\s',\overline{y},\overline{z}} :: (x'_\s : A_\s)} \\\\
r=p+q \qquad
\G_0 \supseteq \overline{a : A} \qquad
\D_1 = \overline{b : D} \qquad (A_\s, A_\s) \esync \\\\
\cdot \semi \G_0, (x_\s : A_\s) \semi \D_2 \entailpot{q} Q :: (d_\c : A_\c)}
{\cdot \semi \G_0 \semi \D_1, \D_2 \entailpot{r} \ecut{x_\s}{\eproc{x_\s'}
{P_{x_\s',\overline{y},\overline{z}}}{\overline{y} \semi \overline{z}}}
{\overline{a} \semi \overline{b}}{Q} :: (d_\c : A_\c)}
}
{\G_0 \potconf{E+r+w} \dc, \proc{d_\c}{w, \ecut{x_\s}{\eproc{x_\s'}
{P_{x_\s',\overline{y},\overline{z}}}{\overline{y} \semi \overline{z}}}
{\overline{a} \semi \overline{b}}{Q}} :: (\G \semi
\D, (d_\c : A_\c))}
\]
The premise for $\{\}I_\s$ gives us $\cdot \semi \G_y \semi \D_z \entailpot{p}
P_{x_\s',\overline{y},\overline{z}} :: (x'_\s : A_\s)$, which by Lemma~\ref{lem:renaming},
gives us $\cdot \semi \G_0 \semi \D_1 \entailpot{p}
P_{c_\s,\overline{a},\overline{b}} :: (c_\s : A_\s)$. Then, by Lemma
\ref{lem:proc-weakening}, we get $\cdot \semi \G_0, (c_\s : A_\s) \semi \D_1 \entailpot{p}
P_{c_\s,\overline{a},\overline{b}} :: (c_\s : A_\s)$
Similarly, we get $\cdot \semi \G_0, (c_\s : A_\s) \semi \D_2 \entailpot{q}
[c_\s/x_\s] Q :: (d_\c : A_\c)$.
First, using Lemma~\ref{lem:conf-weakening}, we get $\G_0, (c_\s : A_\s)
\potconf{E} \dc :: (\G \semi \D, \D_1, \D_2)$. Next, apply the
$\m{proc}_\s$ rule,
\[
\infer[\m{proc}_\s]
{\G_0, (c_\s : A_\s) \potconf{E+p+0} \dc, \proc{c_\s}{0,
P_{c_\s,\overline{a},\overline{b}}} ::
(\G, (c_\s : A_\s) \semi \D, \D_2)}
{\G_0, (c_\s : A_\s) \potconf{E} \dc :: (\G \semi \D, \D_1, \D_2) \qquad
\cdot \semi \G_0, (c_\s : A_\s) \semi \D_1 \entailpot{p}
P_{c_\s,\overline{a},\overline{b}} :: (c_\s : A_\s)
}
\]
Call this new configuration $\dc'$. Now, apply the $\m{proc}_\c$ rule.
\[
\infer[\m{proc}_\c]
{\G_0, (c_\s : A_\s) \potconf{E+p+q+w} \dc', \proc{d_\c}{w, [c_\s/x_\s]Q} ::
(\G, (c_\s : A_\s) \semi \D, (d_\c : A_\c))}
{\G_0, (c_\s : A_\s) \potconf{E+p+0} \dc' :: (\G, (c_\s : A_\s) \semi \D, \D_2) \qquad
\cdot \semi \G, (c_\s : A_\s) \semi \D_2 \entailpot{q}
[c_\s/x_\s] Q :: (d_\c : A_\c)}
\]
where $E+p+q+w = E+r+w$ since $r = p+q$. Hence, in this case
$\G_0' = \G_0, (c_\s : A_\s)$ and $\G' = \G, (c_\s : A_\s)$.

\item Case ($\{\}E_{\c\c}$) : $\W = \dc, \proc{d_\c}{w, \ecut{x_\c}{\eproc{x_\c'}
{P_{x_\c',\overline{y}, \overline{z}}}{\overline{y} \semi \overline{z}}}{a_\s \semi d}{Q}}$
and $\W' = \dc, \proc{c_\c}{0, P_{c_\c,\overline{a_\s}, \overline{d}}},$
$\proc{d_\c}{w, [c_\c/x_\c]Q}$. Inverting the $\m{proc}_\c$ rule
\[
\inferrule*[right = $\m{proc}_\c$]
{\G_0 \potconf{E} \dc :: (\G \semi \D, \D_1, \D_2) \\
\inferrule*[right = $\{\}E_{\c\c}$]
{
\inferrule*[right = $\{\}I_\c$]
{\G_y = \overline{y : A} \\
\D_z = \overline{z : D} \\
\cdot \semi \G_y \semi \D_z \entailpot{p} P_{x'_\c, \overline{y},
\overline{z}} :: (x'_\c : A)}
{\cdot \exppot{p} \eproc{x_\c'}
{P_{x_\c',\overline{y}, \overline{z}}}{\overline{y} \semi
\overline{z}} : \tproc{A}{\overline{A} \semi \overline{D}}_\c}
\\\\
r = p+q \\
\G_0 \supseteq \overline{a_\s : A} \\
\D_1 = \overline{d : D} \\
\cdot \semi \G_0 \semi \D_2, (x_\c : A)
\entailpot{q} Q :: (d_\c : C)}
{\cdot \semi \G_0 \semi \D_1, \D_2 \entailpot{r}
\ecut{x_\c}{\eproc{x_\c'}
{P_{x_\c',\overline{y}, \overline{z}}}{\overline{y} \semi
\overline{z}}}{\overline{a_\s} \semi \overline{d}}{Q} :: (d_\c : C)}}
{\G_0 \potconf{E+r+w} \dc, \proc{d_\c}{w, \ecut{x_\c}{\eproc{x_\c'}
{P_{x_\c',\overline{y}, \overline{z}}}
{\overline{y} \semi \overline{z}}}{\overline{a_\s} \semi \overline{d}}{Q}} ::
(\G \semi \D, (d_\c : C))}
\]
We contract all multiple occurrences of the same channel in
$\overline{a_\s : A}$. Let the resulting vector be $\G' = 
\overline{a'_\s : A'}$. We know,
by Lemma~\ref{lem:subst} that $\cdot \semi \G' \semi \D'
\entailpot{p} P_{x'_\c, \overline{a'_\s}, \overline{z}} :: (x'_\c : A)$
is well-typed. Next, by Lemma~\ref{lem:renaming}, we get
$\G' \semi \D_1 \entailpot{p} P_{c_\c, \overline{a'_\s}, \overline{d}}
:: (c_\c : A)$. Finally, we weaken $\G'$ using Lemma
\ref{lem:proc-weakening} to get $\cdot \semi \G_0 \semi \D_1
\entailpot{p} P_{c_\c, \overline{a'_\s}, \overline{d}} :: (c_\c : A)$.
Also, note that since $\overline{a'_\s}$ is a refinement of
$\overline{a_\s}$ by eliminating duplicates, 
$P_{c_\c, \overline{a'_\s}, \overline{d}} = 
P_{c_\c, \overline{a_\s}, \overline{d}}$.
Hence, we apply the $\m{proc}_\c$ rule,
\[
\inferrule*[right = $\m{proc}_\c$]
{\G_0 \potconf{E} \dc :: (\G \semi \D, \D_1, \D_2) \\
\cdot \semi \G_0 \semi \D_1 \entailpot{p}
P_{c_\c, \overline{a_\s}, \overline{d}} :: (c_\c : A)}
{\G_0 \potconf{E+p+0} \dc, \proc{c_\c}{0,
P_{c_\c, \overline{a_\s}, \overline{d}}} :: (\G \semi \D, \D_2,
(c_\c : A))}
\]
Call this new configuration $\dc'$. Applying renaming using
Lemma~\ref{lem:renaming}, we get $\cdot \semi \G_0 \semi
\D_2, (c_\c : A) \entailpot{q} [c_\c/x_\c]Q :: (d_\c : C)$.
Again, applying the $\m{proc}_\c$ rule, we get
\[
\inferrule*[right = $\m{proc}_\c$]
{\G_0 \potconf{E+p+0} \dc' :: (\G \semi \D, \D_2, (c_\c : A)) \\
\cdot \semi \G_0 \semi \D_2, (c_\c : A) \entailpot{q} [c_\c/x_\c]Q ::
(d_\c : C)}
{\G_0 \potconf{E+p+q+w} \dc', \proc{d_\c}{w, [c_\c/x_\c]Q} ::
(\G \semi \D, (d_\c : C))}
\]
where $E+p+q+w = E+r+w$ since $r = p+q$.

\item Case ($\m{fwd}^+$) : $\W = \dc, \msg{d_k}{w', M},
\proc{c_m}{w, \fwd{c_m}{d_k}}$ and \\
$\W' = \msg{c_m}{w + w',
[c_m/d_k]M}$. First, inverting the $\m{msg}$ rule,
\[
\inferrule*[right = $\m{msg}$]
{\G_0 \potconf{E} \dc :: (\W \semi \D, \D_1) \\
\cdot \semi \cdot \semi \D_1 \entailpot{q} M :: (d_k : A)}
{\G_0 \potconf{E+q+w'} \dc, \msg{d_k}{w', M} :: (\G \semi
\D, (d_k : A))}
\]
Call this new configuration $\dc'$. Next, inverting the
$\m{proc}_m$ rule
\[
\inferrule*[right = $\m{proc}_m$]
{\G_0 \potconf{E+q+w'} \dc' :: (\G \semi \D, (d_k : A)) \\
\cdot \semi \G_0 \semi (d_k : A) \entailpot{0} \fwd{c_m}{d_k} :: (c_m : A)}
{\G_0 \potconf{E+q+w'+0+w} \dc', \proc{c_m}{w, \fwd{c_m}{d_k}}
:: (\G \semi \D, (c_m : A))}
\]
Using Lemma~\ref{lem:renaming}, we get $\cdot \semi \cdot \semi
\D_1 \entailpot{q} [c_m/d_k]M :: (c_m : A)$. Applying the
$\m{msg}$ rule,
\[
\inferrule*[right = $\m{msg}$]
{\G_0 \potconf{E} \dc :: (\W \semi \D, \D_1) \\
\cdot \semi \cdot \semi \D_1 \entailpot{q} [c_m/d_k]M ::
(c_m : A)}
{\G_0 \potconf{E+q+w'+w} \dc, \msg{c_m}{w', [c_m/d_k]M} ::
(\G \semi \D, (c_m : A))}
\]

\item Case ($\m{fwd}^-$) : $\W = \dc, \proc{c_m}{w, \fwd{c_m}
{d_k}}, \msg{e_l}{w', M(c_m)}$ and $\W' = \msg{e_l}{w + w',
M(d_k)}$. First, inverting on the $\m{proc}_m$ rule

\[
\inferrule*[right = $\m{proc}_m$]
{\G_0 \potconf{E} \dc :: (\G \semi \D, \D_1, (d_k : A)) \\
\cdot \semi \G_0 \semi (d_k : A) \entailpot{0} \fwd{c_m}{d_k} :: (c_m : A)}
{\G_0 \potconf{E+0+w} \dc, \proc{c_m}{w, \fwd{c_m}{d_k}}
:: (\G \semi \D, \D_1, (c_m : A))}
\]
Call this new configuration $\dc'$. Next, inverting on the $\m{msg}$
rule,
\[
\inferrule*[right = $\m{msg}$]
{\G_0 \potconf{E+w} \dc' :: (\G \semi \D, \D_1, (c_m : A)) \\
\cdot \semi \cdot \semi \D_1, (c_m : A) \entailpot{q} M(c_m)
:: (e_l : C)}
{\G_0 \potconf{E+w+q+w'} \dc', \msg{e_l}{w', M(c_m)} :: (\G \semi
\D, (e_l : C))}
\]
Using Lemma~\ref{lem:renaming}, we get
$\cdot \semi \cdot \semi \D_1, (d_k : A) \entailpot{q} M(d_k) :: (e_l : C)$.
Reapplying the $\m{msg}$ rule,
\[
\inferrule*[right = $\m{msg}$]
{\G_0 \potconf{E} \dc :: (\G \semi \D, \D_1, (d_k : A)) \\
\cdot \semi \cdot \semi \D_1, (d_k : A) \entailpot{q} M(d_k)
:: (e_l : C)}
{\G_0 \potconf{E+q+w+w'} \dc, \msg{e_l}{w', M(d_k)} :: (\G \semi
\D, (e_l : C))}
\]

\item Case ($\oplus C_s$) : $\W = \dc, \proc{c_m}{w, \esendl{c_m}
{\ell} \semi P}$ and \\
$\W' = \dc, \proc{c^+_m}{w, [c^+_m/c_m]P},
\msg{c_m}{0, \esendl{c_m}{\ell} \semi \fwd{c_m}{c^+_m}}$.
First, inverting on the $\m{proc}_m$ rule,
\[
\inferrule*[right = $\m{proc}_m$]
{\G_0 \potconf{E} \dc :: (\G \semi \D, \D_1) \\
\inferrule*[right = $\oplus R$]
{\cdot \semi \G_0 \semi \D_1 \entailpot{q} P :: (c_m : A_{\ell})}
{\cdot \semi \G_0 \semi \D_1 \entailpot{q} \esendl{c_m}{\ell} \semi P 
:: (c_m : \ichoice{l : A_l}_{l \in L})}}
{\G_0 \potconf{E+q+w} \dc, \proc{c_m}{w, \esendl{c_m}{\ell}
\semi P} :: (\G \semi \D, (c_m : \ichoice{l : A_l}_{l \in L}))}
\]
Using Lemma~\ref{lem:renaming}, we get $\cdot \semi \G_0 \semi
\D_1 \entailpot{q} [c^+_m/c_m]P :: (c^+_m : A_{\ell})$.
Now, applying the $\m{proc}_m$ rule,
\[
\inferrule*[right = $\m{proc}_m$]
{\G_0 \potconf{E} \dc :: (\G \semi \D, \D_1) \\
\cdot \semi \G_0 \semi \D_1 \entailpot{q} [c^+_m/c_m]P ::
(c^+_m : A_{\ell})}
{\G_0 \potconf{E+q+w} \dc, \proc{c_m}{w, \esendl{c_m}{\ell}
\semi P} :: (\G \semi \D, (c^+_m : A_{\ell}))}
\]
Next, typing the message
\[
\cdot \semi \cdot \semi (c^+_m : A_{\ell}) \entailpot{0}
\esendl{c_m}{\ell} \semi \fwd{c_m}{c^+_m} :: (c_m : \ichoice{l :
A_l}_{l \in L})
\]
Call this new configuration $\dc'$. Applying the $\m{msg}$ rule next
\[
\inferrule*[right = $\m{msg}$]
{\G_0 \potconf{E+q+w} \dc' :: (\G \semi \D, (c_m : A_{\ell})) \\
\cdot \semi \cdot \semi (c^+_m : A_{\ell}) \entailpot{0}
\esendl{c_m}{\ell} \semi \fwd{c_m}{c^+_m} :: (c_m : \ichoice{l :
A_l}_{l \in L})}
{\G_0 \potconf{E+q+w} \dc', \msg{c_m}{0, \esendl{c_m}{\ell} \semi
\fwd{c_m}{c^+_m}} :: (\G \semi \D, (c_m : \ichoice{l : A_l}_{l \in
L}))}
\]

\item Case ($\oplus C_r$) : $\W = \dc, \msg{c_m}{w, \esendl{c_m}
{\ell} \semi \fwd{c_m}{c^+_m}}, \proc{d_k}{w', \ecase{c_m}{l}
{Q_l}_{l \in L}}$ and $\W' = \dc, \proc{d_k}{w+w',
[c^+_m/c_m]Q_{\ell}}$. First, inverting the $\m{msg}$ rule,
\[
\inferrule*[right = $\m{msg}$]
{\G_0 \potconf{E} \dc :: (\G \semi \D, \D_1, (c^+_m : A_{\ell})) \\
\cdot \semi \cdot \semi (c^+_m : A_{\ell}) \entailpot{0}
\esendl{c_m}{\ell} \semi \fwd{c_m}{c^+_m} :: (c_m : \ichoice{l :
A_l}_{l \in L})}
{\G_0 \potconf{E+0+w} \dc, \msg{c_m}{w, \esendl{c_m}{\ell} \semi
\fwd{c_m}{c^+_m}} :: (\G \semi \D, \D_1, (c_m : \ichoice{l :
A_l}_{l \in L}))}
\]
Call this new configuration $\dc'$. Next, inverting the $\m{proc}_m$
rule,
\[
\inferrule*[right = $\m{proc}_m$]
{\G_0 \potconf{E+0+w} \dc' :: (\G \semi \D, \D_1, (c_m : \ichoice{l :
A_l}_{l \in L})) \\
\inferrule*[right = $\oplus R$]
{\cdot \semi \G_0 \semi \D_1, (c_m : A_l) \entailpot{q}
Q_l :: (d_k : C)}
{\cdot \semi \G_0 \semi \D_1, (c_m : \ichoice{l : A_l}_{l \in L})
\entailpot{q} \ecase{c_m}{l}{Q_l}_{l \in L} :: (d_k : C)}}
{\G_0 \potconf{E+0+w+q+w'} \dc', \proc{d_k}{w', \ecase{c_m}{l}
{Q_l}_{l \in L}} :: (\G \semi \D, (d_k : C))}
\]
Renaming using Lemma~\ref{lem:renaming}, we get
$\cdot \semi \G_0 \semi \D_1, (c^+_m : A_{\ell}) \entailpot{q}
[c^+_m/c_m]Q_{\ell} :: (d_k : C)$. Next, we apply the
$\m{proc}_m$ rule
\[
\inferrule*[right = $\m{proc}_m$]
{\G_0 \potconf{E} \dc :: (\G \semi \D, \D_1, (c^+_m : A_{\ell})) \\
\cdot \semi \G_0 \semi \D_1, (c^+_m : A_{\ell}) \entailpot{q}
[c^+_m/c_m]Q_{\ell} :: (d_k : C)}
{\G_0 \potconf{E+q+w+w'} \dc', \proc{d_k}{w+w', [c^+_m/c_m]Q_{\ell}} ::
(\G \semi \D, (d_k : C))}
\]

\item Case ($\lolli_n C_s$) : $\W = \dc, \proc{d_k}{w, \esendch{c_m}
{e_n} \semi P}$ and \\
$\W' = \dc, \msg{c^+_m}{0, \esendch{c_m}
{e_n} \semi \fwd{c^+_m}{c_m}}, \proc{d_k}{w, [c^+_m/c_m]P}$.
First, we invert the $\m{proc}_m$ rule,
\[
\inferrule*[right = $\m{proc}_m$]
{\G_0 \potconf{E} \dc :: (\G \semi \D, \D_1, (e_\p : A),
(c_m : A \lolli B)) \\
\inferrule*[right = $\lolli L$]
{\cdot \semi \G \semi \D_1, (c_m : B) \entailpot{q} P :: (d_k : C)}
{\cdot \semi \G \semi \D_1, (e_\p : A), (c_m : A \lolli B)
\entailpot{q} \esendch{c_m}{e_\p} \semi P :: (d_k : C)}}
{\G_0 \potconf{E+q+w} \dc, \proc{d_k}{w, \esendch{c_m}
{e_\p} \semi P} :: (\G \semi \D, (d_k : C))}
\]
Using renaming (Lemma~\ref{lem:renaming}), we get
$\G \semi \D_1, (c^+_m : B) \entailpot{q} [c^+_m/c_m]P ::
(d_k : C)$. Next, we type the message
\[
\cdot \semi \G \semi (e_\p : A), (c_m : A \lolli B) \entailpot{0}
\esendch{c_m}{e_\p} \semi \fwd{c^+_m}{c_m} :: (c^+_m : B)
\]
Next, we apply the $\m{msg}$ rule,
\[
\inferrule*[right = $\m{msg}$]
{\G_0 \potconf{E} \dc :: (\G \semi \D, \D_1, (e_\p : A),
(c_m : A \lolli B)) \\
\cdot \semi \G \semi (e_\p : A), (c_m : A \lolli B) \entailpot{0}
\esendch{c_m}{e_\p} \semi \fwd{c^+_m}{c_m} :: (c^+_m : B)}
{\G_0 \potconf{E} \dc, \msg{c^+_m}{0, \esendch{c_m}
{e_\p} \semi \fwd{c^+_m}{c_m}} :: (\G \semi \D, \D_1,
(c^+_m : B))}
\]
Call this new configuration $\dc'$. Next, we apply the $\m{proc}_m$
rule
\[
\inferrule*[right = $\m{proc}_m$]
{\G_0 \potconf{E} \dc' :: (\G \semi \D, \D_1, (c^+_m : B)) \\
\cdot \semi \G \semi \D_1, (c^+_m : B) \entailpot{q}
[c^+_m/c_m]P :: (d_k : C)}
{\G_0 \potconf{E+q+w} \dc', \proc{d_k}{w, [c^+_m/c_m]P} ::
(\G \semi \D, (d_k : C))}
\]

\item Case ($\lolli C_r$) : $\W = \dc, \proc{c_m}{w', \erecvch{c_m}
{x_\p} \semi Q}, \msg{c^+_m}{w, \esendch{c_m}{e_\p} \semi
\fwd{c^+_m}{c_m}}$ and $\W' = \dc, \proc{c^+_m}{w+w',
[c^+_m/c_m][e_\p/x_\p]Q}$. First, inverting the $\m{proc}_m$
rule,
\[
\inferrule*[right = $\m{proc}_m$]
{\G_0 \potconf{E} \dc :: (\G \semi \D, \D_1, (e_\p : A)) \\
\inferrule*[right = $\lolli R$]
{\cdot \semi \G \semi \D_1, (x_\p : A) \entailpot{q} Q :: (c_m : B)}
{\cdot \semi \G \semi \D_1 \entailpot{q} \erecvch{c_m}{x_\p} \semi
Q :: (c_m : A \lolli B)}}
{\G_0 \potconf{E+q+w'} \dc, \proc{c_m}{w', \erecvch{c_m}{x_\p}
\semi Q} :: (\G \semi \D, (e_\p : A), (c_m : A \lolli B))}
\]
Call this new configuration $\dc'$. Next, we type the message.
\[
\cdot \semi \G \semi (e_\p : A), (c_m : A \lolli B) \entailpot{0}
\esendch{c_m}{e_\p} \semi \fwd{c^+_m}{c_m} :: (c^+_m : B)
\]
Inverting the $\m{msg}$ rule,
\[
\inferrule*[right = $\m{msg}$]
{\G_0 \potconf{E+q+w'} \dc' :: (\G \semi \D, (e_\p : A),
(c_m : A \lolli B)) \\
\cdot \semi \G \semi (e_\p : A), (c_m : A \lolli B) \entailpot{0}
\esendch{c_m}{e_\p} \semi \fwd{c^+_m}{c_m} :: (c^+_m : B)}
{\G_0 \potconf{E+q+w'+0+w'} \dc', \msg{c^+_m}{w, \esendch{c_m}
{e_\p} \semi \fwd{c^+_m}{c_m}} :: (\G \semi \D, (c^+_m : B))}
\]
By renaming using Lemma~\ref{lem:renaming},
$\cdot \semi \G \semi \D_1, (e_\p : A) \entailpot{q}
[c^+_m/c_m][e_\p/x_\p]Q :: (c^+_m : B)$.
Now, applying the $\m{proc}_m$ rule,
\[
\inferrule*[right = $\m{proc}_m$]
{\G_0 \potconf{E} \dc :: (\G \semi \D, \D_1, (e_\p : A)) \\
\cdot \semi \G \semi \D_1, (e_\p : A) \entailpot{q}
[c^+_m/c_m][e_\p/x_\p]Q :: (c^+_m : B)}
{\G_0 \potconf{E+q+w'} \dc, \proc{c^+_m}{w+w',
[c^+_m/c_m][e_\p/x_\p]Q} :: (\G \semi \D, (c^+_m : B))}
\]

\item Case ($\up C$) : $\W = \dc_1,
\proc{a_\s}{w', \eaccept{x_\l}{a_\s} \semi
P_{x_\l}}, \proc{c_m}{w, \eacquire{x_\l}{a_\s} \semi Q_{x_\l}}$ and
$\W' = \dc_1, \proc{a_\l}{w', P_{a_\l}}, \proc{c_m}{w, Q_{a_\l}}$.
Applying the $\m{proc}_\s$ rule first,

\[
\inferrule*[right = $\m{proc}_\s$]
{(a_\s : \up A_\l) \in \G_0, (a_\s : \up A_\l) \\
(\up A_\l, \up A_\l) \esync \quad
\G_0, (a_\s : \up A_\l) \potconf{E} \dc_1 :: (\G \semi \D, \D_1, \D_2) \quad
\mathcal{E}}
{\G_0, (a_\s : \up A_\l) \potconf{E+p+w'} \dc_1,
\proc{a_\s}{w', \eaccept{x_\l}{a_\s} \semi
P_{x_\l}} :: (\G, (a_\s : \up A_\l) \semi \D, \D_2)}
\]
where $\mathcal{E}$ is
\[
\inferrule*[right = $\up R$]
{\cdot \semi \G_0, (a_\s : \up A_\l) \semi \D_1 \entailpot{p}
P_{x_\l} :: (x_\l : A_\l)}
{\cdot \semi \G_0, (a_\s : \up A_\l) \semi \D_1 \entailpot{p}
\eaccept{x_\l}{a_\s} \semi
P_{x_\l} :: (a_\s : \up A_\l)}
\]
Call this new configuration $\dc_1'$.
Applying the $\m{proc}_m$ rule next,
\[
\inferrule*[right = $\m{proc}_m$]
{\G_0, (a_\s : \up A_\l) \potconf{E'} \dc_1' ::
(\G, (a_\s : \up A_\l) \semi \D, \D_2) \\
\inferrule*[right = $\up L$]
{\cdot \semi \G_0 \semi \D_2, (x_\l : A_\l) \entailpot{q} Q_{x_\l} :: (c_m : C)}
{\cdot \semi \G_0, (a_\s : \up A_\l) \semi \D_2 \entailpot{q} \eacquire{x_\l}{a_\s}
\semi Q_{x_\l} :: (c_m : C)}}
{\G_0 \potconf{E'+q+w} \dc_1', \proc{c_m}{w, \eacquire{x_\l}{a_\s}
\semi Q_{x_\l}} :: (\G, (a_\s : \up A_\l) \semi \D, (c_m : C))}
\]

From the first premise, we get by Lemma~\ref{lem:renaming},
$\cdot \semi \G_0, (a_\s : \up A_\l) \semi \D_1 \entailpot{p}
P_{a_\l} :: (a_\l : A_\l)$
while from the second premise, we get by Lemma~\ref{lem:renaming}
and Lemma~\ref{lem:proc-weakening}, $\cdot \semi \G_0, (a_\s : \up A_\l)
\semi \D_2, (a_\l : A_\l) \entailpot{q} Q_{a_\l} :: (c_m : C)$.
Reapplying the $\m{proc}_\l$ rule,

\[
\inferrule*[right = $\m{proc}_\l$]
{(a_\s : \up A_\l) \in \G_0, (a_\s : \up A_\l) \\
(A_\l, \up A_\l) \esync \\
\G_0, (a_\s : \up A_\l) \potconf{E} \dc_1 :: (\G \semi \D, \D_1, \D_2) \quad
\cdot \semi \G_0, (a_\s : \up A_\l) \semi \D_1 \entailpot{p}
P_{a_\l} :: (a_\l : A_\l)}
{\G_0, (a_\s : \up A_\l) \potconf{E+p+w'} \dc_1, \proc{a_\l}{w', P_{a_\l}} ::
(\G, (a_\s : \up A_\l) \semi \D, \D_2, (a_\l : A_\l))}
\]
Call this new configuration $\dc_1''$. Reapplying the $\m{proc}_m$ rule,
\[
\inferrule*[right = $\m{proc}_m$]
{\G_0, (a_\s : \up A_\l) \potconf{E'} \dc_1'' ::
(\G, (a_\s : \up A_\l) \semi \D, \D_2, (a_\l : A_\l)) \\
\cdot \semi \G_0, (a_\s : \up A_\l) \semi \D_2, (a_\l : A_\l) \entailpot{q}
Q_{a_\l} :: (c_m : C)}
{\G_0, (a_\s : \up A_\l) \potconf{E'+q+w} \dc_1'', \proc{c_m}{w,
Q_{a_\l}} :: (\G', (a_\s : \up A_\l) \semi \D', (c_m : C))}
\]

\item Case ($\down C$) : $\W = \dc_1, \proc{a_\l}{w', \edetach{x_\s}{a_\l}
\semi P_{x_\s}}, \proc{c_\c}{w, \erelease{x_\s}{a_\l} \semi Q_{x_\s}}$
and $\W' = \dc_1, \proc{a_\s}{w', P_{a_\s}},
\proc{c_\l}{w, Q_{a_\s}}$.
Applying the $\m{proc}_\l$ rule first,
\[
\inferrule*[right = $\m{proc}_\l$]
{(a_\s : A_\s) \in \G_0 \qquad
(\down A_\s, A_\s) \esync \qquad
\G_0 \potconf{E} \dc_1 :: (\G \semi \D, \D_1, \D_2) \qquad
\mathcal{E}}
{\G_0 \potconf{E+p+w'} \dc_1, \proc{a_\l}{w', \edetach{x_\s}{a_\l}
\semi P_{x_\s}} :: (\G, (a_\s : A_\s) \semi \D, \D_2, (a_\l : \down A_\s))}
\]
where $\mathcal{E}$ is
\[
\inferrule*[right = $\down R$]
{\cdot \semi \G_0 \semi \D_1 \entailpot{p} P_{x_\s} :: (x_\s : A_\s)}
{\cdot \semi \G_0 \semi \D_1 \entailpot{p} \edetach{x_\s}{a_\l}
\semi P_{x_\s} :: (a_\l : \down A_\s)}
\]
Call this configuration $\dc_1'$. Applying the $\m{proc}_m$ rule,
\[
\inferrule*[right = $\m{proc}_m$]
{
\inferrule*[right = $\down L$]
{\cdot \semi \G_0, (x_\s : A_\s) \semi \D_2 \entailpot{q}
Q_{x_\s} :: (c_m : C)}
{\cdot \semi \G_0 \semi \D_2, (a_\l : \down A_\s) \entailpot{q}
\erelease{x_\s}{a_\l} \semi Q_{x_\s} :: (c_m : C)} \\
\G_0 \potconf{E'} \dc_1' :: (\G, (a_\s : A_\s) \semi
\D, \D_2, (a_\l : \down A_\s))
}
{\G_0 \potconf{E'+q+w} \dc_1', \proc{c_\c}{w, \erelease{x_\s}{a_\l}
\semi Q_{x_\s}} :: (\G, (a_\s : A_\s) \semi \D, (c_m : C))}
\]
From the first premise, we get by Lemma~\ref{lem:renaming},
$\cdot \semi \G_0 \semi \D_1 \entailpot{p} P_{a_\s} :: (a_\s : A_\s)$.
From the second premise, by Lemma~\ref{lem:subst}
(contracting $a_\s : A_\s$ and $x_\s : A_\s$), we get
$\cdot \semi \G_0 \semi \D_2 \entailpot{q}
Q_{a_\s} :: (c_m : C)$. Finally, applying the $\m{proc}_\s$ rule,
\[
\inferrule*[right = $\m{proc}_\s$]
{(a_\s : A_\s) \in \G_0 \\
(A_\s, A_\s) \esync \\
\G_0 \potconf{E} \dc_1 :: (\G \semi \D, \D_1, \D_2) \\
\cdot \semi \G_0 \semi \D_1 \entailpot{p} P_{a_\s} :: (a_\s : A_\s)}
{\G_0 \potconf{E+p+w'} \dc_1, \proc{a_\s}{w', P_{a_\s}} :: (\G, (a_\s : A_\s)
\semi \D, \D_2)}
\]
Call this new configuration $\dc_1''$. Applying the $\m{proc}_m$ rule,
\[
\inferrule*[right = $\m{proc}_\c$]
{\G_0 \potconf{E'} \dc_1'' :: (\G, (a_\s : A_\s) \semi \D, \D_2) \\
\cdot \semi \G_0 \semi \D_2 \entailpot{q}
Q_{a_\s} :: (c_m : C)}
{\G_0 \potconf{E'+q+w} \dc_1'', \proc{c_m}{w, Q_{a_\s}} ::
(\G, (a_\s : A_\s) \semi \D, (c_m : C))}
\]

\end{itemize}

\end{proof}

\begin{definition}
A process $\proc{c_m}{w, P}$ is said to be poised if it is trying to
receive a message on $c_m$. A message $\msg{c_m}{w, M}$ is said
to be poised if it is trying to send a message along $c_m$.
A configuration $\W$ is said to be poised if all the
processes and messages in $\W$ are poised. Concretely,
the following processes are poised.
\begin{itemize}
\item $\proc{c_m}{w, \fwdn{c_m}{d_m}}$
\item $\proc{c_m}{w, \ecase{c_m}{l_i}{P_i}_{i \in I}}$
\item $\proc{c_m}{w, \erecvch{c_m}{x_\p} \semi P}$
\item $\proc{c_m}{w, \erecvch{c_m}{x} \semi P}$
\item $\proc{c_\s}{w, \eaccept{c_\l}{c_\s} \semi P}$
\item $\proc{c_\l}{w, \edetach{c_\s}{c_\l} \semi P}$
\item $\proc{c_m}{w, \eget{c_m}{r} \semi P}$
\end{itemize}
Similarly, the following messages are poised.
\begin{itemize}
\item $\msg{c_m}{w, \esendl{c_m}{l_k} \semi P}$
\item $\msg{c_m}{w, \esendch{c_m}{e_n} \semi P}$
\item $\msg{c_m}{w, \esendch{c_m}{N} \semi P}$
\item $\msg{c_m}{w, \eclose{c_m}}$
\item $\msg{c_m}{w, \epay{c_m}{r} \semi P}$
\end{itemize}
\end{definition}

\begin{theorem}[Process Progress]\label{thm:proc_prog}
Consider a closed well-formed and well-typed configuration $\W$
such that $\G_0 \potconf{E} \W :: (\G \semi \D)$. Either $\W$ is
poised, or it can take a step, i.e., $\W \step \W'$, or some
process in $\W$ is blocked along $a_\s$ for some shared channel
$a_\s$ and there is a process $\proc{a_\l}{w, P} \in \W$.
\end{theorem}

\begin{proof}
Either $\W = \W_1, \proc{c_m}{w, P}$ or $\W = \W_1, \msg{c_m}
{w, M}$. In either case, either $\W_1 \step \W_1'$, in which case
we are done. Or there is a process in $\W_1$ blocked along $a_\s$
in which case, we are also done. Hence, in the final case, we get
$\W_1$ is poised and there is no process in $\W_1$ blocked along
$a_\s$. Now, we case analyze on the structure of the process or
message. We start with processes.

\begin{itemize}
\item Case ($\{\}E_{mn}$) : In each case, the process spontaneously
steps by spawning another process.

\item Case ($\m{fwd}^+ : \proc{c_m}{w, \fwdp{c_m}{d_k}}$) :
\[
\cdot \semi \G \semi (d_k : A) \entailpot{0} \fwdp{c_m}{d_k} ::
(c_m : A)
\]
Since $\W_1$ is poised, there must be a message in $\W_1$ offering
along $d_m : A$. We use Lemma~\ref{lem:perm-msg} to move the
message just left of the process, and then apply the $\m{fwd}^+$
rule. Hence, $\W$ can step.

\item Case ($\m{fwd}^- : \proc{c_m}{w, \fwdn{c_m}{d_m}}$) :
This process is poised, hence $\W$ is poised.

\item Case ($\oplus R : \proc{c_m}{w, \esendl{c_m}{k}
\semi P}$) : $\W$ steps using $\oplus C_s$ rule.

\item Case ($\oplus L : \proc{d_k}{w, \ecase{c_m}{l}{Q_l}_{l \in
L}}$) :
\[
\cdot \semi \G \semi (c_m : \ichoice{l : A_l}_{l \in L})
\entailpot{q} \ecase{c_m}{l}{Q_l}_{l \in L} :: (d_k : C)
\]
Since $\W_1$ is poised, there must be a message in $\W_1$ offering
along $c_m : \ichoice{l : A_l}_{l \in L}$. We use Lemma
\ref{lem:perm-msg} to move the message just left of the process,
and then apply the $\oplus C_r$ rule. Hence, $\W$ can step.

\item Case ($\lolli R : \proc{c_m}{w, \erecvch{c_m}{x_n} \semi
P}$) : This process is poised, hence $\W$ is poised.

\item Case ($\lolli L : \proc{c_m}{w, \esendch{c_m}{e_n} \semi
Q}$) : $\W$ steps using $\lolli C_s$ rule.

\item Case ($\up R : \proc{c_\s}{\eaccept{c_\l}{c_\s} \semi P}$) :
This process is poised, hence $\W$ is poised.

\item Case ($\up L : \proc{c_m}{w, \eacquire{a_\l}{a_\s}
\semi Q}$) :
\[
\cdot \semi \G, (a_\s : \up A_\l) \semi \D \entailpot{q}
\eacquire{a_\l}{a_\s} \semi Q :: (c_m : C)
\]
There must be some process in $\W_1$ that offers on $a_\s$.
Either this process is in shared mode or linear mode. If the process
is in shared mode, and since $\W_1$ is poised, the process must
be $\proc{a_\s}{w', \eaccept{a_\l}{a_\s} \semi P}$ in which case,
we can use Lemma~\ref{lem:perm-acquire} to move the two processes
next to each other and
$\W$ can step using $\up C$ rule. Or the process is in linear mode
in which case the acquiring process is blocked and there is some
$\proc{a_\l}{w', P}$ in $\W$.

\item Case ($\down R : \proc{c_\s}{\edetach{c_\l}{c_\s} \semi P}$ :
This process is poised, hence $\W$ is poised.

\item Case ($\down L : \proc{c_\c}{w, \erelease{a_\l}{a_\s}
\semi Q}$) :
\[
\cdot \semi \G \semi \D, (a_\l : \down A_\s) \entailpot{q}
\erelease{a_\l}{a_\s} \semi Q :: (c_m : C)
\]
There must be some process in $\W_1$ that offers along $a_\l$.
Since $\W_1$ is poised, this process must be $\proc{a_\l}{w',
\edetach{a_\s}{a_\l} \semi P}$ in which case we use Lemma
\ref{lem:perm-release} to move the releasing process next to
the detaching process and $\W$ can step using
$\down C$ rule.

\end{itemize}

That completes the cases where the last predicate is a process. Now,
we consider the cases where the last predicate is a message.

\begin{itemize}

\item Case ($\m{fwd}^- : \msg{e_k}{w, M(c_m)}$) : There
must be some process in $\W_1$ that offers along $d_m$. Since
$\W_1$ is poised, if there is a forwarding process $\proc{c_m}
{w', \fwdn{c_m}{d_m}}$ in $\W_1$, then $\W$ steps using
$\m{fwd}^-$ rule. Hence, in the following cases, we assume that
the offering process used by the message will not be a forwarding
process.

\item Case ($\oplus : \msg{c_m}{\esendl{c_m}{k} \semi M}$) :
This message is poised, hence $\W$ is poised.

\item Case ($\lolli : \msg{c^+_m}{\esendch{c_m}{e_\p} \semi
\fwd{c^+_m}{c_m}}$) : There must be a process in $\W_1$
that offers along $c_m$. Since $\W_1$ is poised, this process must
be $\proc{c_m}{\erecvch{c_m}{x_n} \semi P}$. We move the
process to the left of this message using Lemma~\ref{lem:perm-proc}.
And then, $\W$ can step using $\lolli C_r$ rule.
\end{itemize}
\end{proof}

%% Bibliography
\bibliography{refs,lit,publications}

\end{document}